\documentclass[runningheads]{llncs}

\usepackage[dvips]{graphics}
\usepackage[utf8]{inputenc} 
\usepackage{amsmath}
\usepackage[T1]{fontenc}    
\usepackage{url}            
\usepackage{booktabs}       
\usepackage{amsfonts}       
\usepackage{nicefrac}       
\usepackage{lipsum}
\usepackage{graphicx}
\usepackage{textcomp} 
\usepackage{listings}
\usepackage{xcolor}
\usepackage{comment}
\usepackage{mathpartir} 
\usepackage[normalem]{ulem}
\usepackage{tikz}
\usepackage{algpseudocode}
\usepackage{hyperref}       

\usepackage[pass]{geometry}

\usepackage{setspace}
\usepackage{pdfpages}
\usepackage[whole]{bxcjkjatype}
\usepackage{stmaryrd}

\newcommand{\bigast}{\mathop{\scalebox{1.5}{\raisebox{-0.2ex}{$*$}}}}%
\newcommand{\bigplus}{\mathop{\scalebox{1.5}{\raisebox{-0.2ex}{$+$}}}}%

\algnewcommand\algorithmicforeach{\textbf{for each}}
\algdef{S}[FOR]{ForEach}[1]{\algorithmicforeach\ #1\ \algorithmicdo}

\begin{document}
\title{Relative Completeness of Incorrectness Separation Logic
}

\author{Yeonseok Lee\thanks{\email{lee.yeonseok.x2@s.mail.nagoya-u.ac.jp}}  \and 
Koji Nakazawa\thanks{\email{knak@i.nagoya-u.ac.jp}}}

\authorrunning{Y. Lee and K. Nakazawa}

\institute{Nagoya University, Nagoya, Japan}
\maketitle            

\begin{abstract}
Incorrectness Separation Logic (ISL) is a proof system that is tailored specifically to resolve problems of under-approximation in programs that manipulate heaps, and it primarily focuses on bug detection. 
This approach is different from the over-approximation methods that are used in traditional logics such as Hoare Logic or Separation Logic. 
Although the soundness of ISL has been established, its completeness remains unproven. 
In this study, we establish relative completeness by leveraging the expressiveness of the weakest postconditions; expressiveness is a factor that is critical to demonstrating relative completeness in Reverse Hoare Logic. 
In our ISL framework, we allow for infinite disjunctions in disjunctive normal forms, where each clause comprises finite symbolic heaps with existential quantifiers. 
To compute the weakest postconditions in ISL, we introduce a canonicalization that includes variable aliasing.
\end{abstract}

\section{Introduction}
\subsection{Background}
Software verification is a crucial aspect of software engineering, as it ensures that software systems are reliable, safe, and secure. 
To verify software, we rely on formal methods, which are essentially a set of guidelines and tools for checking software.

One of the most important tools is Hoare Logic (HL) \cite{hoare1969axiomatic} which checks \textit{correctness} of programs.
Hoare Triples in the form of $\{P\} \ \mathbb{C} \ \{Q\}$ mean that for all states $s$ in precondition $P$, if running $\mathbb{C}$ on $s$ terminates in $s'$, then $s'$ is in postcondition $Q$.
Let us say that $\textsf{post}(\mathbb{C}, P)$ describes the set of states obtained by executing $\mathbb{C}$ from a state in $P$.
Then the Hoare triple $\{P\} \ \mathbb{C} \ \{Q\}$ means that $Q$ \textit{over-approximates} $\textsf{post}(\mathbb{C}, P)$ that is, $\textsf{post}(\mathbb{C}, P) \subseteq Q$.

Recently, there have been some interesting logics for checking \textit{incorrectness}, meaning the presence of bugs, which are different from HL.
Reverse Hoare Logic (RHL) \cite{de2011reverse} is one of these logics.
The triples of RHL take the form $[P] \ \mathbb{C} \ [Q]$\footnote{Note that in some literature, the notation $[P] \ \mathbb{C} \ [Q]$ is used for the Hoare triple, representing total correctness.
This notation asserts that if $\mathbb{C}$ is executed in a state that satisfies $P$, it always terminates, and the final state satisfies $Q$.
However, in this paper, we use square brackets in triples to signify that the triple is an Incorrectness logic triple \cite{o2019incorrectness}.}, which means that for all states $s'$ in $Q$, $s'$ can be reached by running $\mathbb{C}$ on some $s$ in $P$.
Here, $Q$ \textit{under-approximates} $\textsf{post}(\mathbb{C}, P)$, meaning that $\textsf{post}(\mathbb{C}, P) \supseteq Q$. 
It is important to note that the direction of inclusion is \textit{reversed}.
Let us explore the differences between HL and RHL using some simple examples.
\begin{itemize}
    \item $\{\top\} \ x := y \ \{ \top\}$ is valid, but $[\top] \ x := y \ [\top]$ is not.\\
    In RHL, $x=1 \land y=2$ satisfying $\top$ (postcondition) is unreachable from states satisfying
    $\top$ (precondition) after executing $x := y$.

    \item $\{\top\} \ x := y \ \{ 0<x<10 \land x=y\}$ is invalid, but 
    $[\top] \ x := y \ [ 0<x<10 \land x=y]$ is not. \\
    In HL, after $x := y$ is executed, $x=10 \land y=10$ can be reached from $\top$. However, this state does not satisfy $0<x<10 \land x=y$.

\end{itemize}

Additionally, Incorrectness Logic (IL), which also uses under-approximation, has been proposed \cite{o2019incorrectness}.
Triples in IL are similar to those in RHL, but IL includes an exit condition within postconditions: $[P] \ \mathbb{C} \ [\epsilon:Q]$, where $\epsilon$ is either \textit{ok} for normal termination or \textit{er} for erroneous termination.

Separation Logic (SL), which was built upon HL and introduced by Reynolds \cite{reynolds2002separation}, offers a framework for the modular proofs of programs that manipulate pointers.
It achieves this by allowing us to reason about disjoint parts of a heap using the \textit{separating conjunction} operator denoted as ``$*$''. 
For example, $P * R$
asserts that the heap can be split into two disjoint parts where $P$ and $R$ respectively hold.
A crucial aspect of SL is its \textsc{(Frame)} rule, which supports local reasoning.
\textsc{(Frame)} rule states the following. 
\[
\inferrule[]
{\{P\} \ \mathbb{C} \ \{Q\}}
{\{P * R\} \ \mathbb{C} \ \{Q * R\}}
\]
This rule states that if the program $\mathbb{C}$ does not affect the portion of memory described by $R$, then $R$ can coexist with the conditions $P$ and $Q$.
This feature improves the scalability of verification, especially the verification of complex programs.
Numerous verification tools have been built with SL as the foundation, with one of the most renowned being Infer by Meta \cite{calcagno2011infer,calcagno2015open}.

%
%
Raad et al. \cite{raad2020local} proposed Incorrectness Separation Logic (ISL), which includes heap manipulation like SL in the IL system. 
They demonstrate that the $\textsc{Frame}$ rule is not sound in the IL setting without the inclusion of a \textit{negative heap assertion} $x \not\mapsto$, which indicates  that $x$ exists in the heap domain and has been deallocated.
The following is an example from \cite{raad2020local}, where $x \mapsto
-$ denotes that the location $x$ is allocated.
\[
\inferrule*
{[x \mapsto -  ] \ \texttt{free($x$)} \ [\textsf{emp}] }
{[x \mapsto - * x \mapsto - ] \ \texttt{free($x$)} \ [\textsf{emp} * x \mapsto -]}
\]
The premise is valid, but the conclusion is not. 
Since these are ISL triples, the conclusion means that every state satisfying the postcondition can be reached from some state satisfying the precondition \(x \mapsto - * x \mapsto -\), which is unsatisfiable. 
Hence, the inference is unsound.

However, this case can be corrected by adding $x \not\mapsto$, as was demonstrated in \cite{raad2020local}.
\[
\inferrule*
{[x \mapsto -  ] \ \texttt{free($x$)} \ [x \not\mapsto ] }
{[x \mapsto - * x \mapsto - ] \ \texttt{free($x$)} \ [x \not\mapsto  * x \mapsto -]}
\]
In this case, both triples are valid.
The precondition and postcondition of the conclusion are false, which makes the triple trivially valid.
The semantics of $x \not\mapsto$ are similar to $x \uparrow$ in \cite{nakazawa2018cyclic,nakazawa2020spatial}, but they differ in whether $x$ is in the heap domain ($x \not\mapsto$) or not ($x \uparrow$).


\subsection{Relative Completeness of Under-approximation Systems}
In the context of proof systems, it is crucial to establish both soundness and completeness. 
Soundness ensures that all provable statements are true, while completeness ensures that all true statements can be proven. 
Hoare logic is known to be incomplete. 
Consider the Hoare triple \(\{\top\} \ \texttt{skip} \ \{A\}\), where \(\texttt{skip}\) is a command that does nothing. 
This triple is logically equivalent to the assertion \(A\) because the triple implies \(\top \rightarrow A\), which is equivalent to \(A\). 
Hence, if the assertion theory includes arithmetic, a complete proof system for Hoare logic would contradict Gödel's incompleteness theorems \cite{godel1931formal}.
Similarly, we can demonstrate the incompleteness of ISL, since the ISL triple \([A] \ \texttt{skip} \ [\top]\) is equivalent to $A$.

Cook suggests a concept called
\textit{relative completeness} \cite{cook1978soundness}. 
Relative completeness of Hoare logic means that if we had an oracle capable of checking the validity of entailment, then Hoare logic could prove all valid Hoare triples.
Relative completeness has also been proven in IL \cite{o2019incorrectness} and RHL \cite{de2011reverse}. 
In IL, semantic predicates are utilized.
In RHL, calculating the weakest postcondition is crucial to proving relative completeness. 
De Vries et al. demonstrated that their weakest postcondition could be
expressed using formulas with infinite disjunctions \cite{de2011reverse}.
This property is referred to as \textit{expressiveness}.
They then used this result to prove relative completeness.
However, in ISL \cite{raad2020local}, only soundness has been proven.

\subsection{Our Contribution: Proving the Relative Completeness of ISL by Weakest Postcondition Calculus}
This study establishes the relative completeness of ISL. 
This endeavor aims to preserve the relative completeness of RHL while expanding
it to incorporate exit conditions and heap manipulation within ISL.
Unlike IL \cite{o2019incorrectness} and the original ISL \cite{raad2020local}, 
our proof does not rely on semantic predicates and allows for infinite disjunctions, similar to those used in RHL \cite{de2011reverse}.
It is known that entailment checking for symbolic heaps with restricted predicates, and without arithmetic, is decidable \cite{berdine2005decidable,peltier2023testing,tatsuta2015separation}. 
Consequently, we can develop an ISL proof system that is complete, rather than relatively complete, 
because our system employs symbolic heaps with restricted predicates and no arithmetic.

We follow the proof method of de Vries et al. \cite{de2011reverse} to demonstrate the expressiveness of the weakest postconditions (Proposition \ref{prop: expressiveness of WPO calculus}). 
To calculate the weakest postconditions, we introduce a process called \emph{canonicalization}, where we perform case analysis on variables and use the disjunctive normal form as our syntax. 
Even if our system includes arithmetic, canonicalization for finite formulas remains feasible because we are only dealing with a finite number of cases.
However, for simplicity, we omit arithmetic in this paper.

\paragraph{Case Analysis.}
The function $\textsf{wpo}(P, \mathbb{C}, \epsilon)$, which stands for Weakest POstcondition, calculates the weakest postcondition given the input precondition $P$, program $\mathbb{C}$, and exit condition $\epsilon$.
Consider the following example.
\[ \textsf{wpo}(y \mapsto e, \texttt{free($x$)}, \textit{ok} )  = \ ??? \]
In the above situation, we cannot determine the weakest postcondition because we do not know whether $x$ is aliasing with $y$ or not. 
Consequently, we perform a case analysis to determine whether $x$ is
aliasing with $y$ or not.
For the case $x=y$, we have
\[ \textsf{wpo}(x = y \land y \mapsto e, \texttt{free($x$)} , \textit{ok}) =
(x = y   \land   y \not\mapsto ) .   \]
This case illustrates that if $y$ aliases with $x$ and a location mapped
by $y$ is allocated, then the \texttt{free($x$)} attempts to deallocate a
heap mapped by the location pointed to by $x$. As a result, this program
safely terminates, freeing the heap denoted by $x$ which aliases with
$y$. On the other hand, if $x\neq y$, we have
\[ \textsf{wpo}(x \neq y \land y \mapsto e, \texttt{free($x$)} , \textit{er}) = 
( x \neq y \land y \mapsto e) .    \]

In this case, $\texttt{free}(x)$ terminates with an error because there is no heap denoted by $x$; that is, there is no heap to be freed.
To address these situation, we conduct a case analysis that considers all the possibilities of aliasing. Further details are discussed in Section \ref{sec: body-wpo}.

\paragraph{Disjunctive Normal Form (DNF).}

We write every assertion in \emph{disjunctive normal form} (DNF), i.e.\ as
(possibly infinite) disjunctions of \emph{symbolic heaps}.  
Symbolic heaps are the restricted fragment that has been widely analysed in the
literature~\cite{berdine2005decidable,berdine2005symbolic,tatsuta2009completeness}.
A key algebraic property of the separating conjunction is that it distributes
over disjunction~\cite{reynolds2002separation}:
\[
 (P_1 \lor P_2) * R \;\equiv\; (P_1 * R) \lor (P_2 * R).
\]
Using this equivalence, any formula that mixes $*$ and $\lor$ can be
systematically transformed into a DNF whose clauses are symbolic
heaps \cite{echenim2020bernays,nguyen2007automated}.  
Once we are in DNF, the
\emph{case analysis} outlined in the previous paragraph can be carried out
\emph{clause-by-clause}, and the resulting weakest postcondition is obtained by
simply disjoining the individual results.

Consider a singleton heap assertion \(y \mapsto e\) and an arbitrary
command \(\mathbb{C}\).  
The aliasing between \(x\) and \(y\) gives
rise to two mutually exclusive cases:
\[
  y \mapsto e
  \iff
  (x = y \land y \mapsto e)
  \;\;\lor\;\;
  (x \neq y \land y \mapsto e).
\]
Because \(\textsf{wpo}\) is computed disjunctively on DNF, we obtain
\[
  \textsf{wpo}(y \mapsto e, \mathbb{C}, \epsilon)
  \;=\;
  \textsf{wpo}(x = y \land y \mapsto e, \mathbb{C}, \epsilon)
  \;\lor\;
  \textsf{wpo}(x \neq y \land y \mapsto e, \mathbb{C}, \epsilon).
\]

The remainder of this paper is organized as follows: Section \ref{sec: body-ISL} introduces the syntax and semantics of ISL. 
Section \ref{sec: body-wpo} presents the calculation of $\textsf{wpo}$ with canonicalization.
We demonstrate the relative completeness of our ISL 
in Section \ref{sec: body-completeness}. 
Finally, in Section \ref{sec: Related Work}, we discuss related work, and we conclude in Section \ref{sec: Conclusions and Future work}.

\section{Incorrectness Separation Logic}\label{sec: body-ISL}
We define the ISL, which is slightly changed
from the original one \cite{raad2020local}.

\subsection{Assertions}

First, we introduce assertions in ISL, which are (infinite) disjunctions
of existentially quantified symbolic heaps.  The quantifier-free
symbolic heaps consist of equalities ($x\approx y$) and inequalities
($x\not\approx y$) between variables and heap predicates ($x\mapsto t$ or $x\not\mapsto$)
connected by the separating conjunction $*$. The operators $\lor$,
$\land$, and $\lnot$ are not permitted in symbolic heaps.

We assume countably infinite set of variables {\normalfont\textsc{Var}}, and we use
the metavariables $x,y,z,\ldots$ for variables.

\begin{definition}[Assertions of ISL]
An assertion $P$ is defined as follows, where $I$ is finite or
 countably infinite index set.
A \textit{term} $t$ is either a variable or the constant {\normalfont\texttt{null}}.
\begin{align*}
    P ::= \ & \bigvee_{i \in I} \exists \overrightarrow{x_i} . \psi_i
    & {\text{Top-level}}\\
    \psi ::= \ &  \psi *  \psi   & {\text{Quantifier-free Symbolic Heaps}}\\
    \ &  | \ {\normalfont\textsf{emp}}  
    \ | \ x \mapsto t  \ | \ x \not\mapsto 
    & {\text{Atomic Spatial Formulas}}\\
    &| \ t \approx t' \ | \ t \not\approx t' & {\text{Atomic Pure Formulas}}\\
    t ::= \ & x \ | \ {\normalfont\texttt{null}} & {\text{Terms}}
\end{align*}
Here, we employ vector notation to represent a sequence, for instance,
$\overrightarrow{x_i}$ denotes the sequence $x_{i1}, \ldots, x_{in}$.
The notation $\normalfont\textsf{fv}(P)$ represents the set of free
variables in $P$.  We consider only assertions $P$ such that
$\normalfont\textsf{fv}(P)$ is finite.  
For $I=\emptyset$,
$\bigvee_{i\in I}\exists\overrightarrow{x_i}.\psi_i$ is denoted by
$\normalfont\textsf{false}$. 
We ignore the order of disjuncts in assertions and atomic formulas in symbolic heaps. 
Specifically, we identify the formulas $\psi_1 * \psi_2$ and $\psi_2 * \psi_1$ as being equivalent.
Additionally, we refer to a formula connected by $*$ with only atomic pure formulas as a \emph{pure formula}.
\end{definition}

Our language draws primarily from the conventions of traditional SL
\cite{reynolds2002separation}, IL \cite{o2019incorrectness}, RHL
\cite{de2011reverse}, and ISL \cite{raad2020local}.

For the semantics of the assertions, we assume the set {\normalfont\textsc{Val}} of
values and the set {\normalfont\textsc{Loc}} of locations such that
${\normalfont\textsc{Loc}}\subseteq{\normalfont\textsc{Val}}$. 
We also assume
$\textit{null}\in{\normalfont\textsc{Val}} \setminus {\normalfont\textsc{Loc}}$ for the denotation of the constant
$\texttt{null}$.
Additionally, $\bot \not\in {\normalfont\textsc{Val}}$ is employed to keep track of deallocated locations. 
The concept of using $\bot$ to signify deallocated locations was introduced by Raad et al. \cite{raad2020local}, specifically for negative heaps, expressed as $x \not\mapsto$.
Please note that $\bot$ does not refer to the undefined case of partial functions, as in some literature.

\begin{definition}[States, Stores, and Heaps]
A \emph{state} $\sigma \in {\normalfont\textsc{State}}$ is a pair $(s, h)$, 
where $s \in {\normalfont\textsc{Store}}$ is a \emph{store} and 
$h \in {\normalfont\textsc{Heap}}$ is a \emph{heap}.

\paragraph{Store.}
A store is a total function from ${\normalfont\textsc{Var}}$ to ${\normalfont\textsc{Val}}$.  
We extend $s$ to terms by defining $s(\normalfont\texttt{null}) = \textit{null}$.  
The updated store $s[x \mapsto v]$ is defined as:
\[
s[x \mapsto v](y) =
\begin{cases}
v & \text{if } y = x \\
s(y) & \text{otherwise}
\end{cases}
\]

\paragraph{Heap.}
A heap is a finite partial function from ${\normalfont\textsc{Loc}}$ to ${\normalfont\textsc{Val}} \cup \{ \bot \}$.  
The update $h[l \mapsto v]$ is defined analogously to stores.

Two heaps $h_1$ and $h_2$ are \emph{disjoint} if 
\[
\normalfont\textsf{dom}(h_1) \cap \normalfont\textsf{dom}(h_2) = \emptyset.
\]
If so, their composition $h_1 \circ h_2$ is defined by:
\[
(h_1 \circ h_2)(l) =
\begin{cases}
h_1(l) & \text{if } l \in \normalfont\textsf{dom}(h_1) \\
h_2(l) & \text{if } l \in \normalfont\textsf{dom}(h_2) \\
\text{undefined} & \text{otherwise}
\end{cases}
\]
The composition is undefined if $h_1$ and $h_2$ are not disjoint.

\paragraph{Positive Domain \(\normalfont\textsf{dom}_+(h)\).}
We define \(\normalfont\textsf{dom}_+(h)\) as the set of locations \(l\) such that \(\normalfont h(l) \in {\normalfont\textsc{Val}}\).  
Formally:
\[
l \in \normalfont\textsf{dom}_+(h) \;\; \text{iff} \;\; \normalfont h(l) \in {\normalfont\textsc{Val}}.
\]

\end{definition}

The semantics of the assertions of ISL is defined as follows.

\begin{definition}[Assertion semantics of ISL]\normalfont 
 We define the relation $(s,h)\models P$ as 
\begin{align*}
    (s,h) \models \bigvee_{i \in I} \exists \overrightarrow{x_i} . \psi_i &\Leftrightarrow
   (s,h) \models 
     \exists \overrightarrow{x_i} . \psi_i \text{ for some $i \in I$}, \\
    (s,h) \models \exists \overrightarrow{x_i}  . \psi_i  &\Leftrightarrow
    \exists \overrightarrow{v_i}  \in {\normalfont\textsc{Val}} .  
    (s [\overrightarrow{x_i} \mapsto \overrightarrow{v_i}] ,h) \models \psi_i, \\
    (s,h) \models \psi_1 * \psi_2 &\Leftrightarrow 
    \exists h_1,h_2. h = h_1 \circ h_2 \text{\ and\ }
    (s,h_1) \models \psi_1 \text{\ and\ } (s,h_2) \models \psi_2, \\ 
    (s,h) \models \textsf{emp} &\Leftrightarrow 
    \textsf{dom}(h) = \emptyset, \\
    (s,h) \models x \mapsto t &\Leftrightarrow 
    \textsf{dom}(h) = \{ s(x) \} \text{\ and\ } h(s(x))= s(t), \\
    (s,h) \models x \not\mapsto &\Leftrightarrow 
    \textsf{dom}(h) = \{ s(x) \} \text{\ and\ } h(s(x))= \bot, \\
    (s,h) \models t \approx t' &\Leftrightarrow 
    s(t) = s(t') \text{\ and\ } \textsf{dom}(h) = \emptyset, \\
    (s,h) \models t \not\approx t' &\Leftrightarrow 
    s(t) \neq s(t') \text{\ and\ } \textsf{dom}(h) = \emptyset,
 \end{align*}
where $s[\overrightarrow{x_i} \mapsto \overrightarrow{v_i}]$ means
$s[x_{i1} \mapsto v_{i1}] \ldots [x_{in} \mapsto v_{in}]$.
For a pure formula \(B\), \(s \models B\) means that \((s, \emptyset) \models B\) for the empty heap \(\emptyset\).
\end{definition}

\subsection{Program Language}

We recall the program language for ISL in \cite{raad2020local}.

\begin{definition}[Program Language $\mathbb{C}\in \textsc{Comm}$]\normalfont
\begin{align*}
    \mathbb{C} ::= \ &\texttt{skip} \ | \ x:=t \ | \ x:= \texttt{*} 
     \ | \ \texttt{assume($B$)} \ | \ \texttt{local $x $ in $\mathbb{C}$} \\
    &| \
    \mathbb{C}_1 ; \mathbb{C}_2 \ | \ \mathbb{C}_1 + \mathbb{C}_2 \ | \ \mathbb{C}^\star⋆ \\
    &| \  x:= \texttt{alloc()} \ | \ \texttt{free($x$)} \ | \ x:= [y] \ | \ [x]:= t \ | \ \texttt{error} \\
\end{align*}
Here, $B$ is a pure formula.
\end{definition}

The program language includes standard constructs such as \texttt{skip},
assignment $x := t$, nondeterministic assignment $x := \texttt{*}$
(where \texttt{*} represents a nondeterministically selected value),
assume statements \texttt{{assume}($B$)}, local variable
declarations \texttt{local $x$ in $\mathbb{C}$}, sequential composition
$\mathbb{C}_1 ; \mathbb{C}_2$, nondeterministic choice $\mathbb{C}_1 +
\mathbb{C}_2$, and loops $\mathbb{C}^\star$.
Additionally, it includes error statements \texttt{error} and
instructions for heap manipulation.

Instructions for heap manipulation encompass operations such as allocation, deallocation, lookup, and mutation. 
For instance, the instruction $x := \texttt{alloc()}$ reserves a new unused location on the heap and assigns it to the variable $x$. 
Conversely, \texttt{free($x$)} deallocates the memory location referred to by $x$. 
To read the contents of a heap location, $x := [y]$ is used, where the value stored at the location specified by $y$ is retrieved and assigned to $x$. 
In contrast, the heap mutation $[x]:= t$ involves replacing the data stored at the location indicated by $x$ with the value of $t$.

Note that deterministic choices ({\normalfont\texttt{if}}) and loops ({\normalfont\texttt{while}}) can be encoded using nondeterministic choices ($+$) and {\normalfont\texttt{assume}} statements \cite{raad2020local,o2019incorrectness}.
\begin{align*}
\texttt{if $B$ then $\mathbb{C}_1$ else $\mathbb{C}_2$} & \overset{\text{def}}{=}  
(\texttt{assume($B$)} ; \mathbb{C}_1) + ( \texttt{assume(!$B$)} ; \mathbb{C}_2) \\
\texttt{while($B$) } \mathbb{C} &\overset{\text{def}}{=} (\texttt{assume($B$)};\mathbb{C})^\star ;\texttt{assume(!$B$)}\\
\texttt{assert($B$)} &\overset{\text{def}}{=} (\texttt{assume(!$B$)}; \texttt{error}) + \texttt{assume($B$)} \\
x:=\texttt{malloc()} &\overset{\text{def}}{=} x:=\texttt{alloc()} + x:=\texttt{null}
\end{align*}
Here, we define \texttt{assume(!$B$)} for a pure formula $B =  b_1 * \ldots * b_n $, where each $b_i$ is an atomic pure formula. The definition is as follows:
\[
\texttt{assume(!$B$)} = \texttt{assume}(\lnot b_1) + \ldots + \texttt{assume}(\lnot b_n),
\]
where $\lnot b_i$ is defined as
$\lnot (t_i \approx {t_i}') = t_i \not\approx {t_i}'$ and 
$\lnot (t_i \not\approx {t_i}') = t_i \approx {t_i}'$.

We provide the denotational semantics of our programming languages.

\begin{definition}[Denotational semantics of ISL]\label{def:
 Denotational semantics of ISL} We define $\llbracket  \mathbb{C}\rrbracket  _\epsilon$
 for a program $\mathbb{C}$ and an \textit{exit condition}
 $\epsilon\in\{\textit{ok},\textit{er}\}$ as a binary relation on
 {\normalfont\textsc{State}}. We use the composition of binary relations: for
 binary relations $R_1$ and $R_2$, $R_1;R_2$ is the relation
 $\{(a,c)\mid \exists b.(a,b)\in R_1\text{\ and\ }(b,c)\in R_2\}$.

{\scriptsize\normalfont \begin{align*}
\llbracket   \texttt{skip} \rrbracket  _\epsilon &=
\begin{cases}
\{ (\sigma,\sigma) \ | \ \sigma \in {\normalfont\textsc{State}} \} & \epsilon= \textit{ok} \\
\emptyset &\epsilon= \textit{er} 
\end{cases} \\
\llbracket   x := t \rrbracket  _\epsilon &=
\begin{cases}
\{ ( (s,h),  (s[x \mapsto s(t)], h)   )  \} & \epsilon= \textit{ok} \\
\emptyset &\epsilon= \textit{er}
\end{cases}   \\ 
\llbracket   x := \texttt{*} \rrbracket  _\epsilon &=
\begin{cases}
\{ ( (s,h),  (s[x \mapsto v], h)   )  \ | \ v \in {\normalfont\textsc{Val}} \} & \epsilon= \textit{ok} \\
\emptyset &\epsilon= \textit{er}
\end{cases}   \\ 
\llbracket   \texttt{assume($B$)} \rrbracket  _\epsilon &=
\begin{cases}
\{ (\sigma,\sigma) \ | \  \sigma = (s,h) \land s \models B \} & \epsilon= \textit{ok} \\
\emptyset &\epsilon= \textit{er}
\end{cases}   \\ 
\llbracket   \texttt{error} \rrbracket  _\epsilon &=
\begin{cases}
\emptyset  & \epsilon= \textit{ok} \\
\{ (\sigma,\sigma) \ | \ \sigma \in {\normalfont\textsc{State}} \}&\epsilon= \textit{er}
\end{cases}   \\ 
\llbracket   \mathbb{C}_1; \mathbb{C}_2 \rrbracket  _\epsilon &=
  \begin{cases}
   \llbracket   \mathbb{C}_1\rrbracket  _\textit{ok};\llbracket   \mathbb{C}_2\rrbracket  _\textit{ok}
   & \epsilon= \textit{ok} \\
   \llbracket   \mathbb{C}_1\rrbracket  _\textit{er} \cup
   \llbracket   \mathbb{C}_1\rrbracket  _\textit{ok};\llbracket   \mathbb{C}_2\rrbracket  _\textit{er} 
   & \epsilon= \textit{er}
  \end{cases} \\
\llbracket   \texttt{local $x $ in $\mathbb{C}$} \rrbracket  _\epsilon &=
\{  (  (s ,h),  (s' ,h')   ) \ | \ \exists v , v' \in {\normalfont\textsc{Val}} . 
( (s [x \mapsto v ] ,h)  ,  (s' [x \mapsto v'] ,h')  ) 
\in \llbracket  \mathbb{C}   \rrbracket  _\epsilon    \land s(x) = s'(x) \} \\
\llbracket   \mathbb{C}_1 + \mathbb{C}_2 \rrbracket  _\epsilon &= 
\llbracket   \mathbb{C}_1  \rrbracket  _\epsilon \cup \llbracket   \mathbb{C}_2  \rrbracket  _\epsilon \\
\llbracket   \mathbb{C}^\star \rrbracket  _\epsilon &=  
\bigcup_{m\in \mathbb{N}} \llbracket   \mathbb{C}^m \rrbracket  _\epsilon \\
\llbracket   x := \texttt{alloc()} \rrbracket  _\epsilon &=
\begin{cases}
 \{ (  (s,h), (s[x\mapsto l], h[l\mapsto v]) )  \ | \ 
 v \in {\normalfont\textsc{Val}} \land ( l \notin \textsf{dom} (h) \lor h(l) = \bot)
 \}
 & \epsilon = \textit{ok} \\
 \emptyset & \epsilon = \textit{er}
\end{cases} \\
\llbracket  \texttt{free($x$)} \rrbracket  _\epsilon &= 
\begin{cases}
\{ (  \sigma, (s, h[s(x) \mapsto \bot]) )  \ | \ 
\sigma = (s,h) \land s(x) \in \textsf{dom}_+(h) \} & \epsilon= \textit{ok} \\
\{ (  \sigma, \sigma )  \ | \ 
\sigma = (s,h) \land s(x) \notin \textsf{dom}_+(h)  \}
 & \epsilon= \textit{er} 
\end{cases}\\
\llbracket   x := [y] \rrbracket  _\epsilon &= 
\begin{cases}
\{ (  \sigma, (s[x \mapsto h(s(y))] , h) )  \ | \ 
\sigma = (s,h) \land s(y) \in \textsf{dom}_+(h) \} & \epsilon= \textit{ok} \\
\{ (  \sigma, \sigma )  \ | \ 
\sigma = (s,h) \land s(y) \notin \textsf{dom}_+(h)  \}
 & \epsilon= \textit{er}
\end{cases}\\
\llbracket   [x] := t \rrbracket  _\epsilon &= 
\begin{cases}
\{ (  \sigma, (s , h [s(x) \mapsto s(t)] ) )  \ | \ 
\sigma = (s,h) \land  s(x) \in \textsf{dom}_+(h)  \} & \epsilon= \textit{ok} \\
\{ (  \sigma, \sigma )  \ | \ 
\sigma = (s,h) \land s(x) \notin \textsf{dom}_+(h)  \}
 & \epsilon= \textit{er} 
\end{cases}
\end{align*}
 } 
 Here, $\normalfont\mathbb{C}^0 = \texttt{skip}$ and $\mathbb{C}^{m+1} =
\mathbb{C};\mathbb{C}^{m}$.
\end{definition}

In~\cite{raad2020local,lee2024relative}, the denotational semantics defines the \textit{er} case for \texttt{free($x$)} 
in two cases:
\begin{itemize}
  \item {Null-free}: \(s(x) = \textit{null}\);
  \item {Double-free}: \(h(s(x)) = \bot \) (the location has already been deallocated).
\end{itemize}

We extend the error cases as follows:
\[
  \texttt{free($x$)}\ \text{fails whenever} \quad s(x) \notin \textsf{dom}_{+}(h).
\]
That is, the command succeeds when the address is allocated, and it signals an error otherwise. 
The existing ISL does not capture attempts to free a \emph{never-allocated but non-null} address as erroneous. 
However, in our semantics, such a case is detected: although \(x \neq \texttt{null}\), the location \(s(x)\) is not in \(\textsf{dom}_{+}(h)\), so an error is  raised.

The same principle applies to other heap-manipulating commands such as \(x := [y]\) and \([x] := t\): each of these succeeds only if the accessed address \(s(y)\) or \(s(x)\) lies in \(\textsf{dom}_{+}(h)\); otherwise, an error occurs.

\subsection{ISL Triples and Proof Rules}

The ISL triples are of the form $[P]\ \mathbb{C}\ [\epsilon: Q]$, and
their validity is defined as follows.

\begin{definition}[Validity of ISL Triples]
\label{def: Validity of ISL Triples}
    $$\models [P] \ \mathbb{C} \ [\epsilon: Q] \overset{\text{def}}{\iff}
    \forall \sigma'\models Q  , \exists \sigma  \models P. 
    (\sigma,\sigma') \in \llbracket   \mathbb{C} \rrbracket  _\epsilon$$
\end{definition}

Before presenting our ISL proof rules, we provide some additional definitions.
$\textsf{mod}(\mathbb{C})$ denotes the set of free variables modified by a program $\mathbb{C}$.
\begin{definition}[The set of variables modified by $\mathbb{C}$, $\normalfont\textsf{mod}(\mathbb{C})$]
\label{def: formal definition of mod(C)}
    \normalfont\begin{align*}
    \textsf{mod}(\texttt{skip}) &= \emptyset &
    \textsf{mod}(x := t) &= \{ x \} \\
    \textsf{mod}(x := \texttt{*}) &= \{ x \} &
    \textsf{mod}(\texttt{assume($B$)}) &= \emptyset \\
    \textsf{mod}(\texttt{error}) &= \emptyset &
    \textsf{mod}(x := \texttt{alloc()}) &= \{ x \} \\
    \textsf{mod}(\texttt{free($x$)}) &= \emptyset &
    \textsf{mod}(x := [y]) &= \{ x \} \\
    \textsf{mod}([x] := t) &= \emptyset    & \textsf{mod}(\texttt{local $x$ in $\mathbb{C}$}) &= \textsf{mod}(\mathbb{C}) \setminus \{ x \}
    \end{align*}
    \begin{align*}
    \textsf{mod}(\mathbb{C}_1 ; \mathbb{C}_2) &= 
    \textsf{mod}(\mathbb{C}_1) \cup \textsf{mod}(\mathbb{C}_2) \\
    \textsf{mod}(\mathbb{C}_1 + \mathbb{C}_2) &= 
    \textsf{mod}(\mathbb{C}_1) \cup \textsf{mod}(\mathbb{C}_2) \\
    \textsf{mod}(\mathbb{C}^\star) &=     \textsf{mod}(\mathbb{C}) \\
    \end{align*}
\end{definition}

The canonical forms of quantifier-free symbolic heaps with respect to a set
of variables $V$ are defined. In a canonical form, the case analysis for
all combinations of terms in $V\cup\{\texttt{null}\}$ is done.
As we will see in the next section, any symbolic heap can be transformed
to a disjunction of canonical forms.

\begin{definition}[Canonical forms]
    \begin{align*}
{\normalfont\textsf{CF}_\text{sh}} (V) =  \{ \psi & \mid  \text{$\psi$ is a quantifier free symbolic heap} \\
& \land 
\forall t,u \in V \cup \{{\normalfont\texttt{null}}\} . 
\text{$\psi$ contains either $t \approx u$ or $t \not\approx u$}
 \}        
    \end{align*}
If $\normalfont\psi \in {\normalfont\textsf{CF}_\text{sh}} (\textsf{fv}(\psi)\cup\textsf{fv}(\mathbb{C})) $, 
then $\psi$ is called \textit{canonical} for $\mathbb{C}$.
We also say that $\bigvee_{i \in I} \exists \overrightarrow{x_i} . \psi_i$ is canonical for $\mathbb{C}$ if, for all $i$, $\psi_i$ is canonical for $\mathbb{C}$.
\end{definition}

Regarding certain heap manipulation rules, such as \textsc{Free}, \textsc{Load}, and \textsc{Store}, 
we require that symbolic heaps be in canonical form as a premise.  
To formalize these rules, we introduce the following notation and concepts.

\paragraph{Atomic Membership and Aliases.}
Let \(\psi\) be a canonical quantifier-free symbolic heap of the form:
\[
  \psi = a_1 * a_2 * \cdots * a_k \quad (k \geq 0),
\]
where each \(a_i\) is an atomic formula. We distinguish between two classes of atomic formulas: spatial atoms \((\mathsf{emp},\ v \mapsto t,\ v \not\mapsto)\) and pure atoms \((t \approx u,\ t \not\approx u)\).

\noindent 
We write $\alpha \in \psi$ if and only if $\alpha$ occurs syntactically as one of the atomic formulas $a_1, \ldots, a_k$ in $\psi$.  
Likewise, $\alpha \notin \psi$ denotes the negation $\lnot(\alpha \in \psi)$.

The set of variables aliasing with $x$ in $\psi$ is defined as:
\[
  \mathsf{Aliases}(x, \psi) = \{\, v \in \textsc{Var} \mid  x \approx v \in \psi  \,\}
\]

\noindent
That is, $\mathsf{Aliases}(x, \psi)$ contains all variables that are provably equal to $x$ under the pure part of $\psi$.

\begin{definition}[ISL Proof Rules]
In Figure \ref{fig:generic proof rules of ISL} and \ref{fig: rules for
 variables and mutation}, we present our ISL proof rules, where the
 axioms of the form $[P]\ \mathbb{C}\
 [\textit{ok}:Q_1][\textit{er}:Q_2]$ means that both $[P]\ \mathbb{C}\
 [\textit{ok}:Q_1]$ and $[P]\ \mathbb{C}\ [\textit{er}:Q_2]$ can be
 deduced.
\end{definition}

Note that we adopt an infinitary syntax and incorporate an infinitary version of the \textsc{Disj} rule, following the approach taken in Reverse Hoare Logic (RHL) \cite{de2011reverse}.

\paragraph{Why We Restrict $\normalfont\textsc{Frame}$ to the `{ok}' Case.}
Regarding the \textsc{Frame} rule, we consider only the \emph{ok} case, which differs from the formulation in the existing ISL~\cite{raad2020local,lee2024relative}.
The reason for this restriction is that including the \emph{er} case in the \textsc{Frame} rule leads to a violation of soundness in our semantics. 

For example,  consider the following rule instance of \textsc{Frame} for \textit{er}:
\[
\inferrule[]
{[\textsf{emp}] \ \texttt{free($x$)} \  
[\textit{er}:  \textsf{emp}]}
{[\textsf{emp} * x \mapsto 1] \ \texttt{free($x$)} \  
[\textit{er}:  \textsf{emp} * x \mapsto 1 ]  }
\]

\noindent
Here, applying the \textsc{Frame} rule with $x \mapsto 1$ as the frame yields an invalid triple: although the upper (premise) triple is valid, the lower (conclusion) triple is not.
This breakdown illustrates that extending the \textsc{Frame} rule to the \emph{er} case compromises soundness in our semantics. Therefore, we restrict \textsc{Frame} to apply only to the \emph{ok} case.


\begin{figure}
\scriptsize
\begin{mathpar}
\inferrule[\textsc{Skip}]
{ \ }
{ [P] \ \texttt{skip} \ [\textit{ok}: P] \
[\textit{er}:\textsf{false}] }
\and

\inferrule[\textsc{Error}]
{ \ }
{ [ P ] \ \texttt{error()} \ [\textit{ok}:\textsf{false}] \
[\textit{er}: P]}
\and

\inferrule[\textsc{Seq1}]
{ [P] \ \mathbb{C}_1 \ 
[\textit{er}: Q] }
{ [P] \ \mathbb{C}_1;\mathbb{C}_2 \ 
[\textit{er}: Q] }
\and

\inferrule[\textsc{Seq2}]
{ [P] \ \mathbb{C}_1 \ [\textit{ok}:R] \ \ \ 
 [R] \ \mathbb{C}_2 \ [\epsilon:Q]
}
{ [P] \ \mathbb{C}_1;\mathbb{C}_2 \ 
[ \epsilon : Q] }
\and

\inferrule[\textsc{Loop zero}]
{ \ }
{ [P] \ \mathbb{C}^\star \ [\textit{ok}: P] \ [\textit{er}:\textsf{false}] }
\and

\inferrule[\textsc{Loop non-zero}]
{  [P] \ \mathbb{C}^\star;\mathbb{C} \ [\epsilon: Q]  }
{  [ P ] \ \mathbb{C}^\star \ [\epsilon: Q] }
\and

\inferrule[\textsc{Cons}]
{P' \models P \ \ \ 
 [P'] \ \mathbb{C} \ [\epsilon: Q']
\ \ \ Q \models Q'}
{ [P] \ \mathbb{C} \ [\epsilon: Q] }
\and

\inferrule[\textsc{Disj}]
{ [  P_i ] \ \mathbb{C} \
    [\epsilon:  Q_i  ]  \quad \text{for all $i \in I$}}
{ [   \bigvee_{i \in I} P_i ] \ \mathbb{C} \ 
[ \epsilon :  \bigvee_{i \in I} Q_i] }
\and

\inferrule[\textsc{Choice}]
{ [P] \ \mathbb{C}_1 \ [ \epsilon : Q] \ \ \ 
 [P] \ \mathbb{C}_2 \ [\epsilon:Q]
}
{ [P] \ \mathbb{C}_1 + \mathbb{C}_2 \ 
[ \epsilon : Q] }
\and

\inferrule[\textsc{Exist}]
{ [  \psi ] \ \mathbb{C}  \ 
[ \epsilon :   \varphi  ]  \ \ \  x \notin \textsf{fv}(\mathbb{C})  }
{ [  \exists  x  . \psi  ] \ \mathbb{C}  \ 
[ \epsilon :  \exists  x  . \varphi  ] }

\and

\inferrule[\textsc{Assign}]
{\ }
{ [ \psi ] \ x:=t \ [\textit{ok}:
\exists  x' . \psi [x := x'] * x \approx  t  [x := x']   ] \
[\textit{er}:\textsf{false}] }
\and

\inferrule[\textsc{Havoc}]
{ \ }
{ [ \psi ] \ x:=\texttt{*} \ [\textit{ok}: \exists x'. \psi [x := x'] ] \
[\textit{er}:\textsf{false}] }
\and

\inferrule[\textsc{Assume} \normalfont\text{($B$ is a pure formula)}]
{ \ }
{ [\psi] \ \texttt{assume($B$)} \ [\textit{ok}: \psi * B] \
[\textit{er}:\textsf{false}] }
\and

\inferrule*[lab=\textsc{Local}]
{ [ \psi   ] \ \mathbb{C}  \ 
[ \epsilon :  \varphi]   \ \ \ 
{x \notin \textsf{fv}(\psi)}
}
{ [  \psi  ] 
\ \texttt{local $x $ in $\mathbb{C}$}  \ 
[ \epsilon : \exists  x. \varphi] }

\and

\inferrule[\textsc{FrameOk}]
{ [ \psi] \ \mathbb{C} \ [ \textit{ok} :  \varphi] \ \ \ 
\textsf{mod}(\mathbb{C}) \cap
\textsf{fv}( \phi) = \emptyset }
{ [ \psi * \phi] \ \mathbb{C} \ 
[ \textit{ok} : \varphi * \phi  ] }

\end{mathpar}

\caption{Proof Rules of ISL  }
    \label{fig:generic proof rules of ISL}
\end{figure}

\begin{figure}
\begin{mathpar}

\inferrule[\textsc{Alloc1}]
{ \ }
{ [\psi] \ x:=\texttt{alloc()} \ 
\left[ \textit{ok}:  
\exists x' , v .   \psi[x := x'] * x \mapsto v  \right] \ 
[\textit{er}:\textsf{false}]}
\and

\inferrule[\textsc{Alloc2}]
{ \ }
{ [   \psi * y \not\mapsto] \ x:=\texttt{alloc()} \ 
[\textit{ok}:  
\exists x' , v .  x \mapsto v *  x \approx y * \psi [x:=x'] ] \ 
[\textit{er}:\textsf{false}]}
\and

\inferrule[\textsc{Free}]
{ \psi \in \textsf{CF}_\text{sh} (\textsf{fv}(\psi) \cup \{ x \}) \\
y \in \textsf{Aliases}(x, \psi) \\ 
\psi = \psi' *  y \mapsto t 
}
{ [ \psi ] \ \texttt{free($x$)} \ 
[\textit{ok}:  \psi' *  y \not\mapsto  ] \ 
[\textit{er}:\textsf{false}]
}
\and

\inferrule[\textsc{FreeEr}]
{
\psi \in \textsf{CF}_{\text{sh}}(\textsf{fv}(\psi) \cup \{x\}) \\
\forall v \in \textsf{Aliases}(x, \psi). \ (v \mapsto t) \notin \psi 
}
{
[\psi] \ \texttt{free($x$)} \ 
[\textit{ok}: \textsf{false}] \ 
[\textit{er}: \psi]
}

\and

\inferrule[\textsc{Load}]
{ \psi \in \textsf{CF}_\text{sh} (\textsf{fv}(\psi) \cup \{ x,y \}) \\
z \in \textsf{Aliases}(y, \psi) \\ 
\psi = \psi' *  z \mapsto t}
{ [\psi ] \ x := [y] \ 
[\textit{ok}: \exists x'. \psi [x := x'] * x \approx (t [x := x'])   ]
\ 
[\textit{er}:\textsf{false}]}
\and

\inferrule[\textsc{LoadEr}]
{ \psi \in \textsf{CF}_\text{sh} (\textsf{fv}(\psi) \cup \{ x,y \}) 
\\
\forall v \in \textsf{Aliases}(y, \psi). \ (v \mapsto t) \notin \psi
}
{ [ \psi] \ x := [y] \  
[\textit{ok}:\textsf{false}] \
[\textit{er}: \psi ]}
\and

\inferrule[{\normalfont\textsc{Store}}]
{ \psi \in \textsf{CF}_\text{sh} (\textsf{fv}(\psi) \cup \{ x,t \}) \\
z \in \textsf{Aliases}(x, \psi) \\ 
\psi = \psi' *  z \mapsto t' }
{ [ \psi ] \  [x] := t \ 
[\textit{ok}: \psi' *  z \mapsto t  ] \ 
[\textit{er}:\textsf{false}]}
\and

\inferrule[\textsc{StoreEr}]
{ \psi \in \textsf{CF}_\text{sh} (\textsf{fv}(\psi) \cup \{ x,t \}) 
\\
\forall v \in \textsf{Aliases}(x, \psi). \ (v \mapsto t') \notin \psi
}
{ [ \psi ] \  [x] := t \ 
[\textit{ok}:\textsf{false}] \
[\textit{er}: \psi ]}

\end{mathpar}

\caption{Rules for Heap Manipulation   }
    \label{fig: rules for variables and mutation}
\end{figure}

\section{Weakest Postcondition}\label{sec: body-wpo}
In this section, we describe the definition of the weakest postcondition
and introduce a function $\textsf{wpo}(P,\mathbb{C},\epsilon)$ that computes 
formulas representing weakest postconditions.
Prior to defining $\textsf{wpo}$, we introduce the canonicalization,
which transforms a formula into a canonical formula.

\subsection{Canonicalization}\label{subsec: Canonicalization}
In canonicalization, the function \(\textsf{CA}\) performs case analysis on all possible equalities and inequalities within assertions.
For a quantifier-free
symbolic heap $\psi$ and a program $\mathbb{C}$,
$\textsf{CA}(\psi,\mathbb{C})$ returns the result of case analysis for
$\psi$ with respect to the free variables in $\psi$ and $\mathbb{C}$.

 \begin{definition}[Case Analysis $\normalfont{\textsf{CA}}$]\label{def: Case Analysis for all possible aliasing and non-aliasing}
The function $\normalfont\textsf{CA}$, taking a symbolic heap $\psi$ and
  a command $\mathbb{C}$ as inputs, is defined as follows.
\[
  {\normalfont\textsf{CA}}(\psi, \mathbb{C}) =
  \{\pi * \psi \mid
  \pi \in \Pi ({\normalfont \textsf{fv}(\psi) \cup \textsf{fv}(\mathbb{C})}) \land \text{$(\pi * \psi)$ is satisfiable}
  \},
\]
where $\normalfont\Pi$ is defined as
\begin{align*}
\normalfont\Pi (V) = \{ (\bigast_{(t,s)\in S} t \approx s) * (\bigast_{(t,s)\not\in S} t \not\approx s) \ | \  &  S \subseteq ( V \cup\{{\normalfont\texttt{null}}\})^2
 \\
&\land S \text{\ is reflexive and symmetric}
\}
\end{align*}
for a set $V$ of variables.
\end{definition}

To illustrate the canonicalization function, we compute
\[
  \mathsf{CA}\bigl(y \mapsto t,\; \texttt{free}(x)\bigr).
\]
Omitting the tautological equalities
$x \approx x$ and $y \approx y$ for brevity, 
the result is
\begin{align*}
\mathsf{CA}(y \mapsto t,\;\texttt{free}(x)) = \Bigl\{\;
& x \approx y * y \not\approx \texttt{null} * x \not\approx \texttt{null}
   * y \mapsto t,                       &&\text{(1)} \\[2pt]
& x \not\approx y * y \not\approx \texttt{null} * x \approx \texttt{null}
   * y \mapsto t,                       &&\text{(2)} \\[2pt]
& x \not\approx y * y \not\approx \texttt{null} * x \not\approx \texttt{null}
   * y \mapsto t                        &&\text{(3)} \;\Bigr\},
\end{align*}

\noindent
where unsatisfiable combinations (e.g.\
$
x \approx y * y \approx \texttt{null}
* x \approx \texttt{null} * y \mapsto t
$) are discarded.

We define the canonicalization function.

\begin{definition}[Canonicalization $\normalfont\textsf{cano}$]
Suppose that $P = \bigvee_{i \in I} \exists \overrightarrow{x_i} . \psi_i$ and 
 $\normalfont  \overrightarrow{x_i} \notin \textsf{fv}(\mathbb{C})$ 
 for each $i \in I$. We define
$$\normalfont
 \textsf{cano}(P,\mathbb{C}) =  
 \bigvee_{i \in I}
 \bigvee_{\varphi \in \textsf{CA}(\psi_i,\mathbb{C})}
 \exists \overrightarrow{x_i} .\varphi$$
\end{definition}

The result of this function is canonical and equivalent to the input assertion.

\begin{lemma}\label{lma:equivalence of normalization}
For all $P$ and $\mathbb{C}$, $\normalfont\textsf{cano}(P,\mathbb{C})$
 is canonical for $\mathbb{C}$ and equivalent to $P$.
\end{lemma}

\begin{proof}
It is easy to see that every element in $\textsf{CA}(\psi,\mathbb{C})$
 is in $\textsf{CF}_{\rm sh}(\textsf{fv}(\psi)\cup \textsf{fv}(\mathbb{C}))$, and
 hence $\normalfont\textsf{cano}(P,\mathbb{C})$ is canonical for $\mathbb{C}$.

 Since every symbolic heap $\psi$ is equivalent to $(\psi * t\approx u)\lor
 (\psi * t\not\approx u)$ for any terms $t,u$, we have that $\psi$ is
 equivalent to $\bigvee_{\varphi\in\textsf{CA}(\psi,\mathbb{C})}\varphi$.
 Therefore, $\exists \overrightarrow{x_i}.\psi_i$ is equivalent to
 $\exists\overrightarrow{x_i}.\bigvee_{\varphi\in\textsf{CA}(\psi_i,\mathbb{C})}\varphi$,
 and also to $\bigvee_{\varphi\in\textsf{CA}(\psi_i,\mathbb{C})}\exists\overrightarrow{x_i}.\varphi$.
 Hence, $P$ is equivalent to $\normalfont\textsf{cano}(P,\mathbb{C})$.
\end{proof}

\subsection{Weakest Postcondition and $\normalfont\textsf{wpo}$}
We first define the weakest postcondition as follows. 
\begin{definition}[Weakest Postcondition]
The \emph{weakest postcondition} of a precondition $P$, a program $\mathbb{C}$,
 and an exit condition $\epsilon$ is defined as follows.
$$
\text{\normalfont{WPO}}\llbracket P,\mathbb{C},\epsilon\rrbracket   = \{ \sigma' \mid \exists \sigma. 
\sigma \models P \land
(\sigma,\sigma') \in \llbracket \mathbb{C}\rrbracket  _\epsilon \}
$$
\end{definition}

Next, we define the calculation of the weakest postcondition ($\textsf{wpo}$) as follows.

 \begin{definition}[$\normalfont\textsf{wpo}$]
\label{def: def of wpo}
$\normalfont \textsf{wpo}( P , \mathbb{C}, \epsilon)$ for an assertion
$P$ and $\normalfont \textsf{wpo}_\text{sh}( \psi , \mathbb{C}, \epsilon)$ for a
quantifier-free canonical symbolic heap $\psi$ are simultaneously defined as follows.
 $$\normalfont
\textsf{wpo}( P , \mathbb{C}, \epsilon) =  
\bigvee_{i \in I} \normalfont \exists \overrightarrow{x_i} . \textsf{wpo}_\text{sh} ( \psi_i,\mathbb{C}, \epsilon)  \ \  
 \normalfont\text{ \textit{(for} $\normalfont \textsf{cano} (P, \mathbb{C}) = \bigvee_{i \in I} \exists \overrightarrow{x_i} . \psi_i$\textit{)}}$$
The function $\normalfont\textsf{wpo}_\text{sh}$ is defined in Figure
  \ref{fig: wpo for clauses}.
These are inductively defined in $\mathbb{C}$.
 \end{definition}

Note that infinite disjunctions are needed only when we calculate $\textsf{wpo}$ of $\mathbb{C}^\star$. 
Therefore, if we assume a bound of the number of iterations of $\mathbb{C}^\star$ with a finite number, i.e., we use $\mathbb{C}^n$ ($n \in \mathbb{N}$) instead of $\mathbb{C}^\star$, 
then we can define $\textsf{wpo}$ with finite formulas.

\begin{figure}
\normalfont
\scriptsize
\noindent
\resizebox{1.0\textwidth}{!}{%
  \begin{minipage}{\linewidth}
\begin{align*}
\hline \\
\textsf{wpo}_\text{sh}  (\psi ,\texttt{skip},\epsilon) &= 
\begin{cases}
\psi & (\epsilon= \textit{ok}) \\
\textsf{false} & (\epsilon= \textit{er})
\end{cases}   \\
\textsf{wpo}_\text{sh}  ( \psi ,\texttt{error},\epsilon) &= 
\begin{cases}
\textsf{false}  &  (\epsilon= \textit{ok} ) \\
\psi & (\epsilon= \textit{er})
\end{cases} \\
\textsf{wpo}_\text{sh} ( \psi , \texttt{local $x $ in $\mathbb{C}$} ,\epsilon) &=
\bigvee_{j \in J} \exists   x'' , \overrightarrow{x_j}  . \varphi_j  \\
& \text{\Big(where $\textsf{wpo} (  \psi [x:=x']  ,  \mathbb{C}  ,\epsilon)
[x:=x''] [x' := x] = \bigvee_{j \in J} \exists \overrightarrow{x_j} . \varphi_j $
for fresh $x', x''$\Big) 
}\\
\textsf{wpo}_\text{sh} (\psi,\texttt{assume($B$)},\epsilon) &=
\begin{cases}
\psi * B  & (\epsilon= \textit{ok}) \\
\textsf{false} &( \epsilon= \textit{er})
\end{cases} \\
\textsf{wpo}_\text{sh} (\psi,\mathbb{C}_1;\mathbb{C}_2,\textit{ok}) &= 
\textsf{wpo}
(\textsf{wpo}_\text{sh} (\psi ,\mathbb{C}_1, \textit{ok}), \mathbb{C}_2,\textit{ok}) 
\\
\textsf{wpo}_\text{sh} (\psi,\mathbb{C}_1;\mathbb{C}_2,\textit{er}) &= 
 \textsf{wpo}_\text{sh}(\psi,\mathbb{C}_1, \textit{er})  \lor
\textsf{wpo}
(\textsf{wpo}_\text{sh} (\psi,\mathbb{C}_1, \textit{ok}), \mathbb{C}_2,\textit{er})
\\
\textsf{wpo}_\text{sh} (\psi,\mathbb{C}^\star,\textit{ok}) &= 
\bigvee_{n \in \mathbb{N}} \Upsilon (n) 
\text{ \Big(where $\Upsilon$ is defined by $\Upsilon (0) =\psi$  and 
$\Upsilon (n+1) = \textsf{wpo} (\Upsilon (n),\mathbb{C},\textit{ok})$\Big)}\\
\textsf{wpo}_\text{sh} (\psi,\mathbb{C}^\star,\textit{er}) &= 
\bigvee_{n \in \mathbb{N}} \textsf{wpo} (\Upsilon (n), \mathbb{C}, \textit{er}) \\
\textsf{wpo}_\text{sh} (\psi ,\mathbb{C}_1 + \mathbb{C}_2,\epsilon) &= 
\textsf{wpo}_\text{sh}(\psi ,\mathbb{C}_1,\epsilon) \lor \textsf{wpo}_\text{sh}(\psi,\mathbb{C}_2,\epsilon) \\
\\ \hline \\
\textsf{wpo}_\text{sh} (  \psi , x := t, \epsilon) &= 
\begin{cases}
\exists  x' .  \psi [x := x']  * x \approx  t  [x := x']   & (\epsilon=\textit{ok} )\\
\textsf{false} & (\epsilon= \textit{er})
\end{cases} \\
\textsf{wpo}_\text{sh} (\psi ,x:=\texttt{*},\epsilon) &=
\begin{cases}
\exists x'. \psi [x := x'] & (\epsilon=\textit{ok}) \\
\textsf{false} & (\epsilon= \textit{er})
\end{cases} \\
\textsf{wpo}_\text{sh} (\psi ,x:=\texttt{alloc()},\epsilon) &= 
\begin{cases}
\exists x'.  
(  \psi[x := x'] * x \mapsto - ) \lor
&( \epsilon=\textit{ok}) \\
 \bigvee_{j=1}^n \exists x'.   ( 
\left(\bigast_{i=1}^n y_i \not\mapsto \right) [y_j \not\mapsto := y_j \mapsto - ] *
x \approx y_j * \psi' [x := x'] ) & \\
\big(\text{where } \psi = \left(\bigast_{i=1}^n y_i \not\mapsto \right) * \psi',\; \text{and } \not\mapsto \text{ does not occur in } \psi' \big)
 & \\
\textsf{false} & (\epsilon= \textit{er}) \\
\end{cases}  \\
\textsf{wpo}_\text{sh} (\psi , \texttt{free($x$)},\textit{ok}) &=
\begin{cases}
 \psi' * y \not\mapsto  
&  \text{(if $y \in \textsf{Aliases}(x, \psi)$  and
$\psi = \psi' *  y \mapsto t $)}  \\
\textsf{false}  &   \text{(otherwise)}
\end{cases} \\
\textsf{wpo}_\text{sh} ( \psi , \texttt{free($x$)},\textit{er}) &=
\begin{cases}
 \psi & 
 \text{(if $\forall v \in \textsf{Aliases}(x, \psi). \ (v \mapsto t) \notin \psi$)} 
\\ 
 \textsf{false} &  \text{(otherwise)}
\end{cases} \\
\textsf{wpo}_\text{sh} ( \psi ,  x:=[y] ,\textit{ok}) &=
\begin{cases}
\exists x'. \psi [x := x'] * x \approx (t [x := x']) &
\text{(if $z \in \textsf{Aliases}(y, \psi)$ and
$\psi = \psi' *  z \mapsto t$)}  \\
\textsf{false} & \text{(otherwise)}
\end{cases} \\
\textsf{wpo}_\text{sh} ( \psi , x:=[y]  ,\textit{er}) &=
\begin{cases}
\psi 
& 
 \text{(if $\forall v \in \textsf{Aliases}(y, \psi). \ (v \mapsto t) \notin \psi $)} \\
\textsf{false} & \text{(otherwise)}
\end{cases} \\
\textsf{wpo}_\text{sh} ( \psi ,  [x]:=t ,\textit{ok}) &=
\begin{cases}
\psi' *  z \mapsto t
& \text{(if $z \in \textsf{Aliases}(x, \psi)$ and 
$\psi = \psi' *  z \mapsto t'$)} \\
\textsf{false} & \text{(otherwise)}
\end{cases} \\
\textsf{wpo}_\text{sh} ( \psi , [x]:=t  ,\textit{er}) &=
\begin{cases}
\psi 
& 
 \text{(if $\forall v \in \textsf{Aliases}(x, \psi). \ (v \mapsto t') \notin \psi$)} 
\\ 
\textsf{false} & \text{(otherwise)}
\end{cases} \\
\\ \hline
\end{align*}
\end{minipage}
}
\caption{$\normalfont\textsf{wpo}_\text{sh}(\psi, \mathbb{C} , \epsilon)$
for a quantifier-free symbolic heap $\psi$}
\label{fig: wpo for clauses}
\end{figure}

\section{Proving Relative Completeness}\label{sec: body-completeness}

In this section, we prove soundness and relative completeness of ISL.

Soundness is proved in a standard way.

\begin{proposition}[Soundness of ISL]\label{lma: local soundness}
\begin{center}
For any $P,\mathbb{C},\epsilon,Q$, if 
$\vdash  [P] \ \mathbb{C} \ [\epsilon: Q]$, then
$\models [P] \ \mathbb{C} \ [\epsilon: Q]$.
\end{center}
\end{proposition}
\begin{proof}
 It follows from local soundness of each inference rule.
 The proof is in Section \ref{appendix: proof of soundness}.
\end{proof}

For relative completeness, we will show the following two
propositions.  First, we prove expressiveness of the weakest
postconditions (Proposition \ref{prop: expressiveness of WPO calculus}),
that is, $\textsf{wpo}(P,\mathbb{C},\epsilon)$ exactly represents
$\textsf{WPO}\llbracket P,\mathbb{C},\epsilon\rrbracket $. Secondly, we prove that the
specification whose postcondition is $\textsf{wpo}$ can be derived in
ISL (Proposition \ref{prop: for all p, c, epsilon, we have wpo}).

To prove expressiveness, we show the following lemmas.

\begin{lemma}\label{lma: quantifier is free in and out of WPO}
If $x \notin \normalfont\textsf{fv}(\mathbb{C})$, then the following holds for any 
$(s', h')$, $\psi$, $\mathbb{C}$ and $\epsilon$.
    $$
(s', h') \in \text{\normalfont{WPO}}\llbracket \exists x. \psi,\mathbb{C},\epsilon\rrbracket  \iff
\exists v \in {\normalfont\textsc{Val}} . 
(s' [x \mapsto v] , h') \in \text{\normalfont{WPO}}\llbracket \psi,\mathbb{C},\epsilon\rrbracket  
    $$
\end{lemma}
\begin{proof}
The proof is in Appendix \ref{appendix's subsec: proof of quantifier is free in and out of WPO}.
\end{proof}

\begin{lemma}\label{lma: from bigvee to bigcup}
$$
\text{\normalfont{WPO}}\llbracket  \bigvee_{i \in I}  P ,\mathbb{C},\epsilon \rrbracket  \iff 
\bigcup_{i \in I} \text{\normalfont{WPO}}\llbracket  P,\mathbb{C},\epsilon\rrbracket 
$$
\end{lemma}
\begin{proof}
The proof is in Appendix \ref{appendix's subsec: proof of bigvee to bigcup}.
\end{proof}

\begin{lemma}\label{lma: sensei's one shot lemma}
Suppose $\normalfont \textsf{cano} (P, \mathbb{C}) = \bigvee_{i \in I} \exists \overrightarrow{x_i} . \psi_i$.
If for all $i \in I$, $(s',h')$ and $\epsilon$, 
$\normalfont
(s', h') \models \textsf{wpo}_\text{sh} (\psi_i, \mathbb{C}, \epsilon) \iff 
(s', h') \in \text{WPO} \llbracket  \psi_i ,  \mathbb{C}, \epsilon  \rrbracket $,
then for any $(s',h')$ and $\epsilon$, 
$\normalfont
\text{\textnormal{(}}s', h'\text{\textnormal{)}} \models \textsf{wpo} (P, \mathbb{C}, \epsilon) \iff 
\text{\textnormal{(}}s', h'\text{\textnormal{)}} \in \text{WPO} \llbracket  P ,  \mathbb{C}, \epsilon  \rrbracket  $ holds.
\end{lemma}
\begin{proof}
\begin{align*}\normalfont
& \big( \ (s', h') \models \textsf{wpo}_\text{sh} (\psi_i, \mathbb{C}, \epsilon) \iff 
(s', h') \in \text{WPO} \llbracket  \psi_i ,  \mathbb{C}, \epsilon  \rrbracket  \ \big) \\
\Rightarrow \ & 
\big( \ \exists \overrightarrow{v_i} \in \textsc{Val} . 
(s' [ \overrightarrow{x_i} \mapsto \overrightarrow{v_i} ] , h') \models \textsf{wpo}_\text{sh} (\psi_i, \mathbb{C}, \epsilon) \iff \\
& (s' [ \overrightarrow{x_i} \mapsto \overrightarrow{v_i} ] , h') \in \text{WPO} \llbracket  \psi_i ,  \mathbb{C}, \epsilon  \rrbracket \ \big) \\
\Leftrightarrow \ & 
\big( \ (s'  , h') \models \exists \overrightarrow{x_i} . \textsf{wpo}_\text{sh} (\psi_i, \mathbb{C}, \epsilon) \iff 
(s'  , h') \in \text{WPO} \llbracket  \exists \overrightarrow{x_i} . \psi_i ,  \mathbb{C}, \epsilon  \rrbracket \ \big) \\
& \text{// Lemma \ref{lma: quantifier is free in and out of WPO}}\\
\Rightarrow \ & 
\big( \ (s'  , h') \models \bigvee_{i \in I}
\exists \overrightarrow{x_i} . \textsf{wpo}_\text{sh} (\psi_i, \mathbb{C}, \epsilon) \iff 
(s'  , h') \in 
\bigcup_{i \in I} \text{WPO} \llbracket  \exists \overrightarrow{x_i} . \psi_i ,  \mathbb{C}, \epsilon  \rrbracket \ \big)  \\
\Leftrightarrow \ & 
\big( \  (s', h') \models \textsf{wpo} (P, \mathbb{C}, \epsilon) \iff 
(s', h') \in \text{WPO} \llbracket  \bigvee_{i \in I} \exists \overrightarrow{x_i} . \psi_i ,  \mathbb{C}, \epsilon  \rrbracket \ \big) \\
& \text{// Lemma \ref{lma: from bigvee to bigcup}} \\
\Leftrightarrow \ & 
\big( \ (s', h') \models \textsf{wpo} (P, \mathbb{C}, \epsilon) \iff 
(s', h') \in \text{WPO} \llbracket  P ,  \mathbb{C}, \epsilon  \rrbracket \ \big)
\end{align*}
\end{proof}

\begin{lemma}\label{lma: expressiveness of atomic command}
$$ \forall \sigma' . 
\sigma' \in \text{\normalfont{WPO}}\llbracket  \psi,\mathbb{C},\epsilon\rrbracket 
\iff
\sigma' \models  \normalfont \textsf{wpo}_\text{sh} (\psi, \mathbb{C}, \epsilon) $$
\end{lemma}
\begin{proof}
The proof is provided in Appendix \ref{appendix's subsec: proof of expressiveness of atomic command}, utilizing induction on $\mathbb{C}$.
\end{proof}

Proposition \ref{prop: expressiveness of WPO calculus} explains expressiveness concerning general command $\mathbb{C}$, i.e., the assertions derived by $\textsf{wpo}$ exactly describe the weakest postconditions.

\begin{proposition}[Expressiveness]\label{prop: expressiveness of WPO calculus}
\begin{center}
$\normalfont\forall \sigma'. \sigma' \in \text{\normalfont{WPO}}\llbracket P,\mathbb{C},\epsilon\rrbracket  \iff
\normalfont \sigma' \models 
\textsf{wpo} (
P ,\mathbb{C},\epsilon)$
\end{center}
\end{proposition}
\begin{proof}
This proposition holds by applying Lemma \ref{lma: quantifier is free in and out of WPO}, \ref{lma: from bigvee to bigcup}, \ref{lma: sensei's one shot lemma} and \ref{lma: expressiveness of atomic command}.
\end{proof}

We show the following lemmas to prove the second proposition.

\begin{lemma}\label{lma: for all derivation if and then}
For any $P$, $\mathbb{C}$, $\epsilon$, let us say  $ \normalfont\textsf{cano}(P, \mathbb{C}) =   \bigvee_{i \in I} \exists \overrightarrow{x_i}. \psi_i$. 
Then, the following holds:
$$\normalfont 
\forall i \in I. \vdash [\psi_i] \ \mathbb{C} \ 
[ \epsilon :  
\textsf{wpo}_\text{sh} (\psi_i,\mathbb{C},\epsilon) ]  \ \Longrightarrow \ 
\vdash [P] \ \mathbb{C} \ 
[ \epsilon :  
\textsf{wpo}(P,\mathbb{C},\epsilon) ]$$
\end{lemma}
\begin{proof}
We can assume that $\overrightarrow{x_i} \notin \textsf{fv}(\mathbb{C})$ for each 
$i \in I$.
Then, $[ P ]  \ \mathbb{C} \ 
[ \epsilon : \textsf{wpo} (P, \mathbb{C}, \epsilon ) ] $ is derivable as follows.
\[ 
\inferrule*[right=\textsc{Exist}]
{
[  \psi_i ]  \ \mathbb{C} \ 
[ \epsilon :    \textsf{wpo}_\text{sh}(  \psi_i ,\mathbb{C}, \epsilon ) ]
\ \ \  ( \overrightarrow{x_i} \notin \textsf{fv}(\mathbb{C}) ) \ \ \text{for all $i \in I$}
}
{\inferrule*[right=\textsc{Disj}]
{ 
 [ \exists \overrightarrow{x_i}  .    \psi_i ]  \ \mathbb{C} \ 
[ \epsilon :   \exists \overrightarrow{x_i}  .  \textsf{wpo}_\text{sh}(  \psi_i ,\mathbb{C}, \epsilon ) ] 
\ \ \text{for all $i \in I$} }
{\inferrule*[right=\textsc{Cons}]
{[  \bigvee_{i \in I} \exists \overrightarrow{x_i}  .  \psi_i ]  \ \mathbb{C} \ 
[ \epsilon : \bigvee_{i \in I} \exists \overrightarrow{x_i}  .  \textsf{wpo}_\text{sh}(   \psi_i ,\mathbb{C}, \epsilon ) ] }
{[ P ]  \ \mathbb{C} \ 
[ \epsilon : \textsf{wpo} (P, \mathbb{C}, \epsilon ) ] } }}
\]
\end{proof}

\begin{lemma}\label{thm: for all p, c, epsilon, we have wpo}
For any $\psi,\mathbb{C},\epsilon$, we have
$\vdash [ \psi] \ \mathbb{C} \ 
[\normalfont \epsilon :  \textsf{wpo}_\text{sh} (\psi,\mathbb{C},\epsilon) ]$.
\end{lemma}
\begin{proof}
We prove this lemma by induction on $\mathbb{C}$.
Here, we only prove the case of $\textsc{Choice}$ and leave proofs of the other cases in Appendix \ref{appendix: proof of thm for all we have wpo}.

\noindent
\resizebox{0.9\textwidth}{!}{%
  \begin{minipage}{\linewidth}
\[
\inferrule*[right=\textsc{Disj}]
{
\inferrule*[right=Choice]
{
    \stackrel{\text{(Induction Hypothesis)}}%
{[\psi]  \ \mathbb{C}_1  \ 
[\epsilon : \textsf{wpo}_\text{sh} (\psi, \mathbb{C}_1  , \epsilon ) ]}
}
{[\psi]  \ \mathbb{C}_1 + \mathbb{C}_2 \ 
[\epsilon : \textsf{wpo}_\text{sh} (\psi, \mathbb{C}_1  , \epsilon ) ]} 
\ \ \ 
\inferrule*[right=Choice]
{
    \stackrel{\text{(Induction Hypothesis)}}%
{[\psi]  \ \mathbb{C}_2  \ 
[\epsilon : \textsf{wpo}_\text{sh} (\psi, \mathbb{C}_2  , \epsilon ) ]}
}
{[\psi]  \ \mathbb{C}_1 + \mathbb{C}_2 \ 
[\epsilon : \textsf{wpo}_\text{sh} (\psi, \mathbb{C}_2  , \epsilon ) ]}  }
{\inferrule*[right=\normalfont\text{By the definition}]
{ [\psi]  \ \mathbb{C}_1 + \mathbb{C}_2 \ 
[\epsilon : \textsf{wpo}_\text{sh} (\psi, \mathbb{C}_1  , \epsilon ) \lor
\textsf{wpo}_\text{sh} (\psi,   \mathbb{C}_2 , \epsilon )] }
{ [\psi]  \ \mathbb{C}_1 + \mathbb{C}_2 \ 
[\epsilon : \textsf{wpo}_\text{sh} (\psi, \mathbb{C}_1 + \mathbb{C}_2 , \epsilon )] }} 
\]
  \end{minipage}
}
\end{proof}

\begin{proposition}\label{prop: for all p, c, epsilon, we have wpo}
For any $P,\mathbb{C},\epsilon$,
$\normalfont\vdash [P] \ \mathbb{C} \ 
[ \epsilon :  
\textsf{wpo}(P,\mathbb{C},\epsilon) ]$
\end{proposition}
\begin{proof}
By leveraging Lemma \ref{lma: for all derivation if and then} and 
Lemma \ref{thm: for all p, c, epsilon, we have wpo},  this proposition is established.
\end{proof}
By Proposition \ref{prop: expressiveness of WPO calculus}
and Proposition \ref{prop: for all p, c, epsilon, we have wpo},
we prove relative completeness of ISL.

\begin{theorem}[Relative Completeness]
For any $P,\mathbb{C},\epsilon,Q$, if 
$\models [P] \ \mathbb{C} \ [\epsilon: Q]$, then
$\vdash  [P] \ \mathbb{C} \ [\epsilon: Q]$.
\end{theorem}

\begin{proof}
 Since $[P] \ \mathbb{C} \ [\epsilon: Q]$ is valid, we have $\forall
 \sigma'\models Q.\sigma'\in\textsf{WPO}\llbracket P,\mathbb{C},\epsilon\rrbracket $
 by the definition of the validity of ISL triples. By Proposition
 \ref{prop: expressiveness of WPO calculus}, we have $\forall
 \sigma'\models Q.\sigma'\models\textsf{wpo}(P,\mathbb{C},\epsilon)$, which is equivalent to  $Q\models \textsf{wpo}(P,\mathbb{C},\epsilon)$.

 Then, $[P]\ \mathbb{C}\ [Q]$ is derivable in ISL as follows.
\[
\inferrule*[right=\textsc{Cons}]
  {%
    \stackrel{\text{(By Proposition \ref{prop: for all p, c, epsilon, we have wpo})}}%
             {\vdash [P]\,\mathbb{C}\,[\epsilon:\textsf{wpo}(P,\mathbb{C},\epsilon)]}
    \and
    \stackrel{\text{(By Proposition \ref{prop: expressiveness of WPO calculus})}}%
             {Q \models \textsf{wpo}(P,\mathbb{C},\epsilon)}
  }
  {%
    \vdash [P]\,\mathbb{C}\,[\epsilon:Q]
  }
\]
\end{proof}

\section{Related Work}\label{sec: Related Work}
\subsubsection{Proving relative completeness for graph manipulation.}
Poskitt et al. demonstrate soundness and relative completeness for an under-approximation program logic for a graph manipulation language using the calculation of the weakest postconditions \cite{poskitt2023monadic}. 
Therein, they prove relative completeness of an extensional logic with semantic predicates.
This setting is similar to that of IL \cite{o2019incorrectness} in that there is no need to prove expressiveness. 
Relative completeness of the intensional logic, which does not allow for semantic predicates and requires proof of expressiveness, remains unknown.

\subsubsection{Unifying correctness and incorrectness.}

Efforts to integrate correctness and incorrectness reasoning within a unified program logic are demonstrated by the work of Bruni et al. \cite{bruni2023correctness,bruni2021logic}.
They introduce the Local Completeness Logic, which imposes constraints on the rule of consequence to thereby guarantee the recoverability of an over-approximation of states that are reachable from the postcondition. Additionally, they propose a concept of local completeness that ensures that no false alarms are generated relative to some fixed input.

Similarly, Maksimović et al. \cite{maksimovic2023exact} propose Exact SL, whose operational semantics is complete; their aim is to unify correctness and incorrectness considerations. 
However, it encounters challenges that stem from the limitations of the rule of consequence.
Exact SL introduces explicit specifications of both ``\textit{ok}'' and ``\textit{er}'' cases in postconditions, and it has advantages in addressing   correctness (over-approximation) and incorrectness (under-approximation) of programs. 
However, this approach introduces a trade-off, as it involves a more restrictive coverage of proof rules, such as the rule of consequence, which is applicable only to equivalent formulas.

Outcome logic \cite{zilberstein2023outcome} and 
Outcome SL \cite{zilberstein2024outcome} stand out in unifying over and under-approximate reasoning for heap-manipulating and probabilistic programs. 
Outcome Logic combines and generalizes standard over-approximate Hoare triples with forward under-approximate triples. 
In \cite{zilberstein2024relatively},   relative completeness of Outcome Logic is discussed.

\section{Conclusions and Future work}\label{sec: Conclusions and Future work}
In this study, we have proven   relative completeness of ISL. 
Our objective was to maintain   relative completeness of RHL while expanding to include exit conditions and heap manipulation within ISL. 
The calculation of the weakest postconditions in our ISL and the demonstration of its expressiveness are key to this proof.  
However, this requires a trade-off by allowing for infinite disjunctions, as is done in the proof of RHL.

For our future work, we are exploring the possibility of demonstrating
relative completeness for ISL within a finite syntax extended by some arithmetic theory.
Although infinitary syntax aligns with our research goals to prove relative completeness, it may not be the best fit for constructing a practical
automatic theorem prover.

\subsubsection*{Acknowledgments}
We would like to express our gratitude to Professor Shoji Yuen for providing thoughtful guidance on our research. We also extend our thanks to Professor Hiroyuki Seki and Professor Yuichi Kaji from our research group for their valuable remarks. Additionally, we are grateful to the three anonymous referees of APLAS 2024 and an expert reviewer for their insightful comments and suggestions.

This work was supported by JSPS KAKENHI Grant Number JP22K11901. Furthermore, Yeonseok Lee is financially supported by TMI, one of the WISE programs established by MEXT Japan, as well as by JST SPRING under Grant Number JPMJSP2125. The authors also thank the Interdisciplinary Frontier Next-Generation Researcher Program of the Tokai Higher Education and Research System.

\newpage

%
%
%
\bibliographystyle{splncs04}

\bibliography{bib}

\appendix
\newpage
\newgeometry{left=1in, right=1in, top=1in, bottom=1in}
\section{Proof of Proposition \ref{lma: local soundness}}
\label{appendix: proof of soundness}
The proof rules in our system that do not manipulate the heap are adopted from the original ISL paper \cite{raad2020local} or from Reverse Hoare Logic (RHL) \cite{de2011reverse}, 
where their local 
soundness has already been established. 
Therefore, we do not reprove soundness for these rules. 
However, for the rules that do involve heap manipulation, we provide proofs of  soundness in this paper. 
This is necessary because our semantics for heap manipulation differ from those in \cite{raad2020local}.
Before proving soundness (Proposition \ref{lma: local soundness}), we introduce several lemmas.

\subsection{Lemmas  supporting Proposition \ref{lma: local soundness}}

We recall the heap monotonicity property from the original ISL formulation~\cite{raad2020local}.
\begin{lemma}[Heap Monotonicity~{\cite{raad2020local}}]
\label{lma:Heap-Monotonicity}
Let \( \mathbb{C} \) be any command, and suppose that
\[
((s,h), (s',h')) \in \llbracket \mathbb{C} \rrbracket_{\epsilon}
\quad \text{for } \epsilon \in \{\textit{ok}, \textit{er} \}.
\]
Then,
\[\normalfont
\textsf{dom}(h) \subseteq \textsf{dom}(h').
\]
\end{lemma}

\begin{lemma}[Frame-Preservation Property]
\label{lma: Frame-Preservation Property}
Let \( \mathbb{C} \) be any command, and suppose that
\[
((s,h), (s',h')) \in \llbracket \mathbb{C} \rrbracket_{\textit{ok}}
\quad \text{and} \quad h_r \text{ is disjoint from both } h \text{ and } h'.
\]
Then,
\[
((s, h \circ h_r), (s', h' \circ h_r)) \in \llbracket \mathbb{C} \rrbracket_{\textit{ok}}.
\]
\end{lemma}

\begin{lemma}
\label{lma:canonical_implies_allocated_alias}
Let $\normalfont \psi \in \textsf{CF}_\text{sh}(V)$ 
where $\normalfont \{x\} \cup \textsf{fv}(\psi) \subseteq V$. 
If $(s,h) \models \psi$ and $ \normalfont s(x) \in \textsf{dom}_+(h)$, 
then there exists a variable $\normalfont v \in \textsf{Aliases}(x, \psi)$ 
such that $v \mapsto t \in \psi$ for some term $t$.
\end{lemma}

\subsection{Proof of Proposition~\ref{lma: local soundness}}

From now on, we prove Proposition~\ref{lma: local soundness} for each of the proof rules.

\subsubsection{Case {$\normalfont \textsc{FrameOk}$}}  

We prove the following.

\medskip
\textit{For any assertions $\psi$, $\varphi$, and $\phi$, and any command $\mathbb{C}$, if}
\[
\models [\psi]\ \mathbb{C} \ [ok : \varphi]
\quad\textit{and}\quad
\mathsf{mod}(\mathbb{C}) \cap \mathsf{fv}(\phi) = \emptyset
\]
\textit{hold, then the framed triple}
\[
\models [\psi * \phi]\ \mathbb{C} \ [ok : \varphi * \phi]
\]
\textit{also holds.}

Let $\sigma' = (s_2, h_\varphi \circ h_\phi)$ be an arbitrary state such that
\begin{center}
$\sigma' \models \varphi * \phi$,
where
$(s_2, h_\varphi) \models \varphi$ and
$(s_2, h_\phi) \models \phi$.
\end{center}

From the assumption $\models [\psi] \ \mathbb{C} \ [ok : \varphi]$ and $(s_2, h_\varphi) \models \varphi$, 
the definition of semantic validity yields a state $(s_1, h_\psi)$ such that:
\begin{align}
(s_1, h_\psi) &\models \psi, \\
((s_1, h_\psi), (s_2, h_\varphi)) &\in \llbracket \mathbb{C} \rrbracket_{ok}. \label{eq:premise-trans}
\end{align}

Moreover, by the semantics of $\llbracket \mathbb{C} \rrbracket_{ok}$,
\[
\forall x \notin \mathsf{mod}(\mathbb{C}).\quad s_1(x) = s_2(x).
\]

Since $\mathsf{mod}(\mathbb{C}) \cap \mathsf{fv}(\phi) = \emptyset$, it follows that $s_1$ and $s_2$ agree on all variables in $\mathsf{fv}(\phi)$. That is,
$$\forall x \in \mathsf{fv}(\phi). \quad s_1(x) = s_2(x).$$
Thus, by the satisfaction semantics:
\[
(s_1, h_\phi) \models \phi.
\]

We have
$((s_1, h_\psi), (s_2, h_\varphi)) \in \llbracket \mathbb{C} \rrbracket_{ok}$
by (\ref{eq:premise-trans}).  
Then, by heap monotonicity (Lemma~\ref{lma:Heap-Monotonicity}),
\[
\textsf{dom}(h_\psi) \subseteq \textsf{dom}(h_\varphi).
\]
Since we have \( \textsf{dom}(h_\varphi) \cap \textsf{dom}(h_\phi) = \emptyset \) from  
\( \sigma' \models \varphi * \phi \),  
it follows that \( \textsf{dom}(h_\psi) \cap \textsf{dom}(h_\phi) = \emptyset \).  
Therefore, we have:
\[
(s_1,\, h_\psi \circ h_\phi) \models \psi * \phi.
\]

Combining \eqref{eq:premise-trans} with 
Frame-Preservation Property (Lemma \ref{lma: Frame-Preservation Property}), we obtain
\[
((s_1,\, h_\psi \circ h_\phi),\; (s_2,\, h_\varphi \circ h_\phi))
      \in \llbracket \mathbb{C} \rrbracket_{ok}.
\]

Hence, for any state
\(\sigma' = (s_2,\, h_\varphi \circ h_\phi)\) satisfying
\(\sigma' \models \varphi * \phi\),
we can construct the state
\(\sigma = (s_1,\, h_\psi \circ h_\phi)\) such that
\[
\sigma \models \psi * \phi
\quad\text{and}\quad
(\sigma,\, \sigma') \in \llbracket \mathbb{C} \rrbracket_{ok}.
\]

This completes the proof of local soundness.

\subsubsection{Case {$\normalfont \textsc{Alloc1}$}}   
\begin{align*}
    & \models [\psi] \ x:=\texttt{alloc()} \ 
\left[ \textit{ok}:  
\exists x'.   \psi[x := x'] * x \mapsto -  \right]  \\ 
\Leftrightarrow \ & 
\forall (s',h') \models \exists x'.   \psi[x := x'] * x \mapsto -  , 
\exists (s,h) \models \psi . 
( (s,h) , (s', h') ) \in \llbracket  x:=\texttt{alloc()} \rrbracket _\textit{ok}
\end{align*}
We show the triple is valid.
Let us say that $(s',h')$ is any state satisfying $\exists x'.   \psi[x := x'] * x \mapsto -$.
Then the following holds.
\begin{align*}
&  (s',h') \models \exists x'.   \psi[x := x'] * x \mapsto -  \\
\Leftrightarrow \ &
\exists v' . (s' [x' \mapsto v'] ,h') \models  \psi[x := x'] * x \mapsto -  \\
\Leftrightarrow \ &
\exists v' . (s' [x' \mapsto v'] , h' - ( s'(x) \mapsto h'(s'(x)) ) ) \models  \psi[x := x'] \\
& \text{ and }
(s' [x' \mapsto v'] , s'(x) \mapsto h'(s'(x)) ) \models  x \mapsto -  
\\
& \text{// \( h' - ( s'(x) \mapsto h'(s'(x)) ) \) is the same function as \( h' \), except that its domain no longer contains \( s'(x) \).} \\
& \text{// Since $x'$ is fresh, $s' [x' \mapsto v'] (x) = s'(x)$ holds.} \\
\Leftrightarrow \ &
\exists v' . (s' [x' \mapsto v'] [x \mapsto v'] , h' - ( s'(x) \mapsto h'(s'(x)) ) ) \models  \psi 
\text{ // By Lemma \ref{lma: Substitution for assignment}}\\
& \text{ and }
(s' [x' \mapsto v'] , s'(x) \mapsto h'(s'(x)) ) \models  x \mapsto -  \\
\Leftrightarrow \ &
\exists v' . (s' [x \mapsto v'] , h' - ( s'(x) \mapsto h'(s'(x)) ) ) \models  \psi 
\text{ // Since $x'$ is fresh}\\
& \text{ and }
(s' [x' \mapsto v'] , s'(x) \mapsto h'(s'(x)) ) \models  x \mapsto -  
\end{align*}
If we pick $ (s' [x \mapsto v'] , h' - ( s'(x) \mapsto h'(s'(x)) ) ) $ as $(s,h)$, then $( (s,h) , (s', h') ) \in \llbracket  x:=\texttt{alloc()} \rrbracket _\textit{ok}$ holds.
Because $s'(x) \notin \textsf{dom}(h)$ and $h'(s'(x)) \in \textsc{Val}$ holds.
$(s,h) \models \psi$ is straightforward by above.

The case for \textit{er} is straightforward since the postcondition of the case is $\textsf{false}$, making the triple trivially valid.

\subsubsection{Case {$\normalfont \textsc{Alloc2}$}}   
\begin{align*}
    & \models [   \psi * y \not\mapsto ] \ x:=\texttt{alloc()} \ 
[\textit{ok}:  
\exists x'.  x \mapsto - *  x \approx y * \psi [x:=x'] ]  \\
\Leftrightarrow \ & 
\forall (s',h') \models \exists x'.  x \mapsto - *  x \approx y * \psi [x:=x']  , 
\exists (s,h) \models  \psi * y \not\mapsto . 
( (s,h) , (s', h') ) \in \llbracket  x:=\texttt{alloc()} \rrbracket _\textit{ok}  
\end{align*}
We show the triple is valid.
Let us say that $(s',h')$ is any state satisfying $\exists x'.  x \mapsto - *  x \approx y * \psi [x:=x'] $.
Then the following holds.
\begin{align*}
&  (s',h') \models \exists x'.  x \mapsto - *  x \approx y * \psi [x:=x']   \\
\Leftrightarrow \ &
\exists v' . (s' [x' \mapsto v'] ,h') \models   x \mapsto - *  x \approx y * \psi [x:=x']  \\
\Rightarrow \ &
\exists v' . (s' [x' \mapsto v'] ,h' [s' [x' \mapsto v'](x) \mapsto \bot] ) \models  x \not\mapsto *  x \approx y * \psi [x:=x'] \\
\Leftrightarrow \ &
\exists v' . (s' [x' \mapsto v'] ,h' [s' (x) \mapsto \bot] ) \models 
y \not\mapsto *  x \approx y * \psi [x:=x'] \\
& \text{// By $ x \approx y $, and $s' [x' \mapsto v'](x)  = s'(x)$ since $x'$ is fresh.} \\
\Rightarrow \ &
\exists v' . (s' [x' \mapsto v'] ,h' [s' (x) \mapsto \bot] ) \models 
y \not\mapsto *  \psi [x:=x'] 
\text{ and }
s' [x' \mapsto v'] (x) =  s' [x' \mapsto v'] (y)  \\
\Leftrightarrow \ &
\exists v' . (s'  [x \mapsto v'] ,h' [s' (x) \mapsto \bot] ) \models 
y \not\mapsto *  \psi 
\text{ and }
s'  (x) =  s'  (y)  \\
& \text{// Lemma \ref{lma: Substitution for assignment} and $x'$ is fresh} 
\end{align*}
If we pick $ (s'  [x \mapsto v'] ,h' [s' (x) \mapsto \bot] )  $ as $(s,h)$, then $( (s,h) , (s', h') ) \in \llbracket  x:=\texttt{alloc()} \rrbracket _\textit{ok}$ holds.
Because $h(s'(x)) = \bot $ and $h'(s'(x)) \in \textsc{Val}$ holds.
$(s,h) \models y \not\mapsto * \  \psi$ is straightforward by above.

The case for \textit{er} is straightforward since the postcondition of the case is $\textsf{false}$, making the triple trivially valid.

\subsubsection{Case {$\normalfont \textsc{Free}$} }   
\begin{align*}
    & \models [ \psi ] \ \texttt{free($x$)} \ 
[\textit{ok}:  \psi' *  y \not\mapsto  ]  \\
\Leftrightarrow \ & 
\forall (s',h') \models \psi' *  y \not\mapsto   , 
\exists (s,h) \models  \psi. 
( (s,h) , (s', h') ) \in \llbracket  \texttt{free($x$)} \rrbracket _\textit{ok} 
\end{align*}
We show the triple is valid.
Let us say that $(s',h')$ is any state satisfying $\psi' *  y \not\mapsto$.
Then the following holds.
\begin{align*}
&  (s',h') \models \psi' *  y \not\mapsto  \\
\Rightarrow \ & 
(s',h' [s'(x) \mapsto s'(t) ] ) \models \psi' *  y \mapsto  t \\
\Leftrightarrow \ & 
(s',h' [s'(x) \mapsto s'(t) ] ) \models \psi \ \text{ // By the assumption }
\end{align*}
We pick $ (s',h' [s'(x) \mapsto s'(t) ] )  $ as $(s,h)$. 
Then, $( (s,h) , (s', h') ) \in \llbracket  \texttt{free($x$)} \rrbracket _\textit{ok} $ holds, since $s(x) = s'(x) \in \textsf{dom}_+ (h)$.

The case for \textit{er} is straightforward since the postcondition of the case is $\textsf{false}$, making the triple trivially valid.

\subsubsection{Case {$\normalfont \textsc{FreeEr}$}} 
\begin{align*}
    & \models [ \psi ] \ \texttt{free($x$)} \ 
[\textit{er}:  \psi  ]  \\
\Leftrightarrow \ & 
\forall (s',h') \models \psi  , 
\exists (s,h) \models  \psi. 
( (s,h) , (s', h') ) \in \llbracket  \texttt{free($x$)} \rrbracket _\textit{er} 
\end{align*}

We show that this triple is valid.  
Let $(s',h')$ be an arbitrary state such that $(s',h')\models\psi$.  
By the premise, we know
\[
\psi \in \textsf{CF}_{\text{sh}}(\textsf{fv}(\psi)\cup\{x\})
\quad\text{and}\quad
\forall\,v\in\textsf{Aliases}(x,\psi).\;
(v \mapsto t) \notin \psi.
\]

We pick $(s',h')$ as $(s,h)$.
Then $(s,h)\models\psi$ by assumption.
If \( s(x) \in \textsf{dom}_{+}(h) \), there is a variable 
\( w \in \textsf{Aliases}(x, \psi) \) such that \( w \mapsto t \in \psi \)
by Lemma \ref{lma:canonical_implies_allocated_alias}.  
This contradicts the assumption.  
Hence, \( s(x) \notin \textsf{dom}_{+}(h) \).

Therefore,
\(
((s,h),(s',h')) \in \llbracket\texttt{free}(x)\rrbracket_\textit{er},
\)
and the triple is valid.

As in the previous cases, the \emph{ok} branch is trivial due to the
postcondition being \textsf{false}.

\subsubsection{Case {$\normalfont \textsc{Load}$}}   
\begin{align*}
    & \models [ \psi ] \ x := [y] \ 
[\textit{ok}:  \exists x' .  \psi [x:=x'] * x \approx (t [x:=x'] ) ]  \\
\Leftrightarrow \ & 
\forall (s',h') \models
\exists x' .  \psi [x:=x'] * x \approx (t [x:=x'] )  , 
\exists (s,h) \models  \psi. 
( (s,h) , (s', h') ) \in \llbracket  x := [y] \rrbracket _\textit{ok} 
\end{align*}
We show that this triple is valid.
Let us say that $(s', h') $ is any state satisfying $ \exists x' .  \psi [x:=x'] * x \approx (t [x:=x'] ) $.
Then the following holds.
\begin{align*}
 & (s',h') \models
\exists x' .  \psi [x:=x'] * x \approx (t [x:=x'] ) \\
\Leftrightarrow \ & 
\exists v' . (s' [x' \mapsto v'] ,h') \models
 \psi [x:=x'] * x \approx (t [x:=x'] ) \\
\Rightarrow \ & 
\exists v'  . (s' [x' \mapsto v'] ,h') \models
 \psi [x:=x'] \text{ and } 
 s' [x' \mapsto v'] ( x)  = s' [x' \mapsto v'] (t [x:=x'] ) \\
\Leftrightarrow \ & 
\exists v'  . (s' [x' \mapsto v'] [x \mapsto v'] ,h') \models \psi  \text{ and } 
 s' [x' \mapsto v'] ( x)  = s' [x' \mapsto v'] [x \mapsto v'] (t  )   \text{ // By Lemma \ref{lma: Substitution for assignment}}\\
\Leftrightarrow \ & 
\exists v'  . (s' [x \mapsto v'] ,h') \models \psi  \text{ and } 
 s'  ( x)  = s'  [x \mapsto v'] (t  ) \ \ \  \cdots \ \text{($\spadesuit$)} \text{ // Since $x'$ is fresh}
\end{align*}
It is enough to show that if we pick $(s,h)$ as $(s'[x \mapsto v'], h')$, then 
$((s,h), (s',h')) \in \llbracket  x := [y] \rrbracket _\textit{ok} $ holds.
By $(s,h) \models \psi$, $h(s(y)) = h(s(z)) = s(t) = s'[x \mapsto v'] (t)$ holds.
Consequently, by $s'  ( x)  = s'  [x \mapsto v'] (t  )$ in ($\spadesuit$),
$s[x \mapsto h(s(y))] = s'[x \mapsto v'] [x \mapsto s'[x \mapsto v'] (t) ]
 = s'[x \mapsto s'(x)] = s' $ holds. 
Therefore, $((s,h), (s',h')) \in \llbracket  x := [y] \rrbracket _\textit{ok} $.

The case for \textit{er} is straightforward since the postcondition of the case is $\textsf{false}$, making the triple trivially valid.

\subsubsection{Case {$\normalfont\textsc{LoadEr}$}}

Similar to the case of $\textsc{FreeEr}$.

\subsubsection{Case {$\normalfont \textsc{Store}$}}  
\begin{align*}
    & \models [ \psi ] \ [x] := t \ 
[\textit{ok}:  \psi'  * z \mapsto t  ]  \\
\Leftrightarrow \ & 
\forall (s',h') \models \psi'  * z \mapsto t    , 
\exists (s,h) \models  \psi. 
( (s,h) , (s', h') ) \in \llbracket  [x] := t \rrbracket _\textit{ok} 
\end{align*}
We show the triple is valid.
Let us say that $(s',h')$ is any state satisfying $\psi' *  z \mapsto t$.
Then the following holds.
\begin{align*}
&  (s',h') \models \psi' *  z \mapsto t  \\
\Rightarrow \ & 
(s',h' [s'(x) \mapsto s'(t') ] ) \models \psi' *  z \mapsto  t' \\
\Leftrightarrow \ & 
(s',h' [s'(x) \mapsto s'(t') ] ) \models \psi \ \text{ // By the premise}
\end{align*}
We pick $ (s',h' [s'(x) \mapsto s'(t') ] )  $ as $(s,h)$, 
then $( (s,h) , (s', h') ) \in \llbracket  [x] := t  \rrbracket _\textit{ok} $ holds.
Since $(s,h) \models \psi$, $s(x) \in \textsf{dom}_+ (h)$ holds.

The case for \textit{er} is straightforward since the postcondition of the case is $\textsf{false}$, making the triple trivially valid.

\subsubsection{Case {$\normalfont \textsc{StoreEr}$}}

Similar to the case of $\textsc{FreeEr}$.


\section{Proof of Lemma \ref{lma: quantifier is free in and out of WPO}, 
\ref{lma: from bigvee to bigcup} and \ref{lma: expressiveness of atomic command}}
\label{appendix: proof of lemma (expressiveness of atomic command)}

Lemma \ref{lma: expressiveness of atomic command} is the following.
$$ \forall \sigma' . 
\sigma' \in \text{\normalfont{WPO}}\llbracket   \psi,\mathbb{C},\epsilon\rrbracket 
\iff
\sigma' \models  \normalfont \textsf{wpo}_\text{sh} (\psi, \mathbb{C}, \epsilon)  $$
We prove this lemma by induction on $\mathbb{C}$.
In the proof, we occasionally represent $\sigma$ as $(s, h)$ and $\sigma'$ as $(s', h')$, where $s, s' \in \textsc{Store}$ and $h, h' \in \textsc{Heap}$.
Before proving this lemma, we introduce several lemmas.

\subsection{Lemmas related to Lemma \ref{lma: quantifier is free in and out of WPO}, 
\ref{lma: from bigvee to bigcup} and \ref{lma: expressiveness of atomic command}} 
\label{appendix's subsec: other lemmas related to proof of expressiveness of atomic command}

\begin{lemma}[Substitution]\label{lma: Substitution for assignment} 
    $$s,h \models  \psi [x:=t] 
    \iff s[x \mapsto s(t)], h \models  \psi $$
\end{lemma}

\begin{lemma}\label{lma: quantified variables do not affect to operational semantics}
If $\normalfont x \notin \textsf{fv}(\mathbb{C})$, the following holds.
$$
( (s,h) , (s',h')  ) \in \llbracket   \mathbb{C} \rrbracket _\epsilon \Longrightarrow 
\exists v \in \textsc{Val} . 
( (s [x \mapsto v] ,h) , (s'[x \mapsto v],h')  ) \in \llbracket   \mathbb{C} \rrbracket _\epsilon
$$
\end{lemma}
\begin{proof}
    By induction on $\mathbb{C}$ and the  operational semantics in Definition \ref{def: Denotational semantics of ISL}.
\end{proof}

\begin{lemma}\label{lma: quantified variables have same store value}
If $\normalfont x \notin \textsf{fv}(\mathbb{C})$, the following holds.
$$
( (s,h) , (s',h')  ) \in \llbracket   \mathbb{C} \rrbracket _\epsilon \Longrightarrow 
s(x) = s'(x)
$$
\end{lemma}
\begin{proof}
    By induction on $\mathbb{C}$ and the  operational semantics in Definition \ref{def: Denotational semantics of ISL}.
\end{proof}

\subsection{Proof of Lemma \ref{lma: quantifier is free in and out of WPO}}
\label{appendix's subsec: proof of quantifier is free in and out of WPO}
Lemma \ref{lma: quantifier is free in and out of WPO} is the following.

\textit{If $x \notin \normalfont\textsf{fv}(\mathbb{C})$,
then the following holds for all 
$(s', h')$, $\psi$, $\mathbb{C}$ and $\epsilon$.}
    $$
(s', h') \in \text{\normalfont{WPO}}\llbracket \exists x. \psi,\mathbb{C},\epsilon\rrbracket  \iff
\exists v \in \textsc{Val} . 
(s' [x \mapsto v] , h') \in \text{\normalfont{WPO}}\llbracket \psi,\mathbb{C},\epsilon\rrbracket  
    $$
\begin{proof}
Firstly, we show $(s', h') \in \text{\normalfont{WPO}}\llbracket \exists x. \psi,\mathbb{C},\epsilon\rrbracket  \Longrightarrow
\exists v \in \textsc{Val} . 
(s' [x \mapsto v] , h') \in \text{\normalfont{WPO}}\llbracket \psi,\mathbb{C},\epsilon\rrbracket $.
\begin{align*}
& (s', h') \in \text{\normalfont{WPO}}\llbracket \exists x. \psi,\mathbb{C},\epsilon\rrbracket   \\
\Leftrightarrow \   &  \exists (s,h) . (s  , h) \models \exists x. \psi \land 
( (s,h) , (s',h')  ) \in \llbracket  \mathbb{C} \rrbracket _\epsilon \\
\Leftrightarrow \   &  \exists (s,h), \exists v \in \textsc{Val} . (s [x \mapsto v] , h) \models \psi \land 
( (s,h) , (s',h')  ) \in \llbracket  \mathbb{C} \rrbracket _\epsilon \\
\Rightarrow \   &  
\exists (s,h), \exists v \in \textsc{Val} . (s [x \mapsto v] , h) \models \psi \land 
( (s[x \mapsto v],h) , (s'[x \mapsto v],h')  ) \in \llbracket  \mathbb{C} \rrbracket _\epsilon 
\ \text{ // Lemma \ref{lma: quantified variables do not affect to operational semantics}}\\
\Rightarrow \   &  
\exists (s,h), \exists v \in \textsc{Val} . (s_1, h_1)  \models \psi \land 
( (s_1, h_1)  , (s'[x \mapsto v],h')  ) \in \llbracket  \mathbb{C} \rrbracket _\epsilon  \\
& \text{ // Let us say that $(s_1, h_1) = ( s[x \mapsto v],h ) $ } \\
\Rightarrow \  & 
 \exists (s,h), \exists v \in \textsc{Val} . (s  , h) \models \psi \land 
( (s,h) , (s' [x \mapsto v] , h')  ) \in \llbracket  \mathbb{C} \rrbracket _\epsilon \\
\Leftrightarrow \  & 
\exists v \in \textsc{Val} . 
(s' [x \mapsto v] , h') \in \text{\normalfont{WPO}}\llbracket \psi,\mathbb{C},\epsilon\rrbracket  
\end{align*}
Next, we show $(s', h') \in \text{\normalfont{WPO}}\llbracket \exists x. \psi,\mathbb{C},\epsilon\rrbracket  \Longleftarrow
\exists v \in \textsc{Val} . 
(s' [x \mapsto v] , h') \in \text{\normalfont{WPO}}\llbracket \psi,\mathbb{C},\epsilon\rrbracket $.
\begin{align*}
     & \exists v \in \textsc{Val} . 
(s' [x \mapsto v] , h') \in \text{\normalfont{WPO}}\llbracket \psi,\mathbb{C},\epsilon\rrbracket   \\
\Leftrightarrow \ & 
 \exists (s,h), \exists v \in \textsc{Val} . (s  , h) \models \psi \land 
( (s,h) , (s' [x \mapsto v] , h')  ) \in \llbracket  \mathbb{C} \rrbracket _\epsilon  \\ 
\Rightarrow \ & 
 \exists (s,h), \exists v \in \textsc{Val} . (s [x \mapsto s'(x)] , h) \models \psi \land 
( (s,h) , (s' [x \mapsto v] , h')  ) \in \llbracket  \mathbb{C} \rrbracket _\epsilon  \\ 
\Rightarrow \ & 
 \exists (s,h), \exists v \in \textsc{Val} . 
 (s [x \mapsto s'(x)] , h) \models \psi \land 
( (s[x \mapsto s'(x)],h) , (s' [x \mapsto v][x \mapsto s'(x)] , h')  ) \in \llbracket  \mathbb{C} \rrbracket _\epsilon  \\ 
& \text{// Lemma \ref{lma: quantified variables do not affect to operational semantics}}\\
\Leftrightarrow \ & 
 \exists (s,h), \exists v \in \textsc{Val} . 
 (s [x \mapsto s'(x)][x \mapsto s'(x)] , h) \models \psi \land 
( (s[x \mapsto s'(x)],h) , (s'  , h')  ) \in \llbracket  \mathbb{C} \rrbracket _\epsilon  \\
\Rightarrow \ & 
 \exists (s,h), \exists v \in \textsc{Val} . 
 (s_2[x \mapsto s'(x)] , h_2) \models \psi \land 
( (s_2, h_2) , (s'  , h')  ) \in \llbracket  \mathbb{C} \rrbracket _\epsilon  \\ 
& \text{// Let us say that $(s_2, h_2) = ( s [x \mapsto s'(x)] , h)$} \\
\Rightarrow \   &  \exists (s,h), \exists v \in \textsc{Val} . (s [x \mapsto v] , h) \models \psi \land 
( (s,h) , (s',h')  ) \in \llbracket  \mathbb{C} \rrbracket _\epsilon \\
\Leftrightarrow \   &  \exists (s,h) . (s  , h) \models \exists x. \psi \land 
( (s,h) , (s',h')  ) \in \llbracket  \mathbb{C} \rrbracket _\epsilon \\
\Leftrightarrow \ &  (s', h') \in \text{\normalfont{WPO}}\llbracket \exists x. \psi,\mathbb{C},\epsilon\rrbracket 
\end{align*}

\end{proof}

\subsection{Proof of Lemma \ref{lma: from bigvee to bigcup}}
\label{appendix's subsec: proof of bigvee to bigcup}
Lemma \ref{lma: from bigvee to bigcup} is the following.
$$
\text{\normalfont{WPO}}\llbracket  \bigvee_{i \in I}  P ,\mathbb{C},\epsilon \rrbracket  \iff 
\bigcup_{i \in I} \text{\normalfont{WPO}}\llbracket  P,\mathbb{C},\epsilon\rrbracket 
$$
\begin{proof}
    \begin{align*}
\text{\normalfont{WPO}}\llbracket  \bigvee_{i \in I} P ,\mathbb{C},\epsilon \rrbracket  &= 
\{ \sigma' \mid \exists \sigma . \sigma \models 
\left(\bigvee_{i \in I} P  \right) \land 
(\sigma, \sigma') \in \llbracket  \mathbb{C}  \rrbracket _\epsilon  \} \\
&= 
\{ \sigma' \mid \exists \sigma .  \bigvee_{i \in I}  \left( \sigma \models 
   P \right) \land 
(\sigma, \sigma') \in \llbracket  \mathbb{C}  \rrbracket _\epsilon  \} \\
&= 
\{ \sigma' \mid \exists \sigma . \bigvee_{i \in I} ( \sigma \models 
 P \land 
(\sigma, \sigma') \in \llbracket  \mathbb{C}  \rrbracket _\epsilon ) \}\\
&= 
\{ \sigma' \mid \bigvee_{i \in I} \exists \sigma . ( \sigma \models 
P \land 
(\sigma, \sigma') \in \llbracket  \mathbb{C}  \rrbracket _\epsilon ) \}\\
&= 
\bigcup_{i \in I} \text{\normalfont{WPO}}\llbracket  P ,\mathbb{C},\epsilon\rrbracket 
\end{align*}
\end{proof}

\subsection{Proof of Lemma \ref{lma: expressiveness of atomic command}}
\label{appendix's subsec: proof of expressiveness of atomic command}
Lemma \ref{lma: expressiveness of atomic command} is the following.

$$ \forall \sigma' . 
\sigma' \in \text{\normalfont{WPO}}\llbracket  \psi,\mathbb{C},\epsilon\rrbracket 
\iff
\sigma' \models  \normalfont \textsf{wpo}_\text{sh} (\psi, \mathbb{C}, \epsilon) $$

Now, we prove Lemma \ref{lma: expressiveness of atomic command}.

\subsubsection{Case {\normalfont\texttt{skip}}}
\begin{itemize}
\item $\epsilon = \textit{ok}$

By the definition of $\llbracket \texttt{skip}\rrbracket _\textit{ok}$,
$\text{\normalfont{WPO}}\llbracket  \psi ,\texttt{skip},\textit{ok}\rrbracket  = \{ \sigma' \mid \exists \sigma. 
\sigma \models    \psi  \land
(\sigma,\sigma') \in \llbracket \texttt{skip}\rrbracket _\textit{ok} \}
= \{ \sigma \mid \sigma \models   \psi \} $ and
$ \textsf{wpo}(\psi ,\texttt{skip},\textit{ok}) = \psi$ hold.
Consequently, Lemma \ref{lma: expressiveness of atomic command} holds in this case.

\item $\epsilon = \textit{er}$

By the definition of $\llbracket \texttt{skip}\rrbracket _\textit{er}$,
$\text{\normalfont{WPO}}\llbracket    \psi ,\texttt{skip},\textit{er}\rrbracket  =
\{ \sigma' \mid \exists \sigma. 
\sigma \models   \psi  \land
(\sigma,\sigma') \in \llbracket \texttt{skip}\rrbracket _\textit{er} \} = \emptyset$ and
$ \textsf{wpo}( \psi ,\texttt{skip},\textit{er})  = \textsf{false}$ hold.
Hence, this lemma holds straightforwardly in this case.
\end{itemize}

\subsubsection{Case {\normalfont$\normalfont x := t$}}
\begin{itemize}
\item $\epsilon = \textit{ok}$

Firstly, we prove $\forall\sigma'. \sigma' \in \normalfont
\text{\normalfont{WPO}}\llbracket    \psi , x := t , \textit{ok}\rrbracket  \Longrightarrow 
\normalfont \sigma' \models  
\textsf{wpo}_\text{sh}( \psi , x := t , \textit{ok})$.
By the definition of $\llbracket x := t\rrbracket _\textit{ok}$,
$\text{\normalfont{WPO}}\llbracket   \psi,x := t,\textit{ok}\rrbracket  = 
\{ \sigma' \mid \exists \sigma. 
\sigma \models  \psi \land
(\sigma,\sigma') \in \llbracket x := t\rrbracket _\textit{ok} \} = 
\{ (s[x \mapsto s(t)], h ) \mid  (s,h) \models  \psi \}
$ holds.
By the definition of $\textsf{wpo}_\text{sh}$,
$  \textsf{wpo}_\text{sh}( \psi , x := t , \textit{ok}) = 
\exists  x'.  \psi [x := x']  * x \approx ( t  [x := x'] )$ holds.
It is enough to show that 
$(s,h) \models  \psi \Longrightarrow ( s[x \mapsto s(t) ], h ) \models 
\exists  x'. \psi [x := x'] * x \approx (t [x := x'])$ holds.
We show the following to prove our statement.
\begin{align*}
& (s[x \mapsto s(t)], h ) \models 
\exists   x'.  \psi [x := x']  * x \approx (t  [x := x'] ) & \\
\Leftrightarrow \ & 
(s,h) \models 
( \exists   x'.  \psi [x := x']  * x \approx (t  [x := x'] )) [x := t]
& \text{Lemma \ref{lma: Substitution for assignment}} \\
\Leftrightarrow \ & 
(s,h) \models 
\exists  x'.  \psi [x := x'] [x := t]  * x [x := t] \approx (t  [x := x'] [x := t])
& \\
\Leftrightarrow \ & 
(s,h) \models 
\exists   x'. \psi [x := x']  * t \approx (t  [x := x'] )
& \\
\Leftrightarrow \ & 
\exists v \in \textsc{Val} .  (s [x' \mapsto v] ,h) \models 
\psi [x := x']  * t \approx (t  [x := x'] )
& \\
\Leftarrow \ & 
  (s [x' \mapsto s(x)] ,h) \models \psi [x := x']  * t \approx (t  [x := x'] )
& \\
& \text{// picking $s(x)$ as $v$} & \\
\Leftrightarrow \ & 
  (s  ,h) \models 
 \psi [x := x'] [x' := x] * t [x' := x] \approx (t  [x := x'][x' := x] )
& \text{Lemma \ref{lma: Substitution for assignment}}\\
\Leftrightarrow \ & 
  (s  ,h) \models 
\psi  * t  \approx t  \  \text{ // Since $x' \notin \textsf{fv}(t) \cup \textsf{fv}(\psi)$}
& \\
\Leftrightarrow \ & (s,h) \models  \psi  \ \text{ // this is the premise} &
\end{align*}

Next, we prove the reverse direction: 
$\forall\sigma'. \sigma' \in \normalfont
\text{\normalfont{WPO}}\llbracket \psi , x := t , \textit{ok}\rrbracket  \Longleftarrow 
\normalfont \sigma' \models \textsf{wpo}_\text{sh}(
\psi , x := t , \textit{ok})$.
$(s',h') \in \text{WPO} \llbracket  \psi, x := t, \textit{ok}  \rrbracket $ can be described as
$\exists (s,h) \models \psi . (s', h') = (s [x \mapsto s(t)] , h)$.
Consequently, it is enough to show that 
$(s', h') \models \exists x'.  \psi [x := x']  * x \approx (t  [x := x'] ) \Longrightarrow
\exists (s,h) \models \psi . (s', h') = (s [x \mapsto s(t)] , h)$ holds.
We show the following to prove our statement.
\begin{align*}
& (s', h') \models \exists x'.  \psi [x := x']  * x \approx (t  [x := x'] ) & \\
\Leftrightarrow \ &  
\exists v \in \textsc{Val} . (s' [x' \mapsto v] , h') \models 
 \psi [x := x']  *  x \approx (t  [x := x'] ) & \\
\Rightarrow \ &  
\exists v \in \textsc{Val} . 
(s' [x' \mapsto v] , h') \models 
 \psi [x := x']  \text{ and }
 s' [x' \mapsto v](x) = s' [x' \mapsto v] (t  [x := x']) & \\
\Leftrightarrow \ &  
\exists v \in \textsc{Val} . 
(s' [x' \mapsto v] [x \mapsto s' [x' \mapsto v] (x') ], h') \models 
 \psi  \text{ and } \ \text{ // Lemma \ref{lma: Substitution for assignment}}\\
 & \  s' [x' \mapsto v](x) = s' [x' \mapsto v] (t  [x := x'])  \\
\Leftrightarrow \ &  
\exists v \in \textsc{Val} . 
(s' [x \mapsto v ], h') \models 
 \psi  \text{ and } \\
 & \ s' (x) = s' [x' \mapsto v] [x \mapsto s' [x' \mapsto v](x')] (t ) \ 
 \text{ // Lemma \ref{lma: Substitution for assignment}} \\
\Leftrightarrow \ &  
\exists v \in \textsc{Val} . 
(s' [x \mapsto v ], h') \models 
 \psi  \text{ and }
 s' (x) = s' [x \mapsto v] (t ) \ \text{ // Since $x'$ is fresh} & \\
\Rightarrow \ &  
\exists v \in \textsc{Val} . 
(s_\psi , h_\psi ) \models 
 \psi  \text{ and }
 s' (x) = s_\psi (t )  & \\
& \text{// Let us say  $s_\psi = s' [x \mapsto v]$   and  $h_\psi = h'$} & \\
\Rightarrow \ &  
\exists (s, h) \models 
 \psi  .
 (s', h') = (s [x \mapsto s(t)] , h)  & 
\end{align*}

\item $\epsilon = \textit{er}$

By the definition of $\llbracket  x := t\rrbracket _\textit{er}$,
$\text{\normalfont{WPO}}\llbracket \psi, x := t,\textit{er}\rrbracket  = \{ \sigma' \mid \exists \sigma. 
\sigma \models \psi \land
(\sigma,\sigma') \in \llbracket   x := t \rrbracket _\textit{er} \} = \emptyset$ holds.
Since $\textsf{wpo}_\text{sh}(\psi, x := t,\textit{er}) = \textsf{false}$,
Lemma \ref{lma: expressiveness of atomic command} holds straightforwardly in this case.
\end{itemize}

\subsubsection{Case $x := \texttt{*}$}  
\begin{itemize}
\item $\epsilon = \textit{ok}$

Firstly, we prove $\forall\sigma'. \sigma' \in \normalfont
\text{\normalfont{WPO}}\llbracket \psi , x := \texttt{*} , \textit{ok}\rrbracket  \Longrightarrow 
\normalfont \sigma' \models 
\textsf{wpo}_\text{sh}(\psi , x := \texttt{*} , \textit{ok})$.
By the definition of $\llbracket x := \texttt{*}\rrbracket _\textit{ok}$,
$\text{\normalfont{WPO}}\llbracket \psi,x := \texttt{*},\textit{ok}\rrbracket  = 
\{ \sigma' \mid \exists \sigma. 
\sigma \models \psi \land
(\sigma,\sigma') \in \llbracket x := \texttt{*}\rrbracket _\textit{ok} \} = 
\{ (s[x \mapsto v], h ) \mid  (s,h) \models \psi \land v \in \textsc{Val} \}
$ holds.
By the definition of $\textsf{wpo}_\text{sh}$,
$\textsf{wpo}_\text{sh}(\psi , x := \texttt{*} , \textit{ok}) = 
\exists x'.  \psi [x := x']  $ holds.
It is enough to show that if
$(s,h) \models \psi  \land v \in \textsc{Val} $, then $
 ( s[x \mapsto v ], h ) \models 
\exists x'. \psi [x := x'] $ holds.
We show the following to prove our statement.
\begin{align*}
& (s[x \mapsto v], h ) \models 
\exists x'.  \psi [x := x']   \\
\Leftrightarrow \ & 
\exists w \in \textsc{Val} .
(s[x \mapsto v] [x' \mapsto w] , h ) \models   \psi [x := x'] 
 \\
\Leftrightarrow \ & 
\exists w \in \textsc{Val} .
(s  [x' \mapsto w] , h ) \models   \psi [x := x'] \ \text{ // Since $\psi [x := x']$ is $x$-free} \\
\Leftarrow \ & 
(s  [x' \mapsto s(x)] , h ) \models   \psi [x := x'] \ 
\text{ // Pick $s(x)$ as $w$} \\
\Leftrightarrow \ & 
(s , h ) \models   \psi [x := x'] [x' := x] \ 
\text{ // Lemma \ref{lma: Substitution for assignment}} \\
\Leftrightarrow \ & 
(s , h ) \models   \psi  \\
\Leftarrow \ & 
(s , h ) \models   \psi \land v \in \textsc{Val}
\end{align*}

Next, we prove the reverse direction: 
$\forall\sigma'. \sigma' \in \normalfont
\text{\normalfont{WPO}}\llbracket \psi , x := \texttt{*} , \textit{ok}\rrbracket  \Longleftarrow 
\normalfont \sigma' \models \textsf{wpo}_\text{sh}(
\psi , x := \texttt{*} , \textit{ok})$.
$(s',h') \in \text{WPO} \llbracket  \psi, x := \texttt{*}, \textit{ok}  \rrbracket $ can be described as
$\exists (s,h) \models \psi , \exists v \in \textsc{Val} .
(s', h') = (s [x \mapsto v] , h)$.
Consequently, it is enough to show that 
$(s', h') \models \exists x'.  \psi [x := x'] \Longrightarrow
\exists (s,h) \models \psi , \exists v \in \textsc{Val}  . (s', h') = (s [x \mapsto v] , h)$ holds.
We show the following to prove our statement.
\begin{align*}
& (s', h') \models \exists x'.  \psi [x := x']   \\
\Leftrightarrow \ &
\exists w \in \textsc{Val} . 
(s' [x' \mapsto w] , h') \models   \psi [x := x']   \\
\Leftrightarrow \ &
\exists w \in \textsc{Val} . 
(s' [x' \mapsto w] [x \mapsto s' [x' \mapsto w](x')] , h')
\models   \psi \ \text{ // Lemma \ref{lma: Substitution for assignment}}  \\
\Leftrightarrow \ &
\exists w \in \textsc{Val} . 
(s' [x' \mapsto w] [x \mapsto w] , h')
\models   \psi  \\
\Leftrightarrow \ &
\exists w \in \textsc{Val} . 
(s' [x \mapsto w] , h')
\models   \psi  \ \text{ // Since $x'$ is fresh} \\
\Rightarrow \ &
\exists w \in \textsc{Val} . 
(s_\psi , h_\psi)
\models   \psi \\
& \text{ // Let us say $s_\psi = s' [x \mapsto w]$ and $h_\psi = h'$} \\
& \text{ // It is evident that $s_\psi [x \mapsto s'(x) ] = s'$ and $s'(x) \in \textsc{Val}$} \\
\Rightarrow \ & 
\exists (s,h) \models \psi , v \in \textsc{Val}  . (s', h') = (s [x \mapsto v] , h) 
\end{align*}

    \item $\epsilon = \text{\textit{er}}$
    
    Similar to \textbf{Case} $x := t$.
\end{itemize}

\subsubsection{Case {\normalfont\texttt{{assume($ B$)}}}}
\begin{itemize}
\item $\epsilon = \textit{ok}$

By the definition of $\llbracket {\normalfont\texttt{{assume($\normalfont B$)}}}\rrbracket _\textit{ok}$,
$\text{\normalfont{WPO}}\llbracket \psi,{\normalfont\texttt{{assume($\normalfont B$)}}},\textit{ok}\rrbracket  = \{ \sigma' \mid \exists \sigma. 
\sigma \models \psi \land
(\sigma,\sigma') \in \llbracket {\normalfont\texttt{{assume($\normalfont B$)}}}\rrbracket _\textit{ok} \} = 
\{ \sigma' \mid \exists \sigma. 
\sigma \models \psi \land
\sigma = \sigma' \land
\sigma = (s,h) \land
s(B) \neq 0\} = 
\{ \sigma \mid \sigma \models \psi \textit{ and } \sigma \models B \} = 
\{ \sigma \mid \sigma \models \psi * B\}$.
$\sigma \models \psi * B \iff \sigma \models \psi \textit{ and } \sigma \models B $ holds since $B$ is not heap formula.
Additionally, by the defintion of $\textsf{wpo}_\text{sh}$,
$\textsf{wpo}_\text{sh}(\psi,{\normalfont\texttt{{assume($\normalfont B$)}}},\textit{ok}) = \psi * B$ holds. 
Consequently, Lemma \ref{lma: expressiveness of atomic command} holds in this case.

\item $\epsilon = \textit{er}$

By the definition of $\llbracket \texttt{{assume($\normalfont B$)}}\rrbracket _\textit{er}$, $\text{\normalfont{WPO}}\llbracket \psi ,\texttt{{assume($\normalfont B$)}},\textit{er}\rrbracket  = \{ \sigma' \mid \exists \sigma. 
\sigma \models \psi  \land
(\sigma,\sigma') \in \llbracket \texttt{{assume($\normalfont B$)}}\rrbracket _\textit{er} \} 
= \emptyset$ holds.
Since
$\textsf{wpo}(\psi ,\texttt{{assume($\normalfont B$)}},\textit{er}) = \textsf{false}$,
Lemma \ref{lma: expressiveness of atomic command} holds straightforwardly in this case.
\end{itemize}

\subsubsection{Case $x := \texttt{alloc()}$} 
\begin{itemize}
\item $\epsilon = \textit{ok}$

By the definition of $\llbracket  x := \texttt{alloc()}  \rrbracket _\textit{ok}$,    
$\text{\normalfont{WPO}}\llbracket \psi, x := \texttt{alloc()} ,\textit{ok}\rrbracket  
= \{ \sigma' \mid \exists \sigma. 
\sigma \models \psi \land
(\sigma,\sigma') \in \llbracket  x := \texttt{alloc()}  \rrbracket _\textit{ok} \} = 
\{ (s [x \mapsto l], h [l \mapsto v] ) \mid 
(s,h) \models \psi \land (l \notin \textsf{dom}(h) \lor h(l) = \bot )
\land v \in \textsc{Val}
\}$.

By the definition of $\textsf{wpo}_\text{sh}$,
$\textsf{wpo}_\text{sh} (\psi,  x := \texttt{alloc()} ,\textit{ok}) = 
\exists x'.  (  \psi[x := x'] * x \mapsto - ) \lor 
\bigvee_{j=1}^n ( 
\left(\bigast_{i=1}^n y_i \not\mapsto \right) 
[y_j \not\mapsto := y_j \mapsto -  ] *
x \approx y_j  * \psi' [x := x'] )$, where 
$\psi = \left(\bigast_{i=1}^n y_i \not\mapsto \right) * \psi'$ (here, $\psi'$ does not have atoms with $\not\mapsto$).

Firstly, we prove $\forall\sigma'. \sigma' \in \normalfont
\text{\normalfont{WPO}}\llbracket \psi , x := \texttt{alloc()} , \textit{ok}\rrbracket  
\Longrightarrow 
\normalfont \sigma' \models 
\textsf{wpo}_\text{sh}(\psi , x := \texttt{alloc()} , \textit{ok})$.

(1) We show that if $  (s,h) \models \psi  \text{ and }
l \notin \textsf{dom}(h) \text{ and } v \in \textsc{Val}
$, then
$ (s [x \mapsto l], h [l \mapsto v] ) \models 
\exists x'.  \psi[x := x'] * x \mapsto - $ holds.

By the following, this case is true.
\begin{align*}
& (s [x \mapsto l], h [l \mapsto v] ) \models 
\exists x'.  \psi[x := x'] * x \mapsto - \\ 
\Leftrightarrow \ & 
\exists w \in \textsc{Val} . 
(s [x \mapsto l] [x' \mapsto w] , h [l \mapsto v] ) \models 
 \psi[x := x'] * x \mapsto - \\
\Leftarrow \ & 
(s [x \mapsto l] [x' \mapsto s(x)] , h [l \mapsto v] ) \models 
 \psi[x := x'] * x \mapsto - \\
& \text{// Pick $s(x)$ as $w$} \\
\Leftarrow \ & 
(s ' , h \circ (l \mapsto v) ) \models 
 \psi[x := x'] * x \mapsto - \\
& \text{// Let us say $s' = s [x \mapsto l] [x' \mapsto s(x)] $ and $l \notin \textsf{dom}(h)$} \\
\Leftrightarrow \ & 
(s ' , h  ) \models 
 \psi[x := x'] 
\text{ and }
(s ' ,  (l \mapsto v) ) \models 
 x \mapsto -  \ \text{ // Since $s'(x) = l$} \\
\Leftrightarrow \ & 
(s ' [x \mapsto s' (x') ] , h  ) \models   \psi  
\text{ and  \ // By Lemma \ref{lma: Substitution for assignment}} \\
& (l \mapsto v) (s'(x)) = (l \mapsto v) (l) = v
\text{ \ // $s'(x) = l$ }  \\
\Leftrightarrow \ & 
(s  [x \mapsto s (x) ] , h  ) \models   \psi  \text{ and } (l \mapsto v) (s'(x)) = (l \mapsto v) (l) = v\\
& \text{// Since 
$s [x \mapsto l] [x' \mapsto s(x)] [x \mapsto s' (x') ] = 
s [x' \mapsto s(x)] [x \mapsto s' (x') ]$} \\
& \text{// $s' (x') = s(x)$ and $x' \neq x$ and $x' \notin \textsf{fv}(\psi)$} \\
\Leftrightarrow \ & 
(s   , h  ) \models   \psi  
\text{ and }  (l \mapsto v) (s'(x)) = (l \mapsto v) (l) = v \\
\Leftarrow \ &
(s,h) \models \psi  \text{ and }
l \notin \textsf{dom}(h) \text{ and } v \in \textsc{Val}
\end{align*}

(2) We show that if $(s,h) \models \psi$  and
$h(l) = \bot$ and
$v \in \textsc{Val}$, then the following holds.
$$(s [x \mapsto l], h [l \mapsto v] ) \models \exists x'. 
\bigvee_{j=1}^n ( 
\left(\bigast_{i=1}^n y_i \not\mapsto \right) 
[y_j \not\mapsto := y_j \mapsto -  ] *
x \approx y_j  * \psi' [x := x'] )$$
Here, $(s,h) \models \psi  \text{ and }
h(l) = \bot \land v \in \textsc{Val} \iff 
(s,h) \models \psi  \text{ and }
h(l) = \bot \land v \in \textsc{Val} \text{ and }
l = s(y_i) \text{ (for some $i$)}
$ holds.

We only show the case for the $j$-th disjunct ($1 \leq j \leq n$). 
Specifically, we aim to show that if 
$(s,h) \models \psi  \text{ and }
h(l) = \bot  \text{ and } v \in \textsc{Val} \text{ and }  l = s(y_j)$, then 
$(s [x \mapsto l], h [l \mapsto v] ) \models \exists x'. 
\left(\bigast_{i=1}^n y_i \not\mapsto \right) 
[y_j \not\mapsto := y_j \mapsto -  ] *
x \approx y_j  * \psi' [x := x'] $
 holds through the following. 
\begin{align*}
& (s [x \mapsto l], h [l \mapsto v] ) \models 
\exists x'. 
\left(\bigast_{i=1}^n y_i \not\mapsto \right) 
[y_j \not\mapsto := y_j \mapsto -  ] *
x \approx y_j  * \psi' [x := x']     \\
\Leftrightarrow \ &
\exists w \in \textsc{Val} . 
(s [x \mapsto l] [x' \mapsto w ] , h [l \mapsto v] ) \models 
\left(\bigast_{i=1}^n y_i \not\mapsto \right) 
[y_j \not\mapsto := y_j \mapsto -  ] *
x \approx y_j  * \psi' [x := x']     \\
\Leftarrow \ &
(s [x \mapsto l] [x' \mapsto s(x) ] , h [l \mapsto v] ) \models 
\left(\bigast_{i=1}^n y_i \not\mapsto \right) 
[y_j \not\mapsto := y_j \mapsto -  ] *
x \approx y_j  * \psi' [x := x']     \\
& \text{// We pick   $s(x)$ as $w$}  \\
\Leftarrow \ &
(s [x \mapsto l] [x' \mapsto s(x) ] ,
h [l \mapsto v] [s [x \mapsto l] [x' \mapsto s(x) ](y_j) \mapsto \bot ]  ) \models 
\left(\bigast_{i=1}^n y_i \not\mapsto \right)  *
x \approx y_j  * \psi' [x := x']     \\
\Leftrightarrow \ &
(s [x \mapsto l] [x' \mapsto s(x) ] ,
h [l \mapsto v] [s (y_j) \mapsto \bot ]  ) \models 
\left(\bigast_{i=1}^n y_i \not\mapsto \right)  *
x \approx y_j  * \psi' [x := x']     \\
\Leftrightarrow \ &
(s [x \mapsto l] [x' \mapsto s(x) ] ,
h [l \mapsto v] [s (y_j) \mapsto \bot ]  ) \models 
\left(\bigast_{i=1}^n y_i \not\mapsto \right)  *  \psi' [x := x']   \\
&\text{and } (s [x \mapsto l] [x' \mapsto s(x) ] , h_0  ) \models x \approx y_j      \\
& \text{// Here, $h_0$ means empty heap, i.e., $\textsf{dom} (h_0) = \emptyset $} \\
\Leftrightarrow \ &
(s [x \mapsto l] [x' \mapsto s(x) ] [x \mapsto s [x \mapsto l] [x' \mapsto s(x) ](x')] ,
h [l \mapsto v] [s (y_j) \mapsto \bot ]  ) \models 
\left(\bigast_{i=1}^n y_i \not\mapsto \right)  *  \psi'    \\
&\text{// Lemma \ref{lma: Substitution for assignment}} \\
&\text{and } l = s(y_j)  \text{ and }  \textsf{dom} (h_0) = \emptyset    \\
\Leftrightarrow \ &
(s  [x' \mapsto s(x) ] [x \mapsto s (x)] ,
h [l \mapsto v] [s (y_j) \mapsto \bot ]  ) \models  \psi  \text{ and } l = s(y_j)  \text{ and }  \textsf{dom} (h_0) = \emptyset    \\
\Leftrightarrow \ &
(s   ,
h [s (y_j) \mapsto v] [s (y_j) \mapsto \bot ]  ) \models  \psi  \text{ and } l = s(y_j)  \text{ and }  \textsf{dom} (h_0) = \emptyset    \\
&\text{// By $x' \notin \textsf{fv}(\psi)$ and Lemma \ref{lma: Substitution for assignment}} \\
\Leftarrow \ &
(s   ,h   ) \models  \psi  \text{ and } l = s(y_j)  \text{ and }  \textsf{dom} (h_0) = \emptyset    \\
&\text{// Since $h [s (y_j) \mapsto \bot ] (s (y_j)) = h(s (y_j)) = \bot$ if
$s,h \models \psi$ holds} \\
\Leftarrow \ & 
(s,h) \models \psi  \text{ and }
h(l) = \bot  \text{ and } v \in \textsc{Val} \text{ and }  l = s(y_j)
\end{align*}

Next, we prove $\forall\sigma'. \sigma' \in \normalfont
\text{\normalfont{WPO}}\llbracket \psi , x := \texttt{alloc()} , \textit{ok}\rrbracket  \Longleftarrow 
\normalfont \sigma' \models 
\textsf{wpo}_\text{sh}(\psi , x := \texttt{alloc()} , \textit{ok})$.

(3) We show the case for $(s',h') \models \exists x'.  \psi[x := x'] * x \mapsto -$.
The goal is showing that 
if $(s',h') \models \exists x' . 
 \psi [x:=x'] * x\mapsto - $, then
$\exists ( s,h  ) \models \psi , l \notin \textsf{dom}(h) , v \in \textsc{Val} .
(s', h') = (s[x \mapsto l], h[l \mapsto v]) $ holds.
We show the following to prove our statement.
\begin{align*}
    & (s',h') \models \exists x' . 
\psi [x:=x'] * x\mapsto - \\
\Leftrightarrow \ & 
\exists v, w \in \textsc{Val} . 
(s' [x' \mapsto w] ,h') \models 
\psi [x:=x'] * x\mapsto - \text{ and } h' (s' [x' \mapsto w] (x)) =h'(s'(x)) = v \\
\Rightarrow \ & 
\exists v, w \in \textsc{Val} . 
(s' [x' \mapsto w] ,h_\psi \circ ( s' [x' \mapsto w] (x) \mapsto v ) ) \models 
 \psi [x:=x'] * x\mapsto - \text{ and } h'(s'(x)) = v \\
& \text{ // Let us say that $h' = h_\psi \circ (s' [x' \mapsto w] (x) \mapsto v) = 
h_\psi \circ ( s'  (x) \mapsto v)$} \\
\Rightarrow \ & 
\exists v, w \in \textsc{Val} . 
(s' [x' \mapsto w] ,h_\psi  ) \models  \psi [x:=x'] 
\text{ and } \\
& (s' [x' \mapsto w] ,( s' (x) \mapsto v ) ) \models 
  x\mapsto - \text{ and } h'(s'(x)) = v \\
\Leftrightarrow \ & 
\exists v, w \in \textsc{Val} . 
(s' [x' \mapsto w] [x \mapsto s' [x' \mapsto w](x') ] ,h_\psi  ) \models  \psi 
\text{ and } \ \text{ // Lemma \ref{lma: Substitution for assignment}} \\
& (s' [x' \mapsto w] ,( s' (x) \mapsto v ) ) \models 
  x\mapsto - \text{ and } h'(s'(x)) = v  \\
\Leftrightarrow \ & 
\exists v, w \in \textsc{Val} . 
(s'  [x \mapsto w ] ,h_\psi  ) \models  \psi 
\text{ and }  (s'  ,( s'  (x) \mapsto v ) ) \models   x\mapsto - 
\text{ and } h'(s'(x)) = v \\
&  \text{// Since $x'$ is fresh}\\
\Rightarrow \ & 
\exists v, w \in \textsc{Val} . 
( s_\psi ,h_\psi  ) \models  \psi 
\text{ and }  (s'  , ( s'  (x) \mapsto v ) ) \models 
  x\mapsto -  \text{ and } h'(s'(x)) = v \\
& \text{ // Let us say that $s_\psi = s'  [x \mapsto w ]$}\\
\Rightarrow \ & 
\exists v \in \textsc{Val} . 
( s_\psi ,h_\psi  ) \models  \psi 
\text{ and } s'(x) \notin \textsf{dom} (h_\psi) 
\text{ and } h' = h_\psi \circ (s'(x) \mapsto v)   \text{ and } h'(s'(x)) = v  \\
\Rightarrow \ & 
\exists ( s,h  ) \models \psi , l \notin \textsf{dom}(h) , v \in \textsc{Val} .
(s', h') = (s[x \mapsto l], h[l \mapsto v]) \\
& \text{// Let us say that $s'(x) = l$} \\
&\text{// By $ s_\psi [x \mapsto l] = s' [x \mapsto w] [x \mapsto l] = 
s' [x \mapsto l ] = s'$} \\
&\text{// By $ h' = h_\psi \circ (s'(x) \mapsto v) = h_\psi [l \mapsto v]  $}
\end{align*}

(4) We show the case for $(s',h') \models \exists x' . 
\left(\bigvee_{i=1}^n 
(
\left(\bigast_{i=1}^n y_i \not\mapsto \right) 
[y_j \not\mapsto := y_j \mapsto -  ] *
x \approx y_j  * \psi' [x := x']  )  
\right)$.
The goal is showing that 
if $(s',h') \models \exists x' . 
\left(\bigvee_{i=1}^n 
(
\left(\bigast_{i=1}^n y_i \not\mapsto \right) 
[y_j \not\mapsto := y_j \mapsto -  ] *
x \approx y_j  * \psi' [x := x']  )  
\right)$, then
$ \exists ( s,h  ) \models \psi , h(l) = \bot , v \in \textsc{Val} , l=s(y_j) .
(s', h') = (s[x \mapsto l], h[l \mapsto v])
$ holds.
We show the following to prove our statement.
Here, we only show a case for $j$-th ($1 \leq j \leq n$) disjunct, i.e.,
$\exists x' . 
\left(\bigast_{i=1}^n y_i \not\mapsto \right) 
[y_j \not\mapsto := y_j \mapsto -  ] *
x \approx y_j  * \psi' [x := x']    $.
\begin{align*}
& (s',h') \models 
\exists x' . 
\left(\bigast_{i=1}^n y_i \not\mapsto \right) 
[y_j \not\mapsto := y_j \mapsto -  ] *
x \approx y_j  * \psi' [x := x']
\\
\Leftrightarrow \ & 
\exists v, w \in \textsc{Val} . 
(s' [x' \mapsto w] ,h') \models 
\left(\bigast_{i=1}^n y_i \not\mapsto \right) 
[y_j \not\mapsto := y_j \mapsto -  ] *
x \approx y_j  * \psi' [x := x'] \\
& \text{ and }  h'(s' [x' \mapsto w](y_j)) = h'(s'(y_j)) = v \\
\Rightarrow \ & 
\exists v, w \in \textsc{Val} . 
(s' [x' \mapsto w] ,h' [ s' [x' \mapsto w] (y_j) \mapsto \bot  ] ) \models 
\left(\bigast_{i=1}^n y_i \not\mapsto \right)  *
x \approx y_j  * \psi' [x := x'] 
\text{ and }   h'(s'(y_j)) = v \\
\Leftrightarrow \ & 
\exists v, w \in \textsc{Val} . 
(s' [x' \mapsto w] ,h' [ s'  (y_j) \mapsto \bot  ] ) \models 
\left(\bigast_{i=1}^n y_i \not\mapsto \right)  *
x \approx y_j  * \psi' [x := x']  \text{ and }   h'(s'(y_j)) = v  \\
\Rightarrow \ & 
\exists v, w \in \textsc{Val} . 
s' [x' \mapsto w] (x) = s' [x' \mapsto w] (y_j) = 
s'(x) = s'(y_j)
\text{ and }   h'(s'(y_j)) = v \text{ and }  \\
& (s' [x' \mapsto w] ,h' [ s'  (y_j) \mapsto \bot  ] ) \models 
\left(\bigast_{i=1}^n y_i \not\mapsto \right)  *  \psi' [x := x'] \\
\Leftrightarrow \ & 
\exists v, w \in \textsc{Val} . 
s' [x' \mapsto w] (x) = s' [x' \mapsto w] (y_j) = 
s'(x) = s'(y_j)
\text{ and }   h'(s'(y_j)) = v \text{ and }  \\
& (s' [x' \mapsto w] ,h' [ s'  (y_j) \mapsto \bot  ] ) \models 
\left( \left(\bigast_{i=1}^n y_i \not\mapsto \right)  *  \psi' \right) [x := x'] \ \text{ // Since $x \notin \textsf{fv} \left(\bigast_{i=1}^n y_i \not\mapsto \right) $} \\
\Leftrightarrow \ & 
s'(x) = s'(y_j)
\text{ and } \exists v \in \textsc{Val} . h'(s'(y_j)) = v \text{ and } \\
&\exists w \in \textsc{Val} . 
  (s' [x' \mapsto w] [x \mapsto s' [x' \mapsto w](x') ]
,h' [ s'  (y_j) \mapsto \bot  ] ) \models 
\left(\bigast_{i=1}^n y_i \not\mapsto \right)  *  \psi'  \\
&\text{// Lemma \ref{lma: Substitution for assignment}} \\ 
\Leftrightarrow \ & 
s'(x) = s'(y_j)
\text{ and } \exists v \in \textsc{Val} . h'(s'(y_j)) = v \text{ and } \\
&\exists w \in \textsc{Val} . 
  (s'  [x \mapsto w ]
,h' [ s'  (y_j) \mapsto \bot  ] ) \models 
\left(\bigast_{i=1}^n y_i \not\mapsto \right)  *  \psi' \ 
\text{ // $x'$ is fresh}\\
\Rightarrow \ &
s'(x) = s'(y_j) \text{ and } \exists v \in \textsc{Val} . h'(s'(y_j)) = v 
\text{ and }  (s_\psi, h_\psi)  \models \psi  \\ 
& \text{ // Let us say that $s_\psi = s'[x \mapsto w]$ and 
$h_\psi = h'[s'(y_j) \mapsto \bot ] = h'[s_\psi(y_j) \mapsto \bot ] $}\\
\Rightarrow \ &
\exists ( s,h  ) \models \psi , h(l) = \bot , v \in \textsc{Val} , l=s(y_j) .
(s', h') = (s[x \mapsto l], h[l \mapsto v])\\
& \text{// By $s'(x) = s'(y_j) = s_\psi (y_j)  = l $ and 
$\exists v \in \textsc{Val} . h'(s'(y_j)) = h'(l) = v  $}
\end{align*}

Consequently,  Lemma \ref{lma: expressiveness of atomic command} holds in this case.

\item $\epsilon = \textit{er}$

By the definition of $\llbracket    x := \texttt{alloc()} \rrbracket _\textit{er} $, $\text{\normalfont{WPO}}\llbracket \psi,  x := \texttt{alloc()} ,\textit{er}\rrbracket  = \{ \sigma' \mid \exists \sigma. 
\sigma \models \psi \land
(\sigma,\sigma') \in \llbracket    x := \texttt{alloc()} \rrbracket _\textit{er} \} = \emptyset$.
Since 
$\textsf{wpo}(\psi, x := \texttt{alloc()} ,\textit{er}) = \textsf{false}$, 
Lemma \ref{lma: expressiveness of atomic command} holds straightforwardly in this case.

\end{itemize}


\subsubsection{Case \texttt{free($x$)}}

\begin{itemize}
\item $\epsilon = \textit{ok}$ 
\begin{itemize}
    \item If $y \in \textsf{Aliases}(x, \psi)$ and 
$\psi = \psi'  * y \mapsto t $

Firstly, we prove $\forall\sigma'. \sigma' \in \normalfont
\text{\normalfont{WPO}}\llbracket \psi , \texttt{free($x$)} , \textit{ok}\rrbracket  \Longrightarrow 
\normalfont \sigma' \models 
\textsf{wpo}_\text{sh}(\psi , \texttt{free($x$)} , \textit{ok})$.
By the definition of 
$\llbracket   \texttt{free($x$)}  \rrbracket _\textit{ok}$,
$\text{\normalfont{WPO}}\llbracket \psi, \texttt{free($x$)} ,\textit{ok}\rrbracket  = \{ \sigma' \mid \exists \sigma. 
\sigma \models \psi \text{ and }
(\sigma,\sigma') \in \llbracket   \texttt{free($x$)}  \rrbracket _\textit{ok}
\} = 
\{ (s, h [s(x) \mapsto \bot] ) \mid (s,h) \models \psi \text{ and } s(x) \in \textsf{dom}_+ (h)  \}$ holds. 
Additionally, by the definition of $\textsf{wpo}_\text{sh}$,
$\textsf{wpo}_\text{sh} (\psi, \texttt{free($x$)},\textit{ok}) = 
\psi' *  y \not\mapsto $ holds. 
It is enough to show that if $(s,h) \models \psi \text{ and } s(x) \in \textsf{dom}_+ (h)$, 
then $(s, h [s(x) \mapsto \bot]) \models \psi'  * y \not\mapsto$.
We show the following to prove our statement.

\begin{align*}
& (s,h) \models \psi \text{ and }  s(x) \in \textsf{dom}_+ (h)  \\
\Rightarrow \ &  (s,h) \models \psi \\
\Leftrightarrow \ & 
 (s, h[s(x) \mapsto s(t)]) \models \psi'  * y \mapsto t \\
 & \text{// Since $y \in \textsf{Aliases}(x, \psi)$} \\
\Rightarrow \ &  (s, h[s(x) \mapsto \bot]) \models \psi'  * y \not\mapsto 
\end{align*}

Next, we prove $\forall\sigma'. \sigma' \in \normalfont
\text{\normalfont{WPO}}\llbracket \psi , \texttt{free($x$)} , \textit{ok}\rrbracket  \Longleftarrow 
\normalfont \sigma' \models 
\textsf{wpo}_\text{sh}(\psi , \texttt{free($x$)} , \textit{ok})$.
We show that if $(s', h') \models \psi'  * y\not\mapsto$, then
$\exists (s,h) \models \psi , s(x) \in \textsf{dom}_+ (h) . 
(s', h') = (s, h[s(x) \mapsto \bot])$ holds.
\begin{align*}
    & (s', h') \models \psi' *  y\not\mapsto \\
\Rightarrow \ & (s', h' [s'(x) \mapsto s'(t) ]) \models \psi' *  y \mapsto t \\
& \text{// Since $y \in \textsf{Aliases}(x, \psi)$} \\
\Leftrightarrow \ & (s', h' [s'(x) \mapsto s'(t) ]) \models \psi  \\
\Rightarrow \ & (s_\psi , h_\psi ) \models \psi  \\
& \text{// Let us say that $s_\psi = s'$ and  
$h_\psi = h'[s'(x) \mapsto s'(t)]  = h'[s_\psi (x) \mapsto s_\psi(t)]$}  \\
\Rightarrow \ & \exists (s,h) \models \psi , s(x) \in \textsf{dom}_+ (h) . 
(s', h') = (s, h[s(x) \mapsto \bot]) \\
&\text{// By $h_\psi (s_\psi (x)) =  
h'[s_\psi(x) \mapsto s_\psi(t)] (s_\psi (x)) = s_\psi(t) \in \textsc{Val} $  } \\
&\text{// By $h'(s'(x)) = h'(s_\psi(x)) = \bot $ if $(s', h') \models \psi' * x \approx y * y\not\mapsto$ holds}
\end{align*}

Consequently,
Lemma \ref{lma: expressiveness of atomic command} holds in this case.

\item Otherwise 

By the definition of 
\(\llbracket \texttt{free($x$)} \rrbracket _\textit{ok}\), we have:
\[
\begin{aligned}
\text{\normalfont{WPO}}\llbracket \psi, \texttt{free($x$)}, \textit{ok} \rrbracket 
&= \{ \sigma' \mid \exists \sigma.\ 
\sigma \models \psi \text{ and } 
(\sigma, \sigma') \in \llbracket \texttt{free($x$)} \rrbracket _\textit{ok} \} \\
&= \{ (s, h[s(x) \mapsto \bot]) \mid 
(s, h) \models \psi \text{ and } s(x) \in \textsf{dom}_+(h) \}.
\end{aligned}
\]

Assume \((s,h) \models \psi\).  
If \(s(x) \in \textsf{dom}_+ (h)\), then there exists \(v \in \textsf{Aliases}(x, \psi)\) such that
\(v \mapsto t \in \psi\), since \(\psi\) is canonical.  
This contradicts the assumption that there is no such \(v\) in \(\textsf{Aliases}(x, \psi)\).  
Hence, \((s,h) \models \psi\) implies \(s(x) \notin \textsf{dom}_+ (h)\), and therefore
\[
\text{\normalfont{WPO}}\llbracket \psi, \texttt{free($x$)} ,\textit{ok}\rrbracket = \emptyset.
\]

By the definition of \(\textsf{wpo}_\text{sh}\),  
\(\textsf{wpo}(\psi, \texttt{free($x$)},\textit{ok}) = \textsf{false}\).  
Therefore, this lemma straightforwardly holds in this case.

\end{itemize}

\item $\epsilon = \textit{er}$
\begin{itemize}
    \item If 
      $\displaystyle\forall v \in \textsf{Aliases}(x,\psi).\;(v \mapsto t)\notin\psi$
      \medskip

By the definition of 
    $\llbracket   \texttt{free($x$)}  \rrbracket _\textit{er}
    = \left\{ ((s,h), (s,h)) \mid s(x) \notin \textsf{dom}_+(h) \right\} $,
the following holds:
    \[
    \text{\normalfont{WPO}}\llbracket \psi, \texttt{free($x$)} ,\textit{er}\rrbracket  =
    \left\{ (s,h) \mid (s,h) \models \psi \land s(x) \notin \textsf{dom}_+(h) \right\}
    = \left\{ (s,h) \mid (s,h) \models \psi  \right\} .
    \]

Furthermore, 
        \[
          \mathsf{wpo}_\text{sh}(\psi,\texttt{free}(x),\textit{er})=\psi.
        \]
        Hence,  the  WPO set coincides with the  weakest
        postcondition, and the lemma holds for the error branch.

    \item Otherwise

    If there exists $v \in \textsf{Aliases}(x,\psi)$ such that $(v \mapsto t) \in \psi$,
    then for every $(s,h) \models \psi$, $s(x) \in \textsf{dom}_+ (h)$ and $s(x) \neq \textit{null}$.
    Hence, $(s,h) \notin \llbracket   \texttt{free($x$)}  \rrbracket _\textit{er}$.

    Thus,
    \[
    \text{\normalfont{WPO}}\llbracket \psi, \texttt{free($x$)} ,\textit{er}\rrbracket  = \emptyset,
    \quad
    \textsf{wpo}_\text{sh}(\psi, \texttt{free($x$)},\textit{er}) = \textsf{false}.
    \]
    Therefore,
    Lemma \ref{lma: expressiveness of atomic command} also holds in this case.

\end{itemize}
\end{itemize}

\subsubsection{Case $ x := [y] $}

\begin{itemize}
\item $\epsilon = \textit{ok}$ 
\begin{itemize}
    \item If $z \in \textsf{Aliases}(y, \psi)$ and 
$\psi = \psi'  * z \mapsto t$

Firstly, we prove $\forall\sigma'. \sigma' \in \normalfont
\text{\normalfont{WPO}}\llbracket \psi ,  x := [y] , \textit{ok}\rrbracket  \Longrightarrow 
\normalfont \sigma' \models 
\textsf{wpo}_\text{sh}(\psi ,  x := [y] , \textit{ok})$.
By the definition of $\llbracket    x := [y]  \rrbracket _\textit{ok}$,
$\text{\normalfont{WPO}}\llbracket \psi,  x := [y] ,\textit{ok}\rrbracket  = \{ \sigma' \mid \exists \sigma. 
\sigma \models \psi \land
(\sigma,\sigma') \in \llbracket    x := [y]  \rrbracket _\textit{ok}
\} = 
\{ (s [x \mapsto h(s(y))] , h  ) \mid (s,h) \models \psi \land h(s(y))  \in \textsc{Val}  \}
= 
\{ (s [x \mapsto h(s(y))] , h  ) \mid (s,h) \models \psi \land h(s(y)) = h(s(z)) = s(t) \in \textsc{Val}  \}
= 
\{ (s [x \mapsto s(t)] , h  ) \mid (s,h) \models \psi \land h(s(y)) = h(s(z)) = s(t)\in \textsc{Val}  \}$ holds.
Consequently, $\sigma' \in \text{\normalfont{WPO}}\llbracket \psi,  x := [y] ,\textit{ok}\rrbracket  \iff \exists (s,h) \models \psi , h(s(y)) = s(t)  \in \textsc{Val} .
\sigma' = (s [x \mapsto s(t)] , h  )$ holds.

Additionally, by the definition of $\textsf{wpo}_\text{sh}$,
$\textsf{wpo}_\text{sh}(\psi, x := [y] ,\textit{ok}) = 
\exists x'. (\psi' *   z \mapsto t )   [x := x'] * 
x \approx t [x := x'] $ holds. 
It is enough to show that if $  (s,h) \models \psi \land
h(s(y)) = s(t)  \in \textsc{Val}$, then $(s [x \mapsto s(t)] , h  ) 
\models  \textsf{wpo}_\text{sh} (\psi, x := [y] ,\textit{ok}) $ holds.
To prove our statement, we show the following.

\begin{align*}
& (s,h) \models \psi \land
h(s(y)) = s(t)  \in \textsc{Val} \\
\Rightarrow \ & 
  (s  , h  ) \models 
\psi' *  z \mapsto t   \\
\Leftrightarrow \ & 
  (s  , h  ) \models 
\psi' *  z \mapsto t * t \approx t  \\
\Leftrightarrow \ & 
  (s  , h  ) \models 
(\psi' *   z \mapsto t) [x := x'] [x' := x]
* (t [x' := x] \approx t [x := x'] [x' := x])  \\
& \text{// Since $x'$ is fresh, we have $s(t[x' := x]) = s(t)$} \\
\Leftrightarrow \ & 
 (s  , h  ) \models 
( (\psi' *   z \mapsto t )   [x := x'] * 
t \approx t [x := x'] ) [x' := x]  \\
\Leftrightarrow \ & 
 (s [x' \mapsto s(x)] , h  ) \models 
(\psi' *   z \mapsto t )   [x := x'] * 
t \approx t [x := x'] \\
 & \text{// Lemma \ref{lma: Substitution for assignment}} \\
\Rightarrow \ & 
\exists w \in \textsc{Val} . (s [x' \mapsto w] , h  ) \models 
(\psi' *   z \mapsto t )   [x := x'] * 
t \approx t [x := x']\\
\Leftrightarrow \ & 
(s  , h  ) \models 
\exists x'. (\psi' *   z \mapsto t )   [x := x'] * 
t \approx t [x := x']\\
\Leftrightarrow \ & 
 (s  , h  ) \models 
(\exists x'. (\psi' *   z \mapsto t )   [x := x'] * 
x \approx t [x := x']  ) [x:=t] \\
\Leftrightarrow \ & 
(s [x \mapsto s(t)] , h  ) \models 
\exists x'. (\psi' *   z \mapsto t )   [x := x'] * 
x \approx t [x := x'] \\
 & \text{// Lemma \ref{lma: Substitution for assignment}}  \\
 \Leftrightarrow \ & 
(s [x \mapsto s(t)] , h  ) \models 
\exists x'. \psi   [x := x'] * 
x \approx (t [x := x']) \\
\end{align*}

Next, we prove $\forall\sigma'. \sigma' \in \normalfont
\text{\normalfont{WPO}}\llbracket \psi ,  x := [y]  , \textit{ok}\rrbracket  \Longleftarrow 
\normalfont \sigma' \models 
\textsf{wpo}_\text{sh}(\psi ,  x := [y]  , \textit{ok})$.
We show that if $(s', h') \models \textsf{wpo}_\text{sh}(\psi ,  x := [y]  , \textit{ok})$, then
$\exists (s,h) \models \psi ,  h(s(y)) = s(t)  \in \textsc{Val}  . 
(s', h') = ( s [x \mapsto s(t)] , h)  $ holds.

\begin{align*}
& (s', h') \models 
\exists x'. (\psi' *   z \mapsto t )   [x := x'] * 
x \approx t [x := x']\\
\Leftrightarrow \ & \exists w \in \textsc{Val} . 
(s' [x' \mapsto w] , h') \models (\psi' *   z \mapsto t )   [x := x'] * 
x \approx t [x := x'] \\
\Rightarrow \ & \exists w \in \textsc{Val} . 
(s' [x' \mapsto w] , h') \models (\psi' *   z \mapsto t )   [x := x'] 
\text{ and }
s' [x' \mapsto w] (x) = s' [x' \mapsto w] ( t [x := x'] ) \\
\Leftrightarrow \ & \exists w \in \textsc{Val} . 
(s' [x' \mapsto w] [x \mapsto w] , h') \models \psi' *   z \mapsto t 
\text{ and }
s' [x' \mapsto w] (x) = s' [x' \mapsto w][x \mapsto w] (t  ) \\
&\text{// Lemma \ref{lma: Substitution for assignment} and $s'[x' \mapsto w](x') = w$} \\
\Leftrightarrow \ & \exists w \in \textsc{Val} . 
(s' [x \mapsto w] , h') \models \psi' *  z \mapsto t 
\text{ and }
s'  (x) = s' [x \mapsto w] (t  ) \\
&\text{// $x'$ is fresh} \\
\Rightarrow \ & \exists w \in \textsc{Val} . 
(s_\psi , h_\psi ) \models \psi' *   z \mapsto t 
\text{ and }
s'  (x) = s_\psi ( t  ) \\
&\text{// Let us say that $s_\psi = s' [x \mapsto w]$ and $h_\psi = h'$}\\
\Rightarrow \ & 
\exists (s,h) \models \psi ,  h(s(y)) = s(t)  \in \textsc{Val}  . 
(s', h') = ( s [x \mapsto s(t) ] , h) 
\end{align*}

Consequently,
Lemma \ref{lma: expressiveness of atomic command} holds in this case.

    \item Otherwise

Similar to the case of \texttt{free($x$)}.

\end{itemize}

\item $\epsilon = \textit{er}$
\begin{itemize}
    \item If $\displaystyle\forall v \in \textsf{Aliases}(y,\psi).\;(v \mapsto t)\notin\psi$

Similar to the case of \texttt{free($x$)}.

    \item Otherwise

Similar to the case of \texttt{free($x$)}.

\end{itemize}

\end{itemize}


\subsubsection{Case $[x] := t$  }

\begin{itemize}
\item $\epsilon = \textit{ok}$
\begin{itemize}
    \item If $z \in \textsf{Aliases}(x, \psi)$ and 
$\psi = \psi' * z \mapsto t'$

Firstly, we prove $\forall\sigma'. \sigma' \in \normalfont
\text{\normalfont{WPO}}\llbracket \psi ,  [x] := t , \textit{ok}\rrbracket  \Longrightarrow 
\normalfont \sigma' \models 
\textsf{wpo}_\text{sh}(\psi ,  [x] := t , \textit{ok})$.
By the definition of $\llbracket    [x] := t  \rrbracket _\textit{ok}$,
$\text{\normalfont{WPO}}\llbracket \psi, [x] := t ,\textit{ok}\rrbracket  = \{ \sigma' \mid \exists \sigma. 
\sigma \models \psi \land
(\sigma,\sigma') \in \llbracket    [x] := t  \rrbracket _\textit{ok}
\} = 
\{ (s  , h [s(x) \mapsto s(t)] ) \mid 
(s,h) \models \psi \land h(s(x))  \in \textsc{Val}  \}$ holds.
Consequently, $\sigma' \in \text{\normalfont{WPO}}\llbracket \psi,  [x] := t ,\textit{ok}\rrbracket  \iff \exists (s,h) \models \psi , s(x) \in \textsf{dom}_+ (h) .
\sigma' = (s  , h [s(x) \mapsto s(t)] )$ holds.

Additionally, by the definition of $\textsf{wpo}_\text{sh}$,
$\textsf{wpo}_\text{sh}(\psi, [x] := t ,\textit{ok}) = 
\psi'  * z \mapsto t $ holds. 
It is enough to show that if $  (s,h) \models \psi \land h(s(x))  \in \textsc{Val}$, then $(s  , h [s(x) \mapsto s(t)] ) \models  \textsf{wpo}_\text{sh} (\psi, [x] := t,\textit{ok}) $ holds.
To prove our statement, we show the following.

\begin{align*}
& (s  , h  ) \models \psi \land h(s(x))  \in \textsc{Val} \\
\Rightarrow \ &
(s  , h  ) \models \psi  \\
\Leftrightarrow \ &
(s  , h [s(x) \mapsto s(t')] ) \models \psi \\
& \text{// Since $h(s(x)) = h[s(x) \mapsto s(t')](s(x)) = s(t')$ when $(s, h) \models \psi$} \\
\Leftrightarrow \ &
(s  , h [s(x) \mapsto s(t')] ) \models  \psi' * z \mapsto t' \\
\Rightarrow \ &
(s  , h [s(x) \mapsto s(t)] ) \models  \psi' * z \mapsto t \\
\end{align*}

Next, we prove $\forall\sigma'. \sigma' \in \normalfont
\text{\normalfont{WPO}}\llbracket \psi ,  [x] := t  , \textit{ok}\rrbracket  \Longleftarrow 
\normalfont \sigma' \models 
\textsf{wpo}_\text{sh}(\psi , [x] := t   , \textit{ok})$.
We show that if $(s', h') \models \textsf{wpo}_\text{sh}(\psi ,  [x] := t  , \textit{ok})$, then
$\exists (s,h) \models \psi ,  h(s(x))  \in \textsc{Val}   . 
(s', h') =  (s  , h [s(x) \mapsto s(t)] )  $ holds.
\begin{align*}
& (s', h') \models  \psi' *  z \mapsto t \\
\Rightarrow \ & 
(s', h' [ s'(x) \mapsto s'(t') ] ) \models  \psi' *  z \mapsto t' \\
\Leftrightarrow \ & 
(s', h' [ s'(x) \mapsto s'(t') ] ) \models  \psi \\
\Rightarrow \ & 
(s_\psi , h_\psi ) \models  \psi \\
& \text{// Let us say that $s_\psi = s'$ and 
$ h_\psi  = h' [s'(x) \mapsto s'(t')] = h' [s_\psi (x) \mapsto s_\psi(t') ]  $} \\
\Rightarrow \ & 
\exists (s,h) \models \psi ,  h(s(x))  \in \textsc{Val}   . 
(s', h') =  (s  , h [s(x) \mapsto s(t)] )  \\
&\text{// By $h_\psi (s_\psi (x)) = h' [s'(x) \mapsto s'(t')] (s' (x)) = s'(t')
\in \textsc{Val} $} \\
&\text{// Since $h' (s' (x)) = s'(t) = h'(s_\psi (x)) = s_\psi (t) $ when 
$(s', h') \models  \psi' *  z \mapsto t $} \\
\end{align*}

Consequently,
Lemma \ref{lma: expressiveness of atomic command} holds in this case.

    \item Otherwise

Similar to the case of \texttt{free($x$)}.

\end{itemize}

\item $\epsilon = \textit{er}$
\begin{itemize}
    \item If $\displaystyle\forall v \in \textsf{Aliases}(x,\psi).\;(v \mapsto t')\notin\psi$

Similar to the case of \texttt{free($x$)}.

    \item Otherwise

Similar to the case of \texttt{free($x$)}.

\end{itemize}

\end{itemize}

\subsubsection{Case {${\normalfont{\texttt{local $x$ in $\mathbb{C}$}}}$ }}

By the definition, 
$\text{\normalfont{WPO}}\llbracket \psi , \texttt{local $x $ in $\mathbb{C}$}, \epsilon\rrbracket  \\
=\{ (s', h') \mid \exists (s , h ), \exists v,v' \in \textsc{Val} . 
(s , h ) \models \psi \land 
( (s[x \mapsto v] ,h) , (s' [x \mapsto v'] ,h') ) \in 
\llbracket  \mathbb{C}  \rrbracket _\epsilon 
 \land  s(x) = s'(x) \}  $ and 
$\textsf{wpo}_\text{sh} ( \psi , \texttt{local $x $ in $\mathbb{C}$} ,\epsilon)$ is equivalent to 
$  \exists x'' . 
\textsf{wpo} (  \psi [x:=x'] ,  \mathbb{C}  ,\epsilon)
[x:=x''] [x' := x] $.
In this proof, we use an induction hypothesis (IH), assuming that  Lemma \ref{lma: expressiveness of atomic command} for $\mathbb{C}$ holds.
By the definition of $\normalfont\text{WPO}$ and $\textsf{wpo}_\text{sh}$, we have the following.
\begin{align*}
&(s', h')  \in \text{\normalfont{WPO}}\llbracket \psi , \texttt{local $x$ in $\mathbb{C}$}, \epsilon\rrbracket  \\
  \Leftrightarrow \  & \exists (s , h ), \exists v, v'  . 
(s , h ) \models \psi \land  s(x) = s'(x) 
 \land  ( (s[x \mapsto v] ,h) , (s' [x \mapsto v'] ,h') ) \in 
\llbracket  \mathbb{C}  \rrbracket _\epsilon 
\end{align*}

\begin{align*}
    & (s' , h') \models \textsf{wpo}_\text{sh} ( \psi , \texttt{local $x $ in $\mathbb{C}$} ,\epsilon) \\
\Leftrightarrow \ & (s' , h') \models
\exists x'' . 
( \textsf{wpo} (  \psi [x:=x'] ,  \mathbb{C}  ,\epsilon)
[x:=x''] [x' := x] ) \\
\Leftrightarrow \ & \exists v' \in \textsc{Val} . 
(s' [x'' \mapsto v'] , h')  \models
 \textsf{wpo} (  \psi [x:=x'] ,  \mathbb{C}  ,\epsilon)
[x:=x''] [x' := x] \\
\Leftrightarrow \ & \exists v' \in \textsc{Val} . 
(s' [x'' \mapsto v'] [x' \mapsto s'(x)] [x \mapsto v'] , h')  \models
 \textsf{wpo} (  \psi [x:=x'] ,  \mathbb{C}  ,\epsilon)  \\
 & \text{// By Lemma \ref{lma: Substitution for assignment} (Note that 
 $s' [x'' \mapsto v'] [x' \mapsto s'(x)] (x'') = v '$)} \\
  \Leftrightarrow \ & \exists v'  . 
(s' [x'' \mapsto v'] [x' \mapsto s'(x)] [x \mapsto v']  , h')  \in
  \text{WPO} 
 \llbracket   \psi [x:=x']  ,  \mathbb{C}  ,\epsilon \rrbracket   
 \ \ \text{ // IH and Lemma \ref{lma: sensei's one shot lemma}}\\
   \Leftrightarrow \ &
    \exists (s,h) , \exists v' . (s,h) \models \psi [x:=x']  
   \land ( (s,h) , (s' [x'' \mapsto v'] [x' \mapsto s'(x)] [x \mapsto v']  , h')  ) \in \llbracket  \mathbb{C} \rrbracket _\epsilon  \\
   \Leftrightarrow \ &
   \exists (s,h)  , \exists v' . (s [x \mapsto s(x')] ,h) \models \psi  \land ( (s,h) , 
   (s' [x'' \mapsto v'] [x' \mapsto s'(x)] [x \mapsto v']  , h')  ) \in \llbracket  \mathbb{C} \rrbracket _\epsilon  
\end{align*}
Firstly, we show $(s', h')  \in \text{\normalfont{WPO}}\llbracket \psi , \texttt{local $x  $ in $\mathbb{C}$}, \epsilon\rrbracket   \Longrightarrow
(s' , h') \models \textsf{wpo}_\text{sh} ( \psi , \texttt{local $x $ in $\mathbb{C}$} ,\epsilon)$ holds.
Let us say that $s_1 =  s [x \mapsto v] [x' \mapsto s(x)] [x'' \mapsto v'] $ and
$h_1 =  h  $.
\begin{align*}
& (s', h')  \in \text{\normalfont{WPO}}\llbracket \psi , \texttt{local $x $ in $\mathbb{C}$}, \epsilon\rrbracket  \\
 \Leftrightarrow \ & \exists (s , h ), \exists v , v' . 
(s , h ) \models \psi \land  s(x) = s'(x) 
 \land  ( (s[x \mapsto v] ,h) , (s' [x \mapsto v'] ,h') ) \in 
\llbracket  \mathbb{C}  \rrbracket _\epsilon  \\
 \Rightarrow \ & \exists (s , h ), \exists v  , v' . 
(s_1 [x \mapsto s_1(x')] , h ) \models \psi \land  s(x) = s'(x) 
 \land  ( (s[x \mapsto v] ,h) , (s' [x \mapsto v'] ,h') ) \in 
\llbracket  \mathbb{C}  \rrbracket _\epsilon  \\
&\text{// $ s_1 [x \mapsto s_1(x')] =
s [x\mapsto v] [x' \mapsto s(x)] [x'' \mapsto v'] [x \mapsto s(x)] $, note that 
$s_1(x')= s(x)$}\\
&\text{// Since $x'$ and $x''$ are fresh, the following holds.}\\ 
&\text{// $s,h \models \psi \iff s[x\mapsto s(x)],h \models \psi \iff
s_1 [x \mapsto s_1(x')],h \models \psi $}\\
 \Rightarrow \ & \exists (s , h ), \exists v , v'  . 
(s_1 [x \mapsto s_1(x')] , h ) \models \psi \land  s(x) = s'(x)  \\
 & \land  ( (s[x \mapsto v] [x' \mapsto s(x)] [x'' \mapsto v'] ,h_1) ,
 (s' [x \mapsto v']  [x' \mapsto s(x)] [x'' \mapsto v']  , h')  ) \in 
\llbracket  \mathbb{C}  \rrbracket _\epsilon  \\
&\text{// By Lemma \ref{lma: quantified variables do not affect to operational semantics}}\\
 \Rightarrow \ & \exists (s , h ), \exists v , v'  . 
(s_1 [x \mapsto s_1(x')] , h ) \models \psi \land  s(x) = s'(x)  \\
 & \land  ( (s_1 ,h_1) ,
 (s' [x'' \mapsto v' ] [x' \mapsto s'(x)] [x \mapsto v' ]  , h')  ) \in 
\llbracket  \mathbb{C}  \rrbracket _\epsilon  \\
&\text{// Using $s(x) = s'(x)$} \\
 \Rightarrow \ & \exists (s , h ), \exists v , v'  . 
(s_1 [x \mapsto s_1(x')] , h ) \models \psi     \land  ( (s_1 ,h_1) ,
 (s' [x'' \mapsto v' ] [x' \mapsto s'(x)] [x \mapsto v' ]  , h')  ) \in 
\llbracket  \mathbb{C}  \rrbracket _\epsilon  \\
\Rightarrow \ &
    \exists (s,h) ,\exists v'  . (s [x \mapsto s(x')] ,h) \models \psi  \land ( (s,h) , 
   (s' [x'' \mapsto v'] [x' \mapsto s'(x)] [x \mapsto v']  , h')  ) \in \llbracket  \mathbb{C} \rrbracket _\epsilon  \\
\Leftrightarrow \ & (s' , h') \models \textsf{wpo}_\text{sh} ( \psi , \texttt{local $x$ in $\mathbb{C}$} ,\epsilon) 
\end{align*}

Secondly, we show $(s', h')  \in \text{\normalfont{WPO}}\llbracket \psi , \texttt{local $x $ in $\mathbb{C}$}, \epsilon\rrbracket   \Longleftarrow
(s' , h') \models \textsf{wpo}_\text{sh} ( \psi , \texttt{local $x$ in $\mathbb{C}$} ,\epsilon)$ holds.
Let us say that $s_2 =  s [x \mapsto s(x')] [x' \mapsto s'(x')] [x'' \mapsto s'(x'')] $ and
$h_2 =  h  $.

\begin{align*}
&
(s' , h') \models \textsf{wpo}_\text{sh} ( \psi , \texttt{local $x$ in $\mathbb{C}$} ,\epsilon) \\
\Leftrightarrow \ &
   \exists (s,h) , \exists v' . (s [x \mapsto s(x')] ,h) \models \psi 
   \land ( (s,h) , 
   (s' [x'' \mapsto v'] [x' \mapsto s'(x)] [x \mapsto v']  , h')  ) \in \llbracket  \mathbb{C} \rrbracket _\epsilon  \\
\Rightarrow \ &
   \exists (s,h) , \exists v' . (s_2 ,h_2) \models \psi \land ( (s,h) , 
   (s' [x'' \mapsto v' ] [x' \mapsto s'(x)] [x \mapsto v' ]  , h')  ) \in \llbracket  \mathbb{C} \rrbracket _\epsilon  \\
   & \text{// Since $x'$ and $x''$ are fresh} \\
\Rightarrow \ &
   \exists (s,h) , \exists v' . (s_2 ,h_2) \models \psi \land
   s_2(x) = s'(x)
   \land ( (s,h) , 
   (s' [x'' \mapsto v'] [x' \mapsto s'(x)] [x \mapsto v']  , h')  ) \in \llbracket  \mathbb{C} \rrbracket _\epsilon  \\
&\text{// $s_2(x) = s(x')$ holds since, by Lemma \ref{lma: quantified variables have same store value}, $s(x') = s' [x'' \mapsto v'] [x' \mapsto s'(x)] [x \mapsto v'] (x') = s'(x)$ holds} \\
\Rightarrow \ &
   \exists (s,h) , \exists v' . (s_2 ,h_2) \models \psi \land  s_2(x) = s'(x) \\
&   \land ( (s [x' \mapsto s'(x')] [x'' \mapsto s'(x'')]  ,h) , 
   (s' [x'' \mapsto v'] [x' \mapsto s'(x)] [x \mapsto v'] 
   [x' \mapsto s'(x')] [x'' \mapsto s'(x'')], h')  ) \in \llbracket  \mathbb{C} \rrbracket _\epsilon  \\
& \text{// Lemma \ref{lma: quantified variables do not affect to operational semantics}} \\
\Rightarrow \ &
   \exists (s,h) , \exists v' . (s_2 ,h_2) \models \psi  \land s_2(x) = s'(x) \\
&   \land ( (s [x' \mapsto s'(x')] [x'' \mapsto s'(x'')]  ,h) , 
   (s'  [x \mapsto v'] [x' \mapsto s'(x')] [x'' \mapsto s'(x'')], h')  ) \in \llbracket  \mathbb{C} \rrbracket _\epsilon  \\
\Rightarrow \ &
   \exists (s,h) , \exists v , v' . (s_2 ,h_2) \models \psi \land
     s_2(x) = s'(x)   \land ( (s_2[x \mapsto v]   ,h) , 
   (s'  [x \mapsto v'] , h')  ) \in \llbracket  \mathbb{C} \rrbracket _\epsilon  \\
&\text{// $ \exists v . s_2[x \mapsto v] = 
\exists v . s [x \mapsto s(x')] [x' \mapsto s'(x')] [x'' \mapsto s'(x'')] [x \mapsto v] = 
s [x' \mapsto s'(x')] [x'' \mapsto s'(x'')]  $} \\
\Rightarrow \ & \exists (s_2 , h_2 ), \exists v, v'  . 
(s_2 , h_2 ) \models \psi \land  s_2(x) = s'(x) 
 \land  ( (s_2[x \mapsto v] ,h_2) , (s' [x \mapsto v'] ,h') ) \in 
\llbracket  \mathbb{C}  \rrbracket _\epsilon 
\end{align*}

Therefore, Lemma \ref{lma: expressiveness of atomic command} holds in this case.

\subsubsection{Case $\mathbb{C}_1 ; \mathbb{C}_2$ } 

As an induction hypothesis (IH), Lemma \ref{lma: expressiveness of atomic command} holds in $\mathbb{C}_1$ and $\mathbb{C}_2$.
\begin{itemize}
\item $\epsilon = \textit{ok}$ 

By the definition, 
$\text{\normalfont{WPO}}\llbracket \psi , \mathbb{C}_1 ; \mathbb{C}_2, \textit{ok}\rrbracket  = \{ \sigma' \mid \exists \sigma. 
\sigma \models \psi\land
(\sigma, \sigma') \in \llbracket  \mathbb{C}_1 ; \mathbb{C}_2 \rrbracket _\textit{ok} \} = 
\{ \sigma' \mid \exists \sigma,  \sigma''. \sigma \models \psi\land 
(\sigma, \sigma'') \in \llbracket  \mathbb{C}_1 \rrbracket _\textit{ok} \land 
(\sigma'', \sigma') \in \llbracket  \mathbb{C}_2 \rrbracket _\textit{ok} 
\}$
and
$\textsf{wpo}_\text{sh} (\psi , \mathbb{C}_1 ; \mathbb{C}_2, \textit{ok}) = 
\textsf{wpo}(\textsf{wpo}_\text{sh} (\psi , \mathbb{C}_1, \textit{ok}), \mathbb{C}_2, \textit{ok})$
hold. 
Note that we pick an $\epsilon = \textit{ok}$ case in $\llbracket  \mathbb{C}_1 ; \mathbb{C}_2 \rrbracket _\textit{ok}$ and we denote $\sigma''$ as $(s'', h'')$.

We can use 
$\forall \sigma''. \sigma'' \in \{ \sigma'' \mid \exists \sigma.
\sigma \models \psi\land
(\sigma, \sigma'') \in \llbracket  \mathbb{C}_1 \rrbracket _\textit{ok} \} \iff
\sigma'' \models \textsf{wpo}_\text{sh} (\psi , \mathbb{C}_1, \textit{ok})$
by the induction hypothesis (IH) with respect to $\mathbb{C}_1$.
Consequently, it is enough to show that 
$\forall \sigma'. \sigma' \in 
\text{\normalfont{WPO}}\llbracket \textsf{wpo}_\text{sh} (\psi, \mathbb{C}_1, \textit{ok})
, \mathbb{C}_2, \textit{ok}\rrbracket  =
\{ \sigma' \mid  \exists \sigma''. \sigma'' \models  \textsf{wpo}_\text{sh} (\psi, \mathbb{C}_1, \textit{ok}) \land
(\sigma'', \sigma') \in \llbracket  \mathbb{C}_2 \rrbracket _\textit{ok}
\} \iff \sigma' \models \textsf{wpo}(\textsf{wpo}_\text{sh} (\psi, \mathbb{C}_1, \textit{ok}), \mathbb{C}_2, \textit{ok})$ holds.
Suppose that $\textsf{cano}(\textsf{wpo}_\text{sh}(\psi, \mathbb{C}_1, \textit{ok})
, \mathbb{C}_2) = \exists \overrightarrow{x_i} . \bigvee_{i \in I} \phi_i $, then
$\textsf{wpo} ( \textsf{wpo}_\text{sh}(\psi, \mathbb{C}_1, \textit{ok}) , 
\mathbb{C}_2, \textit{ok}) =
 \bigvee_{i \in I} \exists \overrightarrow{x_i} .\textsf{wpo}_\text{sh} ( \phi_i, \mathbb{C}_2, \textit{ok})  $ and $\text{\normalfont{WPO}}\llbracket \textsf{wpo}_\text{sh} (\psi, \mathbb{C}_1, \textit{ok})
, \mathbb{C}_2, \textit{ok}\rrbracket  = 
\bigcup_{i \in I} \text{\normalfont{WPO}}\llbracket  \exists \overrightarrow{x_i}  . \phi_i,\mathbb{C}_2 , \textit{ok} \rrbracket $ holds by Lemma \ref{lma: from bigvee to bigcup}.
Then, the statement also holds  by IH w.r.t. $\mathbb{C}_2$ and Lemma \ref{lma: sensei's one shot lemma}.
(Note that this procedure is similar to our proof of Proposition \ref{prop: expressiveness of WPO calculus}.)
Consequently, 
Lemma \ref{lma: expressiveness of atomic command} also holds in this case.

\item $\epsilon = \textit{er}$

By the definition of $\textsf{wpo}_\text{sh}$,  
$\textsf{wpo}_\text{sh}(\psi , \mathbb{C}_1;\mathbb{C}_2, \textit{er}) = 
\textsf{wpo}_\text{sh}(\psi , \mathbb{C}_1, \textit{er}) \lor
\textsf{wpo}(\textsf{wpo}_\text{sh}(\psi , \mathbb{C}_1, \textit{ok}), \mathbb{C}_2, \textit{er})$ holds.
Additionally,
by the definition of $\llbracket  \mathbb{C}_1 ; \mathbb{C}_2 \rrbracket _\textit{er}$ and IH w.r.t. $\mathbb{C}_1$, 
$\text{\normalfont{WPO}}\llbracket \psi , \mathbb{C}_1 ; \mathbb{C}_2, \textit{er}\rrbracket  = 
\{ \sigma' \mid \exists \sigma,  \sigma''. \sigma \models \psi\land 
( (\sigma, \sigma') \in  \llbracket  \mathbb{C}_1 \rrbracket _\textit{er}  \lor
(\sigma, \sigma'') \in \llbracket  \mathbb{C}_1 \rrbracket _\textit{ok} \land 
(\sigma'', \sigma') \in \llbracket  \mathbb{C}_2 \rrbracket _\textit{er} )  \}= 
\{ \sigma' \mid \exists \sigma,  \sigma''. 
(\sigma \models \psi \land 
(\sigma, \sigma') \in  \llbracket  \mathbb{C}_1 \rrbracket _\textit{er}  ) \lor
(\sigma \models \psi \land (\sigma, \sigma'') \in \llbracket  \mathbb{C}_1 \rrbracket _\textit{ok} \land 
(\sigma'', \sigma') \in \llbracket  \mathbb{C}_2 \rrbracket _\textit{er} )  \}
= \text{\normalfont{WPO}}\llbracket \psi , \mathbb{C}_1 , \textit{er}\rrbracket  
\cup \text{\normalfont{WPO}}\llbracket \textsf{wpo}_\text{sh} (\psi, \mathbb{C}_1, \textit{ok})
, \mathbb{C}_2, \textit{er}\rrbracket  
$ holds.

IH w.r.t. $\mathbb{C}_1$ confirms the equivalence of following.
\[
\forall \sigma'.
 \sigma' \in \text{\normalfont{WPO}}\llbracket \psi , \mathbb{C}_1 , \textit{er}\rrbracket  
\iff \sigma' \models \textsf{wpo}_\text{sh} (\psi , \mathbb{C}_1, \textit{er}),
\] 
Furthermore, the IH w.r.t. $\mathbb{C}_2$ confirms the following in a manner similar to the case of $\epsilon = \textit{ok}$ before.
\[
\forall \sigma'.
\sigma' \in \text{\normalfont{WPO}}\llbracket \textsf{wpo}_\text{sh} (\psi, \mathbb{C}_1, \textit{ok})
, \mathbb{C}_2, \textit{er}\rrbracket  \iff
\sigma' \models 
\textsf{wpo}(\textsf{wpo}_\text{sh}(\psi , \mathbb{C}_1, \textit{ok}), \mathbb{C}_2, \textit{er})
\]
Consequently, Lemma \ref{lma: expressiveness of atomic command} holds in this case.

\end{itemize}

\subsubsection{Case $\mathbb{C}_1 + \mathbb{C}_2$} 
\begin{itemize}
    \item In this case, we do not make a division between the \textit{ok} and \textit{er} cases. 
    As an induction hypothesis (IH), let us say that  Lemma \ref{lma: expressiveness of atomic command} holds in $\mathbb{C}_1$ and $\mathbb{C}_2$.

    The goal is to prove $\forall \sigma'. \sigma' \in \text{\normalfont{WPO}}
    \llbracket \psi , \mathbb{C}_1 + \mathbb{C}_2 ,\epsilon \rrbracket  \iff \sigma' \models \textsf{wpo}_\text{sh} (\psi, \mathbb{C}_1 + \mathbb{C}_2 ,\epsilon)$. We can generate the following equivalent to the above goal by the definition of $\text{\normalfont{WPO}}\llbracket \psi , \mathbb{C}_1 + \mathbb{C}_2 ,\epsilon\rrbracket $ and $\textsf{wpo}_\text{sh} (\psi, \mathbb{C}_1 + \mathbb{C}_2 , \epsilon)$.
    
    \[
    \begin{aligned}
        \forall \sigma'. &\sigma' \in \{ \sigma' \mid \exists\sigma. \sigma \models \psi\land (\sigma, \sigma') \in \llbracket \mathbb{C}_1 \rrbracket _\epsilon \} \text{ or} \\
        &\sigma' \in \{ \sigma' \mid \exists\sigma. \sigma \models \psi\land (\sigma, \sigma') \in \llbracket \mathbb{C}_2 \rrbracket _\epsilon \} \\
        &\iff  \sigma' \models \textsf{wpo}_\text{sh} ( \psi, \mathbb{C}_1  , \epsilon) \text{ or }
        \sigma' \models \textsf{wpo}_\text{sh} ( \psi, \mathbb{C}_2  , \epsilon)
    \end{aligned}
    \]
    
    The above statement can be proved by the IH for each $\mathbb{C}_1$ and $\mathbb{C}_2$, consequently, Lemma \ref{lma: expressiveness of atomic command} holds in this case.
\end{itemize}

\subsubsection{Case $\mathbb{C}^\star$} 

Before proving this case, we establish the case of finite iteration, i.e., $\mathbb{C}^m$ where $m \geq 1$ and $m \in \mathbb{N}$.
\[
\textsf{wpo}_\text{sh} (\psi , \mathbb{C}^m, \textit{ok}) = \Upsilon(m) \text{ and } 
\textsf{wpo}_\text{sh} (\psi , \mathbb{C}^m, \textit{er}) = \bigvee_{i=0}^{m-1} \textsf{wpo}(\Upsilon(i), \mathbb{C}, \textit{er})
\]
Here, $\Upsilon(0) = \psi $, and for $n \geq 1$, $\Upsilon(n+1) = \textsf{wpo}(\Upsilon(n), \mathbb{C}, \textit{ok})$.

As an induction hypothesis (IH), let us say that Lemma \ref{lma: expressiveness of atomic command} holds in $\mathbb{C}$. 
Then, by the IH and the fact that Lemma \ref{lma: expressiveness of atomic command} holds in the \textbf{Case} $\mathbb{C}_1 ; \mathbb{C}_2$, 
Lemma \ref{lma: expressiveness of atomic command} extends to finite iteration. This is because $\mathbb{C}^m = \overbrace{\mathbb{C} ; \ldots ; \mathbb{C}}^\text{for $m$ times}$ holds.

Now, we prove the case of $\mathbb{C}^\star$.
By the IH and the fact that Lemma \ref{lma: expressiveness of atomic command} holds in the case $\mathbb{C}_1 + \mathbb{C}_2$ and the case of $\mathbb{C}^m$, 
    Lemma \ref{lma: expressiveness of atomic command} extends to the $\mathbb{C}^\star$ case.
    Since by Definition \ref{def: Denotational semantics of ISL}, $\mathbb{C}^\star = \bigplus_{m \in \mathbb{N}} C^m$ holds.

\section{Proof of Lemma \ref{thm: for all p, c, epsilon, we have wpo}}\label{appendix: proof of thm for all we have wpo}
We prove this lemma by induction on $\mathbb{C}$.
In this proof, we omit cases where the return value of $\textsf{wpo}_\text{sh}$, given input $\psi$, is either $\textsf{false}$ or $\psi$.
When it returns $\textsf{false}$, the triple is trivially valid since the postcondition is unsatisfiable.
When it returns $\psi$, we can directly apply the corresponding proof rules, so the lemma holds in this case as well.

\hfill

\noindent{\textbf{Case} $\normalfont \mathbb{C} = x := \texttt{alloc()}$ 
and $\epsilon = \textit{ok}$} 

Here, $\psi = \left(\bigast_{i=1}^n y_i \not\mapsto \right) * \psi'$,
where $\psi'$ contains no atoms of the form $\not\mapsto$.
For space reasons, we abbreviate
$A = \bigvee_{j=1}^n \exists x'. \left(
\left(\bigast_{i=1}^n y_i \not\mapsto \right)[y_j \not\mapsto := y_j \mapsto - ] *
x \approx y_j * \psi'[x := x'] \right)$.

\[
\inferrule*[right=\textsc{Disj}]
{
{\inferrule*[right=\textsc{Alloc1}]
{ \ }
{[\psi]  \ x := \texttt{alloc()} \ 
[\textit{ok} : \exists x'.  
(  \psi[x := x'] * x \mapsto - )   ] }}
\ \ \ 
{\inferrule*[right=\textsc{Disj}]
{ {\normalfont\text{Please refer the below}} }
{ [\psi]  \ x := \texttt{alloc()} \ 
[\textit{ok} : A  ] } } 
}
{ [\psi]  \ x := \texttt{alloc()} \ 
[\textit{ok} : \exists x'.  
(  \psi[x := x'] * x \mapsto - ) \lor A  ] } 
\]

For space reasons, we abbreviate $B = \left( \mathop{\bigast}\limits_{\substack{i = 1 \\ i \ne j}}^{n} y_i \not\mapsto \right) * \psi'$.
The remainder of the proof is given below.

\[\scriptsize
{\inferrule*[right=\textsc{Disj}]
{ 
\inferrule*[right=\textsc{Alloc2}]
{ \ }
{\inferrule*[right=\textsc{Cons}]
{ 
[y_j \not\mapsto * B ]  \ x := \texttt{alloc()} \ 
[\textit{ok} : 
  \exists x'.   ( x \mapsto - * x \approx y_j * B [x := x'] )  ]  
  \ \ \ \text{(for any $1 \leq j \leq n$)} 
  }
{\inferrule*[right=\textsc{Cons}]
{ 
[\psi]  \ x := \texttt{alloc()} \ 
[\textit{ok} : 
  \exists x'.   ( y_j \mapsto - * x \approx y_j *
\left(\bigast_{\substack{i=1 \\ i \neq j }}^n y_i \not\mapsto \right) *
 \psi' [x := x'] )  ] 
 \ \ \ \text{(for any $1 \leq j \leq n$)} 
 }
{
[\psi]  \ x := \texttt{alloc()} \ 
[\textit{ok} : 
  \exists x'.   ( 
\left(\bigast_{i=1}^n y_i \not\mapsto \right) [y_j \not\mapsto := y_j \mapsto - ] *
x \approx y_j * \psi' [x := x'] )  ] 
\ \ \ \text{(for any $1 \leq j \leq n$)} }
}}
}
{ [\psi]  \ x := \texttt{alloc()} \ 
[\textit{ok} : 
 \bigvee_{j=1}^n \exists x'.   ( 
\left(\bigast_{i=1}^n y_i \not\mapsto \right) [y_j \not\mapsto := y_j \mapsto - ] *
x \approx y_j * \psi' [x := x'] )  ] } } 
\]

\noindent{\textbf{Case} $\normalfont \mathbb{C} = x := \texttt{alloc()}$ 
and $\epsilon = \textit{er}$}

This case is straightforward.

\noindent{\textbf{Case} $\normalfont \mathbb{C} = \mathbb{C}_1 ; \mathbb{C}_2$ and $\epsilon = \textit{er}$}

\[
\inferrule*[right=\scriptsize\textsc{Disj}]{
\inferrule*[right=\scriptsize\textsc{Seq1}]{
{\inferrule*[]
{}
{%
\stackrel{\text{(Induction Hypothesis)}}%
{[\psi]  \ \mathbb{C}_1  \ 
[\textit{er} :
\textsf{wpo}_\text{sh}(\psi ,\mathbb{C}_1, \textit{er}) ]}
}
}
}
{[\psi]  \ \mathbb{C}_1 ; \mathbb{C}_2 \ 
[\textit{er} :
\textsf{wpo}_\text{sh}(\psi ,\mathbb{C}_1, \textit{er}) ]} \ \ \ 
\inferrule*[right=\scriptsize\textsc{Seq2}]
{  \text{Please refer the below}
}
{[\psi]  \ \mathbb{C}_1 ; \mathbb{C}_2 \ 
[\textit{er} :
\textsf{wpo}(\textsf{wpo}_\text{sh} (\psi,\mathbb{C}_1, \textit{ok}), \mathbb{C}_2,\textit{er})]}
}
{\inferrule*[right=\scriptsize\normalfont\text{By the definition}]
{ [\psi]  \ \mathbb{C}_1 ; \mathbb{C}_2 \ 
[\textit{er} :
\textsf{wpo}_\text{sh}(\psi ,\mathbb{C}_1, \textit{er})  \lor
\textsf{wpo} (\textsf{wpo}_\text{sh}(\psi,\mathbb{C}_1, \textit{ok}), \mathbb{C}_2,\textit{er})] }
{ [\psi]  \ \mathbb{C}_1 ; \mathbb{C}_2 \ 
[\textit{er} : \textsf{wpo}_\text{sh}(\psi, \mathbb{C}_1 ; \mathbb{C}_2 , \textit{er} )] }} 
\]

For space reasons, we provide the remainder of the proof below.

\[
\inferrule*[right=\scriptsize\textsc{Seq2}]
{
\inferrule*[]
{}
{
\stackrel{\text{(Induction Hypothesis)}}%
{[\psi]  \ \mathbb{C}_1  \ 
[\textit{ok} : \textsf{wpo}_\text{sh}(\psi, \mathbb{C}_1 , \textit{ok})]}
}
\and
\inferrule*[]
{}
{
\stackrel{\text{(Induction Hypothesis and Lemma \ref{lma: for all derivation if and then})}}%
{[ \textsf{wpo}_\text{sh} (\psi, \mathbb{C}_1 , \textit{ok}) ]  \ \mathbb{C}_2 \ 
[ \textit{er} : \textsf{wpo} (\textsf{wpo}_\text{sh}(\psi, \mathbb{C}_1  , \textit{ok} ) 
, \mathbb{C}_2 , \textit{er} ) ]}
}
}
{[\psi]  \ \mathbb{C}_1 ; \mathbb{C}_2 \ 
[\textit{er} :
\textsf{wpo}(\textsf{wpo}_\text{sh} (\psi,\mathbb{C}_1, \textit{ok}), \mathbb{C}_2,\textit{er})]}
\]

\noindent{\textbf{Case} $\normalfont \mathbb{C} = \mathbb{C}_1 ; \mathbb{C}_2$ and $\epsilon = \textit{ok}$}

\[
\inferrule*[right=\textsc{Seq2}]
{  
\inferrule*[]
{}
{
\stackrel{\text{(Induction Hypothesis)}}%
{[\psi]  \ \mathbb{C}_1  \ 
[\textit{ok} : \textsf{wpo}_\text{sh}(\psi, \mathbb{C}_1 , \textit{ok})] }
}
\and
\inferrule*[]
{}
{
\stackrel{\text{(Induction Hypothesis and Lemma \ref{lma: for all derivation if and then})}}%
{[ \textsf{wpo}_\text{sh} (\psi, \mathbb{C}_1 , \textit{ok}) ]  \ \mathbb{C}_2 \ 
[ \textit{ok} :  \textsf{wpo} (\textsf{wpo}_\text{sh}(\psi, \mathbb{C}_1  , \textit{ok} ) 
, \mathbb{C}_2 , \textit{ok} ) ] }
}
}
{\inferrule*[right=\normalfont\text{By the definition}]
{ [ \psi]  \ \mathbb{C}_1 ; \mathbb{C}_2 \ 
[\textit{ok} :  \textsf{wpo}(\textsf{wpo}_\text{sh}(\psi, \mathbb{C}_1  , \textit{ok} ) 
, \mathbb{C}_2 , \textit{ok} )] }
{ [  \psi]  \ \mathbb{C}_1 ; \mathbb{C}_2 \ 
[\textit{ok} :  \textsf{wpo}_\text{sh} (\psi, \mathbb{C}_1 ; \mathbb{C}_2 , \textit{ok} )] }} 
\]

\noindent{\textbf{Case} $\normalfont \mathbb{C} = \mathbb{C}_1 + \mathbb{C}_2$ }

\[\inferrule*[right=\textsc{Disj}]
{
\inferrule*[left=Choice]
{\inferrule*[]
{}
{
\stackrel{\text{(Induction Hypothesis)}}%
{[\psi]  \ \mathbb{C}_1  \ 
[\epsilon : \textsf{wpo}_\text{sh} (\psi, \mathbb{C}_1  , \epsilon ) ]}
}
}
{[\psi]  \ \mathbb{C}_1 + \mathbb{C}_2 \ 
[\epsilon : \textsf{wpo}_\text{sh} (\psi, \mathbb{C}_1  , \epsilon ) ]} 
\and
\inferrule*[right=Choice]
{\inferrule*[]
{}
{
\stackrel{\text{(Induction Hypothesis)}}%
{[\psi]  \ \mathbb{C}_2  \ 
[\epsilon : \textsf{wpo}_\text{sh} (\psi, \mathbb{C}_2  , \epsilon ) ]}
}
}
{[\psi]  \ \mathbb{C}_1 + \mathbb{C}_2 \ 
[\epsilon : \textsf{wpo}_\text{sh} (\psi, \mathbb{C}_2  , \epsilon ) ]}  }
{\inferrule*[right=\normalfont\text{By the definition}]
{ [\psi]  \ \mathbb{C}_1 + \mathbb{C}_2 \ 
[\epsilon : \textsf{wpo}_\text{sh} (\psi, \mathbb{C}_1  , \epsilon ) \lor
\textsf{wpo}_\text{sh} (\psi,   \mathbb{C}_2 , \epsilon )] }
{ [\psi]  \ \mathbb{C}_1 + \mathbb{C}_2 \ 
[\epsilon : \textsf{wpo}_\text{sh} (\psi, \mathbb{C}_1 + \mathbb{C}_2 , \epsilon )] }} 
\]

\noindent{\textbf{Case} $\normalfont \mathbb{C} = \texttt{local $x  $ in $\mathbb{C}$}$}

Please note that 
$\textsf{wpo}_\text{sh} ( \psi , \texttt{local $x$ in $\mathbb{C}$} ,\epsilon) = \bigvee_{j \in J} \exists x'' , \overrightarrow{x_j}   . \varphi_j
= \exists x'' . 
( \textsf{wpo} (  \psi [x:=x'] ,  \mathbb{C}  ,\epsilon)
[x:=x''] [x' := x] )$
(where $
\textsf{wpo} (  \psi [x:=x'] ,  \mathbb{C}  ,\epsilon)
[x:=x''] [x' := x] = \bigvee_{j \in J} \exists \overrightarrow{x_j} . \varphi_j  $).

\[
\inferrule*[right=\scriptsize\normalfont{\textsc{Cons }(Renaming $x$ with $x''$)}]
{  
\inferrule*[right=\textsc{Local}]
{  \inferrule*[]
{} 
{
\stackrel{\text{(Induction Hypothesis and Lemma \ref{lma: for all derivation if and then})}}%
{[\psi ]  \ \mathbb{C} \ 
[\epsilon :  
 \textsf{wpo} (  \psi  ,  \mathbb{C}  ,\epsilon)  ] \ \ \ 
 x \notin \textsf{fv}(\psi)}
}
 }
{[\psi ]  \ \texttt{local $x $ in $\mathbb{C}$} \ 
[\epsilon :  \exists x . 
 \textsf{wpo} (  \psi   ,  \mathbb{C}  ,\epsilon)  ] \ \ \ 
 x \notin \textsf{fv}(\psi)} }
{\inferrule*[right=\normalfont\text{By the definition}]
{[\psi ]  \ \texttt{local $x $ in $\mathbb{C}$} \ 
[\epsilon :  \exists x'' . 
( \textsf{wpo} (  \psi [x:=x']  ,  \mathbb{C}  ,\epsilon)
[x:=x''] [x' := x] )  ]}
{[ \psi]  \ \texttt{local $x$ in $\mathbb{C}$} \ [\epsilon : 
\textsf{wpo}_\text{sh} (\psi, \texttt{local $x$ in $\mathbb{C}$}, \epsilon) ] }}
\]

\noindent{\textbf{Case} $\normalfont \mathbb{C} = \mathbb{C}^\star$ and $\epsilon = \textit{ok}$}

To prove this case, we employ a version of the \textsc{Backwards Variant}, originally outlined in the IL paper \cite{o2019incorrectness}, with slight modifications.
\[
\inferrule[\textsc{Backwards Variant} {\normalfont\text{(where $n \in \mathbb{N}$ is fresh)}}]
{  \forall n \in \mathbb{N} .  [P(n) ] \ \mathbb{C} \ [\textit{ok}:P(n+1)]  }
{  [P(0) ] \ \mathbb{C}^\star \ [\textit{ok}:\bigvee_{n \in \mathbb{N}} P(n)] }
\]
We demonstrate that $\textsc{Backwards Variant}$ is derivable using our ISL proof rules.
At first, we show that $[P(0) ] \ \mathbb{C}^\star \ [\textit{ok}:P(n)] $ is derivable by induction on $n$.
Here, $P(n)$ denotes the loop variant and $P(0)=P$.

As a base case ($n = 0$), we show the following.
\[
\inferrule*[right=\textsc{Loop zero}]
{ \ }
{ [P(0)] \ \mathbb{C}^\star \ [\textit{ok}: P(n)] }
\]

Next, we assume that $[P(0)] \ \mathbb{C}^\star \ [\textit{ok}: P(n)]$ is derivable as an induction hypothesis. Additionally, we utilize the premise $[P(n)] \ \mathbb{C} \ [\textit{ok}: P(n+1)]$, which is a premise of the \textsc{Backwards Variant}.
\[
\inferrule*[right=\textsc{Loop non-zero}]
{\inferrule*[right=\textsc{Seq2}]
{ [P(0)] \ \mathbb{C}^\star \ [\textit{ok}: P(n)] \  \  \  [P(n)] \ \mathbb{C} \ [\textit{ok}: P(n+1)]}
{ [P(0)] \ \mathbb{C}^\star ; \mathbb{C} \ [\textit{ok}: P(n+1)]}}
{ [P(0)] \ \mathbb{C}^\star \ [\textit{ok}: P(n+1)] }
\]
By above, we show that $\forall n \in \mathbb{N} . [P(0)] \ \mathbb{C}^\star \ [\textit{ok}: P(n)]$ is derivable.
Then, we will finally prove that \textsc{Backwards Variant} is derivable in our system.

\[
\inferrule*[right=\textsc{Cons}]
{\inferrule*[right=\textsc{Disj}]
{ \forall n \in \mathbb{N}  . [P(0)] \ \mathbb{C}^\star \ [\textit{ok}: P(n)]}
{ [\bigvee_{n \in \mathbb{N}} P(0)] \ \mathbb{C}^\star \ [\textit{ok}: \bigvee_{n \in \mathbb{N}} P(n)] }}
{ [P(0)] \ \mathbb{C}^\star \ [\textit{ok}: \bigvee_{n \in \mathbb{N}} P(n)] }
\]

Now, we prove Lemma \ref{thm: for all p, c, epsilon, we have wpo} in this case using \textsc{Backwards Variant}.
(In this context, IH denotes the induction hypothesis.)

\[
\inferrule*[right=\textsc{Backwards Variant}]
{
{\inferrule*[]
{}
{\inferrule*[right=\normalfont\text{By the definition}]
{
\stackrel{\text{(Induction Hypothesis and Lemma \ref{lma: for all derivation if and then})}}%
{\forall  n \in \mathbb{N}. [ \Upsilon(n) ]  \ \mathbb{C} \ 
[\textit{ok} :  \textsf{wpo} ( \Upsilon(n) , \mathbb{C} , \textit{ok} ) ]}
}
{\forall  n \in \mathbb{N}. [ \Upsilon(n) ]  \ \mathbb{C} \ 
[\textit{ok} :  \Upsilon(n+1) ] } }
}}
{\inferrule*[right=\normalfont\text{By the definition}]
{[ \Upsilon(0) ]  \ \mathbb{C}^\star  \ 
[\textit{ok} : \bigvee_{n \in \mathbb{N}} \Upsilon(n) ] }
{ [\psi ]  \ \mathbb{C}^\star  \ 
[ \textit{ok} : \textsf{wpo}_\text{sh} (\psi , \mathbb{C}^\star , \textit{ok})] }
}
\]

\noindent{\textbf{Case} $\normalfont \mathbb{C} = \mathbb{C}^\star$ and $\epsilon = \textit{er}$}

For space reasons, we abbreviate \textsc{BV} for \textsc{Backwards Variant}.

\[ 
\inferrule* [right=Seq2]
{  
    {\inferrule* [right=BV]
    {
    {
\stackrel{\text{(Induction Hypothesis and Lemma \ref{lma: for all derivation if and then})}}%
{\forall  n \in \mathbb{N}. [ \Upsilon(n) ]  \ \mathbb{C} \ 
[\textit{ok} :  \textsf{wpo} ( \Upsilon(n) , \mathbb{C} , \textit{ok} ) ]}
}
    }
    {
    [\psi ]  \ \mathbb{C}^\star  \ 
    [ \textit{ok} : \textsf{wpo}_\text{sh} (\psi , \mathbb{C}^\star , \textit{ok})] 
    } }
    \and
  \inferrule*[]
   {}
   {\inferrule* [right=Disj]
   {\inferrule* []
   {}
   {
   \stackrel{\text{(Induction Hypothesis and Lemma \ref{lma: for all derivation if and then})}}%
   {\forall n \in \mathbb{N}  . [\Upsilon (n)] \ \mathbb{C} \ 
    [\textit{er}:\textsf{wpo} (\Upsilon (n), \mathbb{C}, \textit{er})]}
    }
    }
   {[\bigvee_{n \in \mathbb{N}  } \Upsilon (n)] \ \mathbb{C} \ 
    [\textit{er}:\bigvee_{n \in \mathbb{N}    } \textsf{wpo} (\Upsilon (n), \mathbb{C}, \textit{er})]}}
    }
{\inferrule* [right=Loop non-zero]
   {[\psi ] \ \mathbb{C}^\star;\mathbb{C} \ 
    [\textit{er}:\bigvee_{n \in \mathbb{N}   } \textsf{wpo} (\Upsilon (n), \mathbb{C}, \textit{er})]} 
    { [\psi ] \ \mathbb{C}^\star \ 
    [\textit{er}:\bigvee_{n \in \mathbb{N}   } \textsf{wpo} (\Upsilon (n), \mathbb{C}, \textit{er})] } }
\]

\end{document}